\documentclass[12pt]{article}
\usepackage{amsmath}
\usepackage{graphicx}
\usepackage{natbib}
\usepackage{url} 
\usepackage[shortlabels]{enumitem}
\usepackage{latexsym, epsfig, amssymb, amsthm, mathrsfs, bbm, enumitem,mathabx}
\usepackage[OT1]{fontenc}
\usepackage{xcolor}
\usepackage{indentfirst}
\usepackage[colorlinks,citecolor=blue,linkcolor=blue,urlcolor=blue]{hyperref}
\usepackage{multirow,multicol,booktabs}
\usepackage{anysize}
\renewcommand{\baselinestretch}{1}
\usepackage[left=0.95in, right=0.95in, top=1in, bottom=1in]{geometry}%

\makeatletter

\usepackage{algorithm}
\usepackage{algorithmic}
\usepackage{tikz}
\usepackage{float}
\usepackage{booktabs}
\usepackage{adjustbox}

\newcommand{\indep}{\perp \!\!\! \perp}

\newcommand{\argmin}{\ensuremath{\operatornamewithlimits{arg min}}}

\newtheorem{lemma}{Lemma}
\newtheorem{proposition}{Proposition}
\newtheorem{theorem}{Theorem}
\newtheorem{definition}{Definition}
\newtheorem{remark}{Remark}
\newtheorem{example}{Example}
\newtheorem{condition}{Condition} 
\renewcommand{\theequation}{
	\arabic{equation}%
}

\newcommand{\ignore}[1]{}{}

\newcommand{\bSigma}{{\bf \Sigma}}

\newcommand{\bu}{\mbox{\bf u}}

\newcommand{\bw}{\mbox{\bf w}}
\newcommand{\bx}{\mbox{\bf x}}
\newcommand{\bsS}{\mbox{\bf\scriptsize S}}
\newcommand{\bsX}{\mbox{\bf\scriptsize X}}
\newcommand{\bsx}{\mbox{\bf\scriptsize x}}

\newcommand{\bs}{\mbox{\bf s}}

\newcommand{\bS}{\mbox{\bf S}}

\newcommand{\bW}{\mbox{\bf W}}
\newcommand{\bX}{\mbox{\bf X}}

\newcommand{\bY}{\mbox{\bf Y}}

\newcommand{\bzero}{\mbox{\bf 0}}
\newcommand{\bveps}{\mbox{\boldmath $\varepsilon$}}

\newcommand{\balpha}{\mbox{\boldmath $\alpha$}}

\newcommand{\eP}{\mathbb P}
\newcommand{\eE}{\mathbb E}
\newcommand{\eR}{\mathbb R}

\newcommand{\cG}{\mathcal G}
\newcommand{\cD}{\mathcal D}
\newcommand{\cF}{\mathcal F}

\newcommand{\cX}{\mathcal X}
\newcommand{\cS}{\mathcal S}
\newcommand{\cR}{\mathcal R}
\newcommand{\cN}{\mathcal N}

\newcommand{\cE}{\mathcal E}
\newcommand{\cA}{\mathcal A}

\newcommand{\var}{\mathrm{var}}

\newcommand{\corr}{\mathrm{corr}}

\def\T{{ \mathrm{\scriptscriptstyle T} }}
\def\CLT{{ \mathrm{\scriptscriptstyle CLT} }}
\def\BF{{ \mathrm{\scriptscriptstyle BF} }}
\def\ASY{{ \mathrm{\scriptscriptstyle ASY} }}
\def\EB{{ \mathrm{\scriptscriptstyle EB} }}

\def\independenT#1#2{\mathrel{\setbox0\hbox{$#1#2$}%
		\copy0\kern-\wd0\mkern4mu\box0}}

\renewcommand{\hat}{\widehat}

\usepackage{xcolor}
\usepackage[draft,inline,nomargin,index]{fixme}
\fxsetup{theme=color,mode=multiuser}
\FXRegisterAuthor{yf}{ayf}{\color{magenta}}

\usepackage{ulem}

\def\spacingset#1{\renewcommand{\baselinestretch}%
	{#1}\small\normalsize} \spacingset{1}

\begin{document}

\title{SLOACI: Surrogate-Leveraged Online Adaptive Causal Inference%
\thanks{
Yingying Fan is Centennial Chair in Business Administration and Professor, Data Sciences and Operations Department, Marshall School of Business, University of Southern California, Los Angeles, CA 90089 (E-mail: \textit{fanyingy@marshall.usc.edu}). 
Zihan Wang is Ph.D. Candidate, Department of Statistics and Data Science, Tsinghua University, China (E-mail: \textit{wangzh21@mails.tsinghua.edu.cn}). %
Waverly Wei is an Assistant Professor, Data Sciences and Operations Department, Marshall School of Business, University of Southern California, Los Angeles, CA 90089 (E-mail: \textit{waverly@marshall.usc.edu})
}
\date{}
\author{
Yingying Fan$^1$, Zihan Wang$^2$ and Waverly Wei$^1$
\medskip\\
University of Southern California$^1$ and Tsinghua University$^2$
\\
} %
}

\maketitle

\maketitle

\spacingset{1.4}
\begin{abstract}	

Adaptive experimental designs have gained increasing attention across a range of domains. In this paper, we propose a new methodological framework, surrogate-leveraged online adaptive causal inference (SLOACI), which integrates predictive surrogate outcomes into adaptive designs to enhance efficiency. For downstream analysis, we construct the adaptive augmented inverse probability weighting estimator for the average treatment effect using collected data. Our procedure remains robust even when surrogates are noisy or weak. We provide a comprehensive theoretical foundation for SLOACI. Under the asymptotic regime, we show that the proposed estimator attains the semiparametric efficiency bound. From a non-asymptotic perspective, we derive a regret bound to provide practical insights. We also develop a toolbox of sequential testing procedures that accommodates both asymptotic and non-asymptotic regimes, allowing experimenters to choose the perspective that best aligns with their practical needs. Extensive simulations and a synthetic case study are conducted to showcase the superior finite-sample performance of our method. 
\end{abstract}



 \clearpage
\spacingset{1.77}

\section{Introduction}
\label{sec.intro}

 Adaptive randomized experiments, also known as adaptive designs, have gained increasing attention across a range of domains, including field experiments, online A/B testing, and clinical trials \citep{zhao2025adaptive}.
In adaptive designs, units arrive over time, and a data-dependent mechanism governs aspects of the experiment—such as treatment assignment, enrollment proportions, sampling intensity, or follow-up decisions, with the aim of improving participant outcomes and collecting information efficiently for downstream inference. The key advantage of adaptive experiments lies in their ability to dynamically adjust treatment allocations based on data accumulated during the course of the study.   In particular, a popular class of designs that revise treatment probabilities in response to observed outcomes from earlier participants is known as the response-adaptive randomization (RAR) designs \citep{hu2006theory}. This iterative adaptation of RAR designs often favors treatment arms that demonstrate more informative results, thereby maximizing the information gain from each participant and improving the overall efficiency of the experiment, leading to an effective way to collect data for downstream estimation and inference. 


In contemporary applications, in addition to the primary outcomes from the adaptive experiments, low cost (possibly) noisy side information that is predictive of the primary outcomes may also be available.   Such side information can be used to construct ``surrogate" outcomes  that are predictive of the primary outcomes at low cost but likely of low quality too. How to fuse such side information into the adaptive design to further improve its efficiency is the main purpose of our study. We provide two examples to illustrate when/how surrogate outcomes can be collected/constructed for our purpose. 

\begin{example}
Large language models (LLMs) make it possible to construct informative surrogates from unstructured text. Consider an online learning platform wanting to evaluate the effectiveness of a learning intervention (treatment). After completing a lesson (either treated or not),  each participant submits a short written response $T_t$ describing their understanding at Stage (time) $t$. The primary outcome  is the quiz score $Y_t(z)$ corresponding to assigned treatment $z\in \{0,1\}$, but the free-text response $T_t$ also contains useful information about knowledge mastery. Together with a low-dimensional covariate vector $\mathbf X_t \in \mathcal X$ (e.g., past academic performance, demographics) that contains well-understood confounders, one may fine-tune or prompt an LLM to extract predictive signals from the text \citep{veitch2020adapting} using external data or data-splitting, yielding
the LLM-based predictive models $
\big(f_{0}(\cdot), f_{1}(\cdot)\big): \mathcal{T}\times \mathcal X  \to \mathbb{R}^{2}
$ that approximate the potential primary outcomes. At Stage $t$, upon observing $(T_t, \mathbf X_t)$, the experimenter generates a pair of surrogates
$
\big(S_t(0), S_t(1)\big) = \big(f_{0}(T_t, \mathbf X_t),\, f_{1}(T_t, \mathbf X_t)\big).$
During the causal modeling stage, only confounder $\mathbf X_t$ is used in modeling the primary potential outcomes $(Y_t(0), Y_t(1))$, while the surrogates serve as auxiliary signals containing predictive information about $(Y_t(0), Y_t(1))$ beyond what is explained by $\mathbf X_t$. 
\end{example}

\begin{example}
The recent advances in data acquisition allow us to collect a wide array of inexpensive auxiliary measurements, though some are of lower predictive quality for the outcome of interest.
For example, in the case study in Section \ref{sec:case}, nine covariates are recorded, and the goal is to study whether small class size improves test scores. However, for causal identification, the structural models typically include only three confounders $\mathbf X_t$ motivated by prior literature \citep{chetty2011does, athey2025combining}. One may leverage external data (or a held-out sample) to construct predictive models $\big({f}_0(\cdot),{f}_1(\cdot)\big):\widetilde{\cX}\subset \mathbb R^9\to\eR^2$ to map the nine auxiliary variables $\widetilde{\bX}_t\in\widetilde{\cX}$ to predictions of the primary outcomes under two treatments $z\in \{0,1\}$. 
At Stage $t$, one can generate a pair of surrogates $(S_t(0),S_t(1))=(f_0(\widetilde{\bX}_t),f_1(\widetilde{\bX}_t))$. Though the constructed surrogates may be low-quality proxies of primary potential outcomes $(Y_t(0), Y_t(1))$ due to budget concerns,  they can still contain predictive information about $(Y_t(0), Y_t(1))$ beyond what is explained by $\mathbf X_t$. 
See Section~\ref{sec:case} for detailed illustration.
\end{example}

In this paper, 
we propose a novel framework, namely the surrogate-leveraged online adaptive causal inference (SLOACI), to seamlessly integrate predictive surrogate outcomes into adaptive experiments to improve the efficiency of adaptive experiments and the downstream causal inference. 
Unlike existing approaches, our method can accommodate surrogate outcomes of varying predictive quality, rather than assuming they are highly informative of the primary outcomes. This flexibility is critical in real-world settings, where surrogate quality can be inconsistent or only moderately predictive, as illustrated in the two examples above. 
We provide theoretical and empirical characterizations of how the efficiency gain of adaptive designs depends on surrogate quality, offering practical guidance on when and how surrogate-based adaptation is beneficial.


Our study also contributes to the theoretical frontier of adaptive designs.  Existing theoretical results are typically  in the asymptotic setting \citep{zhang2007asymptotic}. However, in real-world adaptive experiments, sample sizes may be moderate or small due to budget concerns, where asymptotic analysis only provides limited theoretical insights or practical guidance. Therefore, studying the non-asymptotic regime is crucial for better finite-sample guarantees and practical guidance \citep{howard2021time,waudby2024estimating}.

We provide a comprehensive theoretical framework studying the sampling properties of SLOACI in both asymptotic and non-asymptotic regimes. \textit{On the one hand}, our  finite-sample theoretical results offer a refined lens of viewing adaptive experiments that provide substantial theoretical insights, which are not fully captured in asymptotic analysis. Concretely, in our regret bound decomposition (see Theorem~\ref{thm:regret}), we provide  theoretical  insights on the convergence rate of the initialization phase, as well as the convergence rates corresponding to the estimated conditional surrogate outcome model, the estimated conditional primary outcome model, and the adaptive treatment allocation, respectively, in a finite-sample setting. Such a decomposition also provides a set of practical insights. \textit{On the other hand}, we are the first to study the regret bound for adaptive ATE estimation under nonparametric regression and provide the explicit characterization of the lower-order terms and constant factors. 
We also incorporate surrogate outcomes explicitly in our non-asymptotic analysis, addressing an important practical scenario that has not been systematically studied in previous adaptive experimentation literature.  We show that incorporating surrogates as ``supplement" instead of ``covariates" is non-trivial, as simply treating surrogates as covariates leads to a higher order regret bound (see Section~\ref{subsubsec:role_surrogate}).

Additionally, we develop a comprehensive toolbox of sequential testing procedures for evaluating the significance of ATE that accommodates both asymptotic and non-asymptotic regimes. Our framework allows experimenters to choose the perspective, whether asymptotic or non-asymptotic, that best aligns with their practical needs and their views on the asymptotic regime. We elaborate on the proposed testing procedures and discuss the trade-offs between the two regimes in Section~\ref{subsec:tradeoff}.

\subsection{Literature review} 

Our framework is closely related to frequentist RAR designs, a line of work that 
has been developed in \cite{hu2006theory,tymofyeyev2007implementing,hu2009efficient}. In parallel, covariate-adjusted RAR designs seek to improve covariate balance across treatment arms by allowing future assignments to depend on both past assignments and subject-specific covariates \citep{
zhang2007asymptotic,hu2015unified,bertsimas2019covariate,
zhao2025pigeonhole}. Our work is more closely aligned with the classical RAR literature.

Our work also connects to the adaptive experimentation literature aimed at improving ATE estimation efficiency. This line of literature includes two-stage designs \citep{hahn2011adaptive}, fully adaptive AIPW-based procedures attaining semiparametric efficiency \citep{kato2025efficient}, and clipped AIPW methods accommodating modern uncertainty quantification \citep{cook2024semiparametric}. These approaches focus primarily on asymptotic properties. 
Recent work has introduced regret-based analyses under design-based and superpopulation settings \citep{dai2023clip,neopane2025logarithmic,neopane2025optimistic}, but finite-sample regret bounds for the AIPW estimator under fully adaptive, nonparametric designs remain an open question. Beyond estimation, prior studies have developed sequential testing methods, including asymptotic confidence sequences based on Lindeberg-type martingales \citep{waudby2024time,cook2024semiparametric,oprescu2025efficient}, as well as non-asymptotic LIL-based and empirical Bernstein confidence sequences grounded in anytime-valid concentration results \citep{waudby2024estimating,cook2024semiparametric,kato2025efficient}.

The existing literature on the use of surrogates in randomized experiments can be broadly divided into two lines of work. The first line utilizes surrogates from a data fusion perspective \citep{athey2025surrogate,imbens2025long,chen2025efficient,athey2025combining}. 
The second line treats surrogates as supplements for the primary outcomes \citep{chen2008improving, anderer2022adaptive, kallus2025role}. We also leverage surrogate information, but in different ways and under weaker assumptions on the surrogate quality.   

The remainder of the paper is organized as follows. Section~\ref{sec:method} introduces the adaptive data collection mechanism and presents the SLOACI framework.
Section~\ref{sec:theory} carries out a non-asymptotic theoretical analysis and derives a regret bound to provide practical insights under finite samples. 
Section~\ref{sec:testing} develops a comprehensive sequential testing framework that complements our proposed design under both asymptotic and non-asymptotic regimes. 
Section \ref{sec:case} evaluates the empirical performance of SLOACI via a synthetic case study. Section~\ref{sec:discuss} discusses future directions.
The asymptotic theoretical results and all technical proofs and simulation studies are relegated to the supplementary material.

\section{Methodology}\label{sec:method}


\subsection{Formulation of adaptive data collection mechanism}\label{subsec:adaptive data collection}

Suppose that experimental units arrive sequentially over $t$ and 
the experimenter observes 
$
  \left\{\bX_t,  Z_t, Y_t, \bS_t\right\}_{t=1}^T,
$
across $T \in \mathbb N$ experimental stages, where $\bX_t\in\cX\subset\eR^d$ is a
$d$-dimensional covariate profile collected at Stage $t$,  $Z_t\in\{0,1\}$ denotes the treatment assignment, and  
$Y_t\in\eR$ denotes the observed primary outcome of interest. The potential outcomes under control ($Z_t=0$) and treatment ($Z_t=1$) are denoted by $Y_t(0)$ and $Y_t(1)$, respectively, following the Neyman-Rubin potential outcome framework \citep{splawa1990application,rubin2005causal}. Then $Y_t=Y_t(1)Z_t+Y_t(0)(1-Z_t)$. As discussed in the Introduction, we consider the scenario where surrogate outcomes, $\bS_t=(S_t(0),S_t(1))^{\T}\in\cS\subset\eR^2$ under the control and the treatment arms, are also available and predictive of the primary outcomes. 

We consider an adaptive experimentation setting where the treatment allocation is adaptively updated at each Stage $t$ based on the experimental data collected up to Stage $t-1$, while our approach can be easily extended to accommodate batch-wise updates. We denote the treatment assignment probability (also known as allocation) as 
$$\pi_t=\eP(Z_t=1|\cF_{t-1}),$$
where 
$\cF_{t-1}=\sigma(\{\bX_j,  Z_j, Y_j, \bS_j\}_{j=1}^{t-1})$ 
is the $\sigma$-field generated by the history up to $t-1$. 



Our parameter of interest is the average treatment effect (ATE), defined as 
$
    \tau_0:=\eE\{Y_t(1)\}-\eE\{Y_t(0)\}.
$ 
Alternatively, we can also represent the ATE as 
\begin{align*}
\tau_0:=\eE_{(\bsX,\bsS)}\{\tau(\bX,\bS)\},
\end{align*}
where $\eE_{(\bsX,\bsS)}(\cdot)$ denotes the expectation over $(\bX,\bS)\sim(\cX,\cS)$, an i.i.d. copy of pairs $\{(\bX_t,\bS_t)\}_{t\ge1}$, and 
$\tau(\bx,\bs):=\mu_1(\bx,\bs)-\mu_0(\bx,\bs)=\eE\{Y_t(1)|\bX_t=\bx,\bS_t=\bs\}-\eE\{Y_t(0)|\bX_t=\bx,\bS_t=\bs\}$ denotes the conditional average treatment effect (CATE). The functions $\mu_0(\bX_t,\bS_t)$ and $\mu_1(\bX_t,\bS_t)$ respectively denote the conditional means of $Y_t(0)$ and $Y_t(1)$ given $(\bX_t,\bS_t)$.

In our setting, we view the surrogate outcome $S_t(z)$ as a supplement to the corresponding primary outcome $Y_t(z)$ for $z\in\{0,1\}$. We assume $\rho_0^2+\rho_1^2>0$, where
\begin{equation}
    \label{eq:surr_cov}
   \rho_0 := \corr\{Y_t(0),S_t(0)|\bX_t\}~\text{ and }~ \rho_1:=\corr\{Y_t(1),S_t(1)|\bX_t\},
\end{equation}
indicating that at least one surrogate is predictive of the primary outcome given the covariates $\bX_t$. 
It is important to note that we do not require $S_t(z)$'s are good proxies of potential outcomes $Y_t(z)$'s; rather, the nonzero conditional correlation is sufficient for our framework. Note also that we consider a pair of counterfactual surrogate outcomes, where the surrogates can be either directly observed or synthetically generated, and are therefore assumed to be always observable regardless of the assigned treatment. In the two examples discussed in the Introduction, such conditions can be satisfied as long as surrogates are linearly predictive of primary outcomes after accounting for $\mathbf X_t$.  


\subsection{The oracle problem}
\label{subsec:oracle}
We aim at finding an adaptive treatment allocation strategy $\pi_t$ that leads to an ATE estimator with optimal variance, achieving the semiparametric efficiency bound. 
For any fixed treatment allocation $\pi$ that does not vary with $t$, according to Theorem 1 of \cite{hahn1998role},  the asymptotic variance of any regular estimator for the ATE based on the resulting collected data is lower bounded by
\begin{equation}
    \label{eq:varsigma_pi}
    \ell_1(\pi):=\frac{\var\{Y_t(1)|\bX_t,S_t(1)\}}{\pi}+\frac{\var\{Y_t(0)|\bX_t,S_t(0)\}}{1-\pi}+\eE_{(\bsX,\bsS)}\{\tau(\bX,\bS)-\tau_0\}^2,
\end{equation}
which corresponds to the semiparametric efficiency bound. 
This naturally motivates us to find $\pi^*$ that minimizes the design objective:
$\argmin_{\pi\in(0,1)}\ell_1(\pi)$. 
Intuitively, the objective function value $\ell_1(\pi^*)$  acts as an optimal benchmark to evaluate the effectiveness of different treatment allocations.
The minimization problem has a closed-form solution:
\begin{equation}
    \label{eq:ney_all}
    \begin{aligned}
        \pi^*
        =\frac{\sqrt{1-\rho_1^2}\sqrt{\var\{Y_t(1)|\bX_t\}}}{\sqrt{1-\rho_1^2}\sqrt{\var\{Y_t(1)|\bX_t\}}+\sqrt{1-\rho_0^2}\sqrt{\var\{Y_t(0)|\bX_t\}}},
    \end{aligned}
\end{equation}
where $\rho_1$ and $\rho_0$ are defined in \eqref{eq:surr_cov}. The oracle treatment allocation $\pi^*$ in \eqref{eq:ney_all} is similar to the classical Neyman allocation in the literature \citep[e.g.,][]{cai2024performance,zhao2025adaptive}; however, in our setting, it depends on the conditional variances of primary outcomes given covariate and the effectiveness of surrogates as measured by $\rho_1$ and $\rho_0$.


\subsection{A warm-up of our method under a linear model}
\label{subsec:linear}

In this section, we present our framework under the simplified setting of  linear models for clear elucidation. We start by introducing additional necessary notations.

\noindent\textit{Notation}. $\eR^d$ and $\eR_{\ge0}$ denote the $d$-dimensional real space and the set of non-negative real numbers, respectively. For a vector $\bx=(x_1,\dots,x_d)$, denote its Euclidean norm and $\ell_{\infty}$-norm by $\|\bx\|=(\sum_{i=1}^{d}x_i^2)^{1/2}$ and $\|\bx\|_{\infty}=\max_{1\le i\le d}|x_i|$, respectively. $I(\cdot)$ denotes the indicator function. For a positive integer $n,$ write $[n]=\{1,\dots,n\}$. For $x,y \in {\mathbb R},$ we use $x \wedge y = \min(x,y)$, and $x\vee y=\max(x,y)$. For two positive sequences $\{a_n\}$ and $\{b_n\}$, we write: (i) $a_n=O(b_n)$ or $b_n=\Omega(a_n)$ if there exists a positive constant $c$ such that $a_n/b_n\le c$; (ii) $a_n=o(b_n)$ if $a_n/b_n\to0$; (iii) $a_n\asymp b_n$ or $a_n=\Theta(b_n)$ if and only if $a_n=O(b_n)$ and $a_n=\Omega(b_n)$ hold simultaneously; and (iv) $a_n=\widetilde{O}(b_n)$ if $a_n=O(b_n\log^k b_n)$ for some $k\ge 0.$

To start, we assume the following joint linear models:
\begin{equation}
    \label{eq:model_linear}
    \begin{pmatrix}
    Y_t(0) \\
    S_t(0)
    \end{pmatrix}
    \Big| \bX_t \sim 
    \begin{pmatrix}
    \balpha_{Y,0}^{\T}\bX_t \\
    \balpha_{S,0}^{\T}\bX_t
    \end{pmatrix}
    +
    \begin{pmatrix}
    \varepsilon_{Y,t}(0) \\
    \varepsilon_{S,t}(0)
    \end{pmatrix},
    \quad
    \begin{pmatrix}
    Y_t(1) \\
    S_t(1)
    \end{pmatrix}
    \Big| \bX_t \sim 
    \begin{pmatrix}
    \balpha_{Y,1}^{\T}\bX_t \\
    \balpha_{S,1}^{\T}\bX_t
    \end{pmatrix}
    +
    \begin{pmatrix}
    \varepsilon_{Y,t}(1) \\
    \varepsilon_{S,t}(1)
    \end{pmatrix},
\end{equation}
where $\balpha_{Y,z},\balpha_{S,z}\in\eR^d$ are unknown coefficients. For $z\in\{0,1\}$, we assume  
\(\eE\{\varepsilon_{Y,t}(z)|\bX_t\} = \eE\{\varepsilon_{S,t}(z)|\bX_t\} = 0\), 
    \(\var\{\varepsilon_{Y,t}(z)|\bX_t\} = \sigma_{Y,z}^2 > 0\) and  
   \( \var\{\varepsilon_{S,t}(z)|\bX_t\} = \sigma_{S,z}^2 > 0.\)

To characterize the relationships between the primary outcomes and the surrogates, we further assume that $\eE\{\varepsilon_{S,t}(0)\varepsilon_{S,t}(1)|\bX_t\}=0$ and 
\begin{align}\label{eq:model_varep}
\varepsilon_{Y,t}(0)|\varepsilon_{S,t}(0)\sim\gamma_0\varepsilon_{S,t}(0)+\varepsilon_t(0),\quad
 \varepsilon_{Y,t}(1)|\varepsilon_{S,t}(1)\sim \gamma_1\varepsilon_{S,t}(1)+\varepsilon_t(1),
\end{align}
where $\gamma_z$ is the linear coefficient, $\eE\{\varepsilon_t(z)|\varepsilon_{S,t}(z)\}=0$, and $\var\{\varepsilon_t(z)|\varepsilon_{S,t}(z)\}=\sigma_z^2>0$ is the conditional variance for  $z\in\{0,1\}$. These model assumptions require that, conditional on the covariate, the unexplained variation in the primary outcomes can be partially captured by that of the surrogates through a linear dependence structure, with remaining exogenous noises \citep[Cf.,][]{mikusheva2025linear}.

We note that models similar to \eqref{eq:model_linear} have been studied in \cite{mccaw2023leveraging,mccaw2024synthetic} in a non-causal offline setting. We provide a comparison 
in Remark \ref{rmk:biv_gau} below. 

\begin{remark}
    \label{rmk:biv_gau}
    While our use of surrogate in Model~\eqref{eq:model_linear} is similar to \cite{mccaw2023leveraging,mccaw2024synthetic}, there are three major differences: (i) \cite{mccaw2023leveraging} propose a surrogate phenotype regression in a non-causal offline setting, where the data have already been collected, and then \cite{mccaw2024synthetic} apply this framework to a population biobank study. By contrast, our Model~\eqref{eq:model_linear} is a counterfactual model under a causal framework, and the estimation procedure follows an online setting. (ii)  \cite{mccaw2023leveraging,mccaw2024synthetic} assume their  outcome variable is missing at random, whereas in our potential outcome framework, the missingness of our potential outcomes $Y_t(0)$ and $Y_t(1)$ depends on the assigned treatment $Z_t$ so may not be random. 
    (iii) \cite{mccaw2023leveraging,mccaw2024synthetic} assume the joint distribution of $\bveps_t(z):=(\varepsilon_{Y,t}(z),\varepsilon_{S,t}(z))^{\T}$ follows a bivariate Gaussian for each $z\in\{0,1\}$, whereas our Model \eqref{eq:model_linear} relaxes this assumption. 
\end{remark}

Based on the model assumptions in \eqref{eq:model_varep}, Model~\eqref{eq:model_linear} can be transformed into
\begin{equation}
    \label{eq:model_linear_cond_XS}
    \begin{aligned}
        Y_t(0)|X_t,S_t(0)\sim&~(\balpha_{Y,0}-\gamma_0\balpha_{S,0})^{\T}\bX_t+\gamma_0S_t(0)+\varepsilon_t(0)=:\mu_0(\bX_t,\bS_t)+\varepsilon_t(0),\\
        Y_t(1)|X_t,S_t(1)\sim&~(\balpha_{Y,1}-\gamma_1\balpha_{S,1})^{\T}\bX_t+\gamma_1S_t(1)+\varepsilon_t(1)=:\mu_1(\bX_t,\bS_t)+\varepsilon_t(1).
    \end{aligned}
\end{equation}
We will be working with the transformed Model \eqref{eq:model_linear_cond_XS} to illustrate our proposed method. Under this model, for a fixed $t$ with the observed data $\{(\bX_j,Z_j,\bS_j,Y_j)\}_{j=1}^{t}$ , we adopt the following procedure to estimate Model~\eqref{eq:model_linear_cond_XS}. For each $z\in\{0,1\}$, we obtain the least squares estimates of the parameters in Model~\eqref{eq:model_linear_cond_XS}, denoted as $\hat{\balpha}_{Y,z,t},\hat{\balpha}_{S,z,t}$ and $\hat{\gamma}_{z,t}$, with their explicit expressions provided in Section~\ref{supsubsec:bound_linear} of the supplementary material.
Then, the conditional mean functions and variances can be estimated as:
\begin{align}
    &\hat{\mu}_{z,t}(\bx,\bs)=(\hat{\balpha}_{Y,z,t}-\hat{\gamma}_{z,t}\hat{\balpha}_{S,z,t})^{\T}\bx+\hat{\gamma}_{z,t}s(z),\label{eq:mu_hat_linear}\\
    &\hat{\sigma}_{z,t}^2=\frac{1}{\sum_{j=1}^{t}I(Z_j=z)}\sum_{j=1}^{t}\{Y_j-\hat{\mu}_{z,t}(\bX_j,\bS_j)\}^2I(Z_j=z).\label{eq:sigma_hat_linear}
\end{align}

Our proposed SLOACI algorithm under Model \eqref{eq:model_linear_cond_XS} works as follows.

\begin{itemize}
    \item[(i)] {\bf Initialization phase}: For $t=1,\dots,T_0$, 
    assign the treatment $Z_t$ to 0 for odd $t$ and to 1 for even $t$. At stage $T_0$, obtain the conditional mean function and variance estimators $\hat{\mu}_{0,T_0}(\bx,\bs),$ $\hat{\mu}_{1,T_0}(\bx,\bs),$ $\hat{\sigma}_{0,T_0}^2$, and $\hat{\sigma}_{1,T_0}^2$ according to \eqref{eq:mu_hat_linear} and \eqref{eq:sigma_hat_linear}. 
    

        \item[(ii)] {\bf Concentration phase}: For $t=T_0+1,\dots,T,$ update the adaptive treatment allocation as $\pi_t = \widetilde\pi_t$ with  $\widetilde\pi_t$ defined in \eqref{eq:init_pi}, assign the treatment $Z_t$ as 1 with probability $\pi_t$ and 0 with probability $1-\pi_t$, and obtain the conditional mean function and variance estimators $\hat{\mu}_{0,t}(\bx,\bs),$ $\hat{\mu}_{1,t}(\bx,\bs),$ $\hat{\sigma}_{0,t}^2$, and $\hat{\sigma}_{1,t}^2$ according to~\eqref{eq:mu_hat_linear} and \eqref{eq:sigma_hat_linear}.
\end{itemize}


The initialization phase uses a fixed treatment assignment to accumulate a non-trivial number of observations under each treatment arm so that a treatment allocation can be formed at Stage $T_0+1$ to start the concentration phase.
Motivated by the oracle treatment allocation in \eqref{eq:ney_all}, this allocation is computed as $\widetilde{\pi}_{T_0+1}$ with
\begin{equation}
    \label{eq:init_pi}
    \widetilde{\pi}_t=\hat{\sigma}_{1,t-1}/(\hat{\sigma}_{1,t-1}+\hat{\sigma}_{0,t-1}). 
\end{equation}
Note that in this phase, we only calculate \eqref{eq:mu_hat_linear}--\eqref{eq:init_pi} once at the end.

With the initial treatment allocation $\widetilde{\pi}_{T_0+1}$ at the beginning of the concentration phase, we calculate \eqref{eq:mu_hat_linear}--\eqref{eq:init_pi} for each $t=T_0+1,\dots,T$ to update the conditional mean function and variance estimators, as well as the treatment allocation. 
Since $\widetilde{\pi}_t$ in \eqref{eq:init_pi} can get close to 0 and 1 at early stages in this phase due to small sample size, instead of using $\widetilde{\pi}_t$ directly, we apply a clipping procedure to $\widetilde{\pi}_t$ in our formal construction of adaptive treatment allocation $\pi_t$, as detailed in \eqref{eq:clip} in Section~\ref{subsec:procedure}. However, we ignore this clipping procedure here to simplify the presentation. 

Although we assume both $S_t(0)$ and $S_t(1)$ observable at each Stage $t$, our proposal remains applicable when only the surrogate associated with the assigned treatment is observable. In such cases, $\hat{\balpha}_{S,z,t}$ is obtained using only the observed surrogate. The corresponding asymptotic theory still holds, with minor differences in finite-sample performance. Similar arguments apply to the nonlinear setting in Section~\ref{subsec:procedure}. 
Additionally, our design can be extended to a multi-stage adaptive setting, where allocations are updated only a small number of times (e.g. two experimental stages). Further details are provided in Section~\ref{supsubsec:batch_adaptive} of the supplementary material. Under a multi-stage setting, the theoretical results established in Section~\ref{sec:theory} can be extended naturally.

After Stage $T$, our final estimate of the ATE utilizing all the data collected is:
\begin{equation}
    \label{eq:proposed_non}
    \begin{aligned}
        \hat{\tau}_T = \frac{1}{T}\sum_{t=1}^T\phi_t,
    \end{aligned}
\end{equation}
where $\phi_t$ is the single-stage estimation for the ATE at Stage $t$ and is defined as
\begin{equation}\label{eq:single-stage-ATE}
 \begin{aligned}
\phi_t: = & \hat{\mu}_{1,t-1}(\bX_t,\bS_t)-\hat{\mu}_{0,t-1}(\bX_t,\bS_t)
 +\pi_t^{-1}I(Z_t=1)\{Y_t-\hat{\mu}_{1,t-1}(\bX_t,\bS_t)\}\\ & -(1-\pi_t)^{-1}I(Z_t=0)\{Y_t-\hat{\mu}_{0,t-1}(\bX_t,\bS_t)\},
\end{aligned}
\end{equation}
and $\pi_t$ is the adaptive treatment allocation at Stage $t$. Based on the estimator $\hat{\tau}_T$ and its sampling properties, we can further conduct the statistical inference for ATE; the details will be made clear in Section~\ref{sec:testing}.

\subsection{SLOACI under nonparametric framework}
\label{subsec:procedure}

We now formalize our  framework in a general nonlinear setting. Concretely, we generalize Model~\eqref{eq:model_linear} to the following nonparametric nonlinear model:
\begin{equation}
\label{eq:model}
    \begin{pmatrix}
    Y_t(0) \\
    S_t(0)
    \end{pmatrix}
    \Big| \bX_t \sim 
    \begin{pmatrix}
    m_{Y,0}(\bX_t) \\
    m_{S,0}(\bX_t)
    \end{pmatrix}
    +
    \begin{pmatrix}
    \varepsilon_{Y,t}(0) \\
    \varepsilon_{S,t}(0)
    \end{pmatrix},
    \quad
    \begin{pmatrix}
    Y_t(1) \\
    S_t(1)
    \end{pmatrix}
    \Big| \bX_t \sim 
    \begin{pmatrix}
    m_{Y,1}(\bX_t) \\
    m_{S,1}(\bX_t)
    \end{pmatrix}
    +
    \begin{pmatrix}
    \varepsilon_{Y,t}(1) \\
    \varepsilon_{S,t}(1)
    \end{pmatrix},
\end{equation}
where $m_{Y,0}(\cdot),m_{S,0}(\cdot),m_{Y,1}(\cdot),m_{S,1}(\cdot)\in\eR^d\to\eR$ are unknown nonlinear functions, and the remaining setups are the same as those in Model \eqref{eq:model_linear}. Similar to Section~\ref{subsec:linear}, we characterize the dependence of $\varepsilon_{Y,t}(0)$ on $\varepsilon_{S,t}(0),$ and of $\varepsilon_{Y,t}(1)$ on $\varepsilon_{S,t}(1)$, by \eqref{eq:model_varep}.
Note that we do not require independence between $\bX_t$ and $(\varepsilon_{Y,t}(z),\varepsilon_{S,t}(z))$.  We also assume that the conditional variances of $\varepsilon_t(0)$ and $\varepsilon_t(1)$ are constant and independent of $\big(\bX_t,S_t(0)\big)$ and $\big(\bX_t,S_t(1)\big)$, respectively. This assumption requires that the covariate and surrogates account for all systematic individual heteroskedasticity, leaving the residual variation constant across individuals \citep{mccaw2023leveraging,mccaw2024synthetic}.
More generally, one may allow for heteroskedasticity by modeling the conditional variance functions as $\var\{\varepsilon_t(0)|\bX_t,S_t(0)\}=\sigma_0^2\big(\bX_t,S_t(0)\big)$ and $\var\{\varepsilon_t(1)|\bX_t,S_t(1)\}=\sigma_1^2\big(\bX_t,S_t(1)\big)$. Then, the proposed design can be extended to a covariate-adjusted surrogate-leveraged RAR design to address the heteroskedasticity. See Section~\ref{supsubsec:heter} of the supplementary material for discussion.

Combining \eqref{eq:model} and \eqref{eq:model_varep}, we have
\begin{equation}
    \label{eq:model_cond_XS}
    \begin{aligned}
        Y_t(0)|X_t,S_t(0)\sim&~m_{Y,0}(\bX_t)-\gamma_0m_{S,0}(\bX_t)+\gamma_0S_t(0)+\varepsilon_t(0)=:\mu_0(\bX_t,\bS_t)+\varepsilon_t(0),\\
        Y_t(1)|X_t,S_t(1)\sim&~m_{Y,1}(\bX_t)-\gamma_1m_{S,1}(\bX_t)+\gamma_1S_t(1)+\varepsilon_t(1)=:\mu_1(\bX_t,\bS_t)+\varepsilon_t(1),
    \end{aligned}
\end{equation}
where 
\( \mu_0(\bX_t,\bS_t)=m_{Y,0}(\bX_t)-\gamma_0m_{S,0}(\bX_t)+\gamma_0S_t(0)\), and
\(\mu_1(\bX_t,\bS_t)=m_{Y,1}(\bX_t)-\gamma_1m_{S,1}(\bX_t)+\gamma_1S_t(1).\) 
Let $m_0(\bx):=m_{Y,0}(\bx)-\gamma_0m_{S,0}(\bx)$ and $m_1(\bx):=m_{Y,1}(\bx)-\gamma_1m_{S,1}(\bx)$. Model~\eqref{eq:model_cond_XS} can be rewritten as:
\begin{align}\label{eq:model_cond_XS_simplified}
Y_t(0)=m_0(\bX_t)+\gamma_0S_t(0)+\varepsilon_t(0),\quad Y_t(1)=m_1(\bX_t)+\gamma_1S_t(1)+\varepsilon_t(1),
\end{align}
which are two partially linear models with $m_0(\bX_t)$ and $m_1(\bX_t)$ the nonparametric parts, $\gamma_0S_t(0)$ and $\gamma_1S_t(1)$ the linear parts, and $\varepsilon_t(0)$ and $\varepsilon_t(1)$ error terms with conditional variances $\var\{\varepsilon_t(0)|\bX_t,S_t(0)\}=\sigma_0^2$ and $\var\{\varepsilon_t(1)|\bX_t,S_t(1)\}=\sigma_1^2$, in respective models. To estimate the partially linear models, we apply the classical residual-based method proposed by \citet{robinson1988root}. Other estimation approaches in the literature include profile least squares estimation \citep{speckman1988kernel}, series estimation \citep{newey1997convergence}, and double/debiased machine learning \citep{chernozhukov2018double}, among others. See Section~\ref{supsubsec:profile} of the supplementary material for detailed extensions. In what follows, we summarize our proposed design under Model \eqref{eq:model_cond_XS_simplified} in Algorithm \ref{alg:main}.

\begin{algorithm}[ht]
\spacingset{1.2}
\caption{SLOACI for nonparametric nonlinear model}
\begin{algorithmic}[1]
\label{alg:main}
\STATE \textbf{Input}: Clipping rate $\eta$, the length of initialization phase $T_0$, the total time horizon $T$, and the kernel function $K_h(\cdot)$.

\STATE \textbf{Set}: $\pi_0 \leftarrow 1/2$, $\hat{\mu}_{0,0}(\bx,\bs)\leftarrow 0$, $\hat{\mu}_{1,0}(\bx,\bs)\leftarrow 0$, $\hat{\sigma}_{0,0}^2\leftarrow 0$, and $\hat{\sigma}_{1,0}^2\leftarrow 0$.

\STATE \textbf{Initialization}: For $t=1,\dots,T_0,$ set $\pi_t \leftarrow 1/2$, $\hat{\mu}_{0,t}(\bx,\bs)\leftarrow 0$, and $\hat{\mu}_{1,t}(\bx,\bs)\leftarrow 0$, assign the treatment $Z_t$ to 0 for odd $t$ and to 1 for even $t$, observe the covariates $\bX_t$, the surrogates $\bS_t=(S_t(0),S_t(1))^{\T}$ and the outcome $Y_t=Y_t(1)Z_t+Y_t(0)(1-Z_t)$, and obtain the conditional mean function and variance estimators $\hat{\mu}_{0,T_0}(\bx,\bs),$ $\hat{\mu}_{1,T_0}(\bx,\bs),$ $\hat{\sigma}_{0,T_0}^2$, and $\hat{\sigma}_{1,T_0}^2$ according to \eqref{eq:m_S_hat}--\eqref{eq:sigma_hat}.

\FOR{$t=T_0+1,\dots,T,$}
    \STATE Calculate the initial allocation as $\widetilde{\pi}_t=\hat{\sigma}_{1,t-1}/(\hat{\sigma}_{1,t-1}+\hat{\sigma}_{0,t-1}).$
    \STATE Set the clipping threshold as $\zeta_t=(1/2)\cdot t^{-\eta}$. 
    \STATE Obtain the adaptive treatment allocation by clipping as $\pi_t={\rm CLIP}(\widetilde{\pi}_t,\zeta_t,1-\zeta_t)$.
    \STATE Assign the treatment $Z_t$ as 1 with probability $\pi_t$ and 0 with probability $1 - \pi_t$.
    \STATE Observe the covariates $\bX_t$, the surrogates $\bS_t=(S_t(0),S_t(1))^{\T}$, and the outcome $Y_t=Y_t(1)Z_t+Y_t(0)(1-Z_t)$.
    \STATE Calculate the estimated $\hat{m}_{S,0,t}(\cdot),\hat{m}_{Y,0,t}(\cdot),\hat{\gamma}_{0,t},\hat{m}_{S,1,t}(\cdot),\hat{m}_{Y,1,t}(\cdot)$, and $\hat{\gamma}_{1,t}$ according to \eqref{eq:m_S_hat}--\eqref{eq:gamma_hat}.
    \STATE Obtain the conditional mean function and variance estimators $\hat{\mu}_{0,t}(\bx,\bs),$ $\hat{\mu}_{1,t}(\bx,\bs),$ $\hat{\sigma}_{0,t}^2$, and $\hat{\sigma}_{1,t}^2$ according to~\eqref{eq:mu_hat} and \eqref{eq:sigma_hat}.
\ENDFOR
\STATE \textbf{Estimate}: Form an estimate of average treatment effect $\hat\tau_T$ using \eqref{eq:proposed_non}. 
\STATE \textbf{Output}: The adaptive AIPW estimator $\hat{\tau}_T$ under the proposed SLOACI design.
\end{algorithmic}
\end{algorithm}

Similar to Section~\ref{subsec:linear}, an initialization phase of length $T_0$ (see \textbf{lines 2 and 3} of Algorithm~\ref{alg:main}; the same applies hereafter) is conducted before the concentration phase. At each Stage $t$ in the concentration phase, given the observations $\{(\bX_j,Z_j,\bS_j,Y_j)\}_{j=1}^{t}$, we estimate the conditional mean functions for $z\in\{0,1\}$ using the Nadaraya-Watson estimator with form (\textbf{line 10}):
    \begin{equation}
        \label{eq:m_S_hat}
        \hat{m}_{S,z,t}(\bx)=\frac{\sum_{j=1}^{t}K_h(\bx-\bX_j)S_j(z)}{\sum_{j=1}^{t}K_h(\bx-\bX_j)},\quad \bx\in\eR^d,
    \end{equation}
    where $K_h(\bx)=K(\bx/h)/h$ is a multivariate kernel function with bandwidth $h>0$, and we define $\hat{m}_{S,z,t}(\bx)=0$ if $\sum_{j=1}^{t}K_h(\bx-\bX_j)=0.$ Similarly, the Nadaraya-Watson estimator of $m_{Y,z}(\cdot)$ is (\textbf{line 10}):
    \begin{equation}
        \label{eq:m_Y_hat}
        \hat{m}_{Y,z,t}(\bx)=\frac{\sum_{j=1}^{t}K_h(\bx-\bX_j)Y_jI(Z_j=z)}{\sum_{j=1}^{t}K_h(\bx-\bX_j)I(Z_j=z)},\quad\bx\in\eR^d,
    \end{equation}
where we define $\hat{m}_{Y,z,t}(\bx)=0$ if $\sum_{j=1}^{t}K_h(\bx-\bX_j)I(Z_j=z)=0.$
    Note that the Nadaraya-Watson method here can be replaced with any nonparametric regression method (e.g., spline regression, local linear regression, $k$-nearest neighbors), or supervised machine learning methods (e.g., random forests, deep neural networks). 
    In the empirical studies of the existing literature on adaptive experimental design, \cite{kato2025efficient} adopted the Nadaraya-Watson regression, and \cite{cook2024semiparametric} employed $k$-nearest neighbors and random forests, with the latter also utilized in \cite{oprescu2025efficient}.

    The linear coefficient $\gamma_z$ can be estimated via least squares as (\textbf{line 10}):
    \begin{equation}
        \label{eq:gamma_hat}
        \hat{\gamma}_{z,t}=\frac{\sum_{j=1}^{t}\big\{Y_j-\hat{m}_{Y,z,t}(\bX_j)\big\}\cdot\big\{S_j(z)-\hat{m}_{S,z,t}(\bX_j)\big\}I(Z_j=z)}{\sum_{j=1}^{t}\big\{S_j(z)-\hat{m}_{S,z,t}(\bX_j)\big\}^2I(Z_j=z)}.
    \end{equation}
The conditional mean functions and variances can be estimated as (\textbf{line 11}):
    \begin{align}
        &\hat{\mu}_{z,t}(\bx,\bs)=\hat{m}_{Y,z,t}(\bx)-\hat{\gamma}_{z,t}\hat{m}_{S,z,t}(\bx)+s(z)\hat{\gamma}_{z,t},\label{eq:mu_hat}\\
        &\hat{\sigma}_{z,t}^2=\frac{1}{\sum_{j=1}^{t}I(Z_j=z)}\sum_{j=1}^{t}\{Y_j-\hat{\mu}_{z,t}(\bX_j,\bS_j)\}^2I(Z_j=z).\label{eq:sigma_hat}
    \end{align} 
The initial allocation at Stage $t+1$ is updated as $\widetilde{\pi}_{t+1}$ according to \eqref{eq:init_pi} (\textbf{line 5}).
The adaptive treatment allocation is then obtained by clipping, defined in \eqref{eq:clip} below, as $\pi_{t+1}={\rm CLIP}(\widetilde{\pi}_{t+1},\zeta_{t+1},1-\zeta_{t+1}).$
Finally, after Stage $T$, the proposed adaptive AIPW estimator based on a cumulative sample size of $T$ is defined as \eqref{eq:proposed_non} (\textbf{line 13}).

We conclude this section by providing full details on clipping. 
To prevent $\widetilde{\pi}_t$ getting arbitrarily close to 0 or 1  at early stages in concentration phase,
we apply a clipping procedure \citep{dai2023clip,cook2024semiparametric,neopane2025logarithmic}, denoted as CLIP, 
to stabilize it. The resulting adaptive treatment allocation is defined as:
\begin{equation}
    \label{eq:clip}
    \pi_t={\rm CLIP}(\widetilde{\pi}_t,\zeta_t,1-\zeta_t):=\zeta_t\vee\{\widetilde{\pi}_t\wedge(1-\zeta_t)\},
\end{equation}
where $\zeta_t=(1/2)\cdot t^{-\eta}$ is a clipping threshold with clipping rate $\eta>0$. It projects $\widetilde{\pi}_t$ onto an interval that initially collapses to the treatment allocation 1/2 and then grows in width as $t$ increases, allowing for increased exploitation at later stages. The form of $\zeta_t$ can be chosen flexibly and does not need to follow the polynomial decay used in \eqref{eq:clip}. For example, one can adopt the exponential decay as in \cite{cook2024semiparametric}. 

The clipping procedure is necessary for both practical and theoretical considerations. For practical implementation, at the early concentration stages when the oracle treatment allocation $\pi^*$ cannot be estimated accurately, it is reasonable to assign treatment and control nearly evenly to ensure sufficient samples in both groups for more accurate statistical estimations. From a theoretical analysis perspective, controlling the magnitude of $\pi_t^{-1}$ is critical because otherwise, only probabilistic bounds on $\pi_t^{-1}$ can be derived, which are inadequate to establish the consistency and asymptotic normality of our proposed estimator.
To this end, an enforced upper bound on $\pi_t^{-1}$, such as $\pi_t^{-1}\le2t^{\eta}$, is needed for our technical analysis. 
In summary, along with the existing literature \citep[e.g.,][]{cook2024semiparametric,li2024optimal,oprescu2025efficient}, we employ both CLIP and initialization phase together to mitigate the efficiency loss caused by severely imbalanced allocations at early stages.


\section{Theoretical investigation}\label{sec:theory}

We now investigate the efficiency gain and theoretical properties of the proposed design and ATE estimator, with a particular focus on the non-asymptotic regime of $T<\infty$ that characterizes sampling properties under finite budgets.

\subsection{Efficiency gain from the surrogates}\label{ubsec:efficiency_gain}

To facilitate the analysis in this section, we first introduce an ATE estimator under the oracle treatment allocation $\pi^*:=\sigma_1/(\sigma_0+\sigma_1)$ as defined in \eqref{eq:varsigma_pi}, and then discuss how incorporating surrogates contributes to the efficiency gain. 

The oracle ATE estimator leverages the true conditional mean functions and the oracle treatment allocation, and is defined as:
\begin{equation}
    \label{eq:opt_est}
    \begin{aligned}
        \hat{\tau}_{T}^{*}= \frac{1}{T}\sum_{t=1}^{T}\phi_t^*,  
    \end{aligned}
\end{equation}
where $\phi_t^*$ is defined analogously to $\phi_t$ in \eqref{eq:single-stage-ATE} with sample estimates $\hat{\mu}_{z,t-1}(\bX_t,\bS_t)$ replaced with their population targets $\mu_z(\bX_t,\bS_t)$ for $z\in\{0,1\}$, and the adaptive treatment allocation $\pi_t$ replaced with the oracle treatment allocation $\pi^*$, respectively.

\begin{proposition}[Semiparametric efficiency bound with surrogates]\label{propos:efficiency_bound}
Under Condition \ref{cond:sutva} of the supplementary material, the oracle estimator $\hat{\tau}_{T}^{*}$ in \eqref{eq:opt_est} has a variance $V_T^*$ satisfying:
    \begin{equation}
    \label{eq:ney_var}
    T\cdot V_T^*=\varsigma_*^2:=\ell_1(\pi^*)=(\sigma_0+\sigma_1)^2+\eE_{(\bsX,\bsS)}\{\tau(\bX,\bS)-\tau_0\}^2,
\end{equation}
where $\sigma_0^2=(1-\rho_0^2)\sigma_{Y,0}^2,\sigma_1^2=(1-\rho_1^2)\sigma_{Y,1}^2$, and $\ell_1(\cdot)$ is defined in \eqref{eq:varsigma_pi}.
\end{proposition}

Proposition~\ref{propos:efficiency_bound} implies that the oracle estimator $\hat{\tau}_{T}^{*}$ achieves the semiparametric efficiency bound defined in Section~\ref{subsec:oracle} for any finite $T$. Remark \ref{rmk:efficiency gain of incorporating surrogates} below highlights the efficiency gain from the use of surrogates.

\begin{remark}\label{rmk:efficiency gain of incorporating surrogates}
  The semiparametric efficiency bound without surrogates takes the form
$
    T\cdot\widetilde{V}_T^*=(\sigma_{Y,0}+\sigma_{Y,1})^2+\eE_{\bsX}\{\tau(\bX)-\tau_0\}^2,
$
where $\tau(\bx):=\eE\{Y_t(1)|\bX_t=\bx\}-\eE\{Y_t(0)|\bX_t=\bx\}$, 
and $\eE_{\bsX}(\cdot)$ denotes the expectation over $\bX\sim\cX$. 
It can be derived that the efficiency gain from incorporating surrogates is
\begin{equation}
    \label{eq:ney_var_red}
 T\cdot\widetilde{V}_T^* -  T\cdot V_T^*  =2\Big(1-\sqrt{1-\rho_0^2}\sqrt{1-\rho_1^2}\Big)\sigma_{Y,0}\sigma_{Y,1} > 0, 
\end{equation}
if and only if $\rho_0^2+\rho_1^2>0$. 
This suggests that, as long as the surrogate is conditionally correlated with the primary outcome under at least one treatment arm, incorporating surrogates in estimating ATE yields an efficiency gain.
\end{remark}



The asymptotic properties of the proposed method and the required regularity conditions are provided in Section~\ref{supsec:asymptotic_theory} of the supplementary material; these results align with the standard results in the literature. Specifically, Proposition~\ref{propos:pi} implies that our treatment allocation converges in probability to the oracle one. Theorem~\ref{thm:asy_nor} shows that our proposed estimator $\hat{\tau}_T$ has an asymptotic variance attaining the semiparametric efficiency bound $\varsigma_*^2=T\cdot V_T^*$, demonstrating its optimality in the asymptotic framework.

We highlight two unique contributions of our asymptotic analysis here. \textit{First}, we relax the bounded outcome assumption commonly imposed in prior works \citep[e.g.,][]{kato2025efficient,neopane2025logarithmic,oprescu2025efficient}. \textit{Second}, we are the first to establish a theoretically valid range for the clipping rate $\eta$ under an asymptotic framework. This result provides formal guidance on a tuning parameter that has previously been chosen in an ad hoc or entirely empirical manner.

The asymptotic normality in Theorem~\ref{thm:asy_nor} allows for the construction of a fixed-time confidence interval for $\tau_0$ under the asymptotic regime. 
Let $u_t:=\phi_t-\tau_0$ with $\phi_t$ defined in \eqref{eq:single-stage-ATE}. By Lemmas~\ref{lem:mart} and \ref{lem:con_var_as} of the supplementary material, the sequence $\{u_t\}_{t\ge1}$ forms a square integrable martingale difference sequence with respect to (w.r.t.) $\{\cF_t\}_{t\ge1}$, whose conditional variance converges in probability to $\varsigma_*^2$. A plug-in estimator of $\varsigma_*^2$ is given by $\hat{\varsigma}_T^2=T^{-1}\sum_{t=1}^{T}(\phi_t-\hat{\tau}_T)^2.$
We establish the CLT-based asymptotic confidence interval for $\tau_0$ in Theorem \ref{thm:CI}.

\begin{theorem}
    \label{thm:CI}
    Suppose that the conditions of Theorem~\ref{thm:asy_nor} hold. 
    Then, as $T\to\infty$, $(L_T^{\CLT},U_T^{\CLT}):=\big(\hat{\tau}_T\pm\hat{\varsigma}_T\Phi^{-1}(1-\alpha/2)/\sqrt{T}\big)$ forms a $(1-\alpha)$ fixed-time confidence interval for $\tau_0$, where $\Phi^{-1}(\cdot)$ is the inverse CDF of the standard Gaussian distribution. 
\end{theorem}





\subsection{The regret function}\label{subsec:regret}

To facilitate the non-asymptotic analysis, we introduce the notion of regret. Formally, let $\cG=\{\mu:\cX\times\cS\to\eR\}$ denote the function class containing the conditional mean functions $\mu_1(\cdot,\cdot)$ and $\mu_0(\cdot,\cdot)$. At Stage $t$, we define the action
for ATE estimation as 
$a_t=(\pi_t,\hat{\mu}_{0,t-1},\hat{\mu}_{1,t-1})\in\cF_{t-1}$, with action space $\cA\subset (0,1)\times \cG\times \cG.$ 
The realized loss at Stage $t$ is given by the conditional variance of $\phi_t$ under $a_t$, that is
$
\ell(a_t):=\var(\phi_t|a_t)=\var(\phi_t|\cF_{t-1}),
$
where $\ell:\cA\to\eR_{\ge0}$.
Based on the definition of $a_t$, the conditional variance of $\phi_t$ arises from two distinct sources: the variability induced by the adaptive treatment allocation and the estimation error of the conditional mean functions.

According to Theorem~\ref{thm:variance} of the supplementary material, the normalized variance of $\hat{\tau}_T$ can then be expressed as the expected average loss
$
    T\cdot V_T=\eE\Big\{T^{-1}\sum_{t=1}^{T}\ell(a_t)\Big\}.
$
The realized loss can also be decomposed as two parts: $\ell(a_t):=\ell_1(\pi_t)+\ell_2(a_t)$, where $\ell_1(\pi_t)$ denotes the loss depending on the adaptive treatment allocation as defined in  \eqref{eq:varsigma_pi}, and $\ell_2(a_t)$ denotes the estimation loss. As shown in Remark~\ref{rmk:ell_2} of the supplementary material, $\ell_2(a_t)\ge0$, with equality if $\hat{\mu}_{0,t-1}=\mu_0$ and $\hat{\mu}_{1,t-1}=\mu_1$ almost surely. Since $\ell_1(\pi_t)\ge\ell_1(\pi^*)$, we further derive that 
$
    \ell(a_t)\ge\ell_1(\pi^*),
$
where recall that $\ell_1(\pi^*)$ is the semiparametric efficiency bound. This motivates us to define the optimal action $a^*$  as
$
    a^*=(\pi^*,\mu_0,\mu_1),
$
the one that leads to the oracle estimator $\hat{\tau}_T^*$ in \eqref{eq:opt_est}. 
Formally, we define the regret and its expectation for a sequence of actions $\{a_t\}_{t\in[T]}$ respectively as
\begin{equation}
    \label{eq:ney_regret}
    \begin{aligned}
        \cR_T:=&\sum_{t=1}^{T}\ell(a_t)-\min_{a\in\cA}\sum_{t=1}^{T}\ell(a)=\sum_{t=1}^{T}\ell(a_t)-T\ell(a^*),\\
        \eE\cR_T=&T^2\cdot(V_T-V_T^*)=T^2\cdot\eE\big(\hat{\tau}_T-\tau_0\big)^2-T^2\cdot\eE\big(\hat{\tau}_{T}^*-\tau_0\big)^2.
    \end{aligned}
\end{equation}

The regret $\cR_T$ measures the gap between the cumulative loss incurred by an adopted action and that of the optimal action with access to the true models.
By \eqref{eq:ney_regret}, $\eE\cR_T\ge0$, and its order w.r.t. $T$ characterizes how fast the normalized variance of an estimator under an adaptive experimental design converges to the semiparametric efficiency bound. Therefore, deriving the finite-sample bound for the expected regret $\eE\cR_T$ is of particular significance, as a smaller regret implies better finite-sample performance of the design.

\subsection{Non-asymptotic properties of our proposed estimator}\label{subsec:nonasymptotic}

Under the definition of regret in Section \ref{subsec:regret}, we now establish a finite-sample bound in Theorem \ref{thm:regret} below for the expected regret of our proposed design. 

\begin{theorem}
    \label{thm:regret}
    Suppose that Conditions~\ref{cond:sutva}--\ref{cond:varep} of the supplementary material hold, $\delta=T^{-2\beta/(2\beta+d)}$, $T_0$ is given by \eqref{eq:T0}, and there exists some constant $\ell_{\max}\ge\ell(a^*)$ such that $\max_{t\in[T]}\ell(a_t)\le\ell_{\max}$. Then, it holds that
    $$
    \begin{aligned}
        &\eE\cR_T\le\underbrace{\widetilde{C}_1T_0}_{\text{\small (I) initialization phase}}+
        \underbrace{\widetilde{C}_{\delta,T,2}\underline{\pi}^{-2\beta/(2\beta+d)-3}T^{d/(2\beta+d)}}_{\text{\small (II) cumulative loss from adaptive design}}
        \\
        &+\underbrace{\widetilde{C}_{\delta,T,3}\underline{\pi}^{-2\beta/(2\beta+d)-1}T^{d/(2\beta+d)}}_{\text{\small (III) estimation error of $\gamma_z$}}
        +\underbrace{\widetilde{C}_4\underline{\pi}^{-1}T^{d/(2\beta+d)}}_{\text{\small (IV) estimation error of $m_{S,z}(\cdot)$}}
        +\underbrace{\widetilde{C}_5\underline{\pi}^{-2\beta/(2\beta+d)-1}T^{d/(2\beta+d)}}_{\text{\small (V) estimation error of $m_{Y,z}(\cdot)$}}\\
        &+\underbrace{2\{\ell_{\max}-\ell(a^*)\}T^{d/(2\beta+d)}}_{\text{\small (VI) expected regret over the bad event}}\\
        &=\widetilde{O}(\underline{\pi}^{-2\beta/(2\beta+d)-3}T^{d/(2\beta+d)}),
    \end{aligned}
    $$
    where $\beta$ measures the smoothness of the conditional mean functions as in Definition~\ref{def:holder} of the supplementary material, $\underline{\pi}=\pi^*\wedge(1-\pi^*)$, $\widetilde{C}_1,\widetilde{C}_4$ and $\widetilde{C}_5$ are positive constants independent of $\delta$ and $T$, $\widetilde{C}_{\delta,T,2}\asymp(\log \delta^{-1}T)^2,$ and $\widetilde{C}_{\delta,T,3}\asymp\log \delta^{-1}T$, with explicit expressions provided in Section~\ref{supsubsubsec:proof_regret} of the supplementary material.
\end{theorem}

Theorem \ref{thm:regret} decomposes the non-asymptotic upper bound on the expected regret into various components. We discuss intuitions for each of these components below. 

Term (I) corresponds to the initialization phase and increases linearly with $T_0$, where $T_0$ grows polynomially in $\log T$. In practice, $T_0$ should be chosen as small as possible while satisfying the lower bound specified in~\eqref{eq:T0} below, so as to obtain a reasonably good allocation $\pi_{T_0+1}$ at the beginning of the concentration phase. Specifically, 
\begin{equation}
    \label{eq:T0}
    T_0=2\left\{\frac{2C_{\delta,T,4}}{\sigma_0^2\wedge\sigma_1^2\wedge(2^{-\beta/(2\beta+d)-2}\sigma_0\sigma_1\underline{\pi}^{\beta/(2\beta+d)})}\right\}^{(2\beta+d)/\beta}\vee\frac{16C_{\delta,T,0}^2}{\underline{\pi}^2}\vee\frac{1}{(2\underline{\pi})^{1/\eta}},
\end{equation}
where $C_{\delta,T,4}\asymp \log \delta^{-1}T$ and $C_{\delta,T,0}\asymp\sqrt{\log\log T+\log(1/\delta)}$, with explicit expressions provided in Section~\ref{supsubsubsec:event} of the supplementary material.
The impact of choices of $T_0$ on performance is further examined through simulation studies, as reported in Figure~\ref{fig:T0} of the supplementary material.

Terms (II)--(V) represent the cumulative loss incurred during the adaptive experimentation process and the cumulative estimation errors under the ``good" event where various model nuisances can be estimated well with probability at least $1-2\delta$; see Section~\ref{supsubsubsec:event} of the supplementary for formal definition of such a good event. 
Term (II) captures the loss due to the discrepancy between the adaptive treatment allocation $\pi_t$ and its oracle counterpart $\pi^*$. Terms (III)--(V) arise from  the estimation errors of model nuisances $(\gamma_0, \gamma_1)$,  $(m_{S,0}(\cdot), m_{S,1}(\cdot))$, and $(m_{Y,0}(\cdot), m_{Y,1}(\cdot))$, respectively;
all four terms are of the same polynomial order in $T$, up to a logarithmic factor. For large $T$ and small $\underline{\pi}$, Term (II) dominates, as $\pi_t$, being adaptively updated, accumulates bias from the upstream estimators over time. Furthermore, when $\underline{\pi}$ is sufficiently small, Term (V) further dominates Term (IV). Intuitively, this is because estimating $m_{Y,z}(\cdot)$ relies only on data observed under the corresponding treatment arm up to the current stage, whereas estimating $m_{S,z}(\cdot)$ can leverage data from both treatment and control groups, leading to faster convergence.

Term (VI) accounts for the expected regret incurred on the ``bad" event (i.e., complement of the good event discussed above). The uniform boundedness assumption on the loss $\ell(a_t)$, commonly adopted in existing literature \citep[e.g.,][]{perchet2013multi,qian2016kernel}, ensures that the cumulative regret $\cR_T$ increases at most linearly with $T$. Within our framework, this condition further implies that the conditional variance of $\phi_t$ remains bounded. On the good event with probability at least $1-2\delta$, we derive a regret bound of $C_{\delta,T}T^{-2\beta/(2\beta+d)}$ (Terms (II)--(V)), where $C_{\delta,T}$ depends only logarithmically on $\delta$ and $T$. Consequently, by setting $\delta\asymp T^{-2\beta/(2\beta+d)}$, the regret incurred under the bad event will not dominate the overall expected regret bound.


To avoid technical difficulty caused by arbitrarily small denominators in \eqref{eq:m_S_hat} and \eqref{eq:m_Y_hat}, we replace Condition~\ref{cond:kernel}(ii) of the supplementary material by incorporating a stabilization device following the approach in Theorem 2 of \cite{qian2016kernel}. Specifically, we replace the kernel function in \eqref{eq:m_S_hat} with a uniform kernel $K(\bu)=\underline{k}\cdot I(\|\bu\|_{\infty}\le 1)$ whenever $\sum_{j=1}^{t}K_{h}(\bX_j-\bx)\le\underline{k}\sum_{j=1}^{t}I(\|\bX_j-\bx\|_{\infty}\le h)/h$ for some small constant $\underline{k}>0.$ A similar adjustment applies to the estimator in \eqref{eq:m_Y_hat}. Since $\underline{k}$ can be arbitrarily small, this adjustment does not affect the use of $K(\cdot)$ in implementation. Alternatively, one may directly use the uniform kernel $K(\bu)=I(|\bu|_{\infty}\le1)$, which is also employed in the non-asymptotic analysis of the Nadaraya-Watson estimator (see Theorem 5.2 of \cite{gyorfi2006distribution}). 
Let $N_{1,t}=\sum_{j=1}^{t}I(Z_j=1)$ and $N_{0,t}=t-N_{1,t}$ for $t\in[T]$. To balance the bias-variance tradeoff, the kernel bandwidths are selected as follows: $h_{S,z,t}=c_h t^{-1/(2\beta+d)}$ for estimating $\hat{m}_{S,z,t}(\cdot)$, and $h_{Y,z,t}=c_h N_{z,t}^{-1/(2\beta+d)}$ for estimating $\hat{m}_{Y,z,t}(\cdot)$. Here, $c_h\asymp 1$ is a pre-specified or data-driven constant. Without loss of generality, we set $c_h= 1$ at each stage. In practice, we may regard $\beta$ as 1, corresponding to the Lipschitz continuity case.

For the linear Model~\eqref{eq:model_linear}, the non-asymptotic analysis can be similarly obtained and we provide the details in Section~\ref{supsubsec:bound_linear} of the supplementary material.

\subsubsection{The role of surrogates in regret bound}\label{subsubsec:role_surrogate}

Our non-asymptotic analysis offers a new perspective on the role of surrogates in adaptive experimental design. We formalize this connection by relating the expected regret $\eE\cR_T$ to a variance-based efficiency metric. Building on the definition of expected regret, we define the Neyman ratio as
$
\upsilon_T =(V_T - V_T^*)/V_T^* \asymp T^{-1}\eE\cR_T,
$
which measures the relative efficiency loss of the proposed design compared to the oracle benchmark. This definition is consistent with the one used in \cite{dai2023clip}. Similar regret-based efficiency analyses appear in \cite{dai2023clip, neopane2025logarithmic, neopane2025optimistic, noarov2025stronger}.

According to Theorem~\ref{thm:regret}, 
$\eE\cR_T=\widetilde{O}(\underline{\pi}^{-2\beta/(2\beta+d)-3}T^{d/(2\beta+d)})=o(T)$. Since $\eE\cR_T = o(T)$, the ratio $\upsilon_T$ converges to zero, implying that our adaptive design asymptotically attains the same efficiency as the oracle design. Intuitively, the regret bound improves when the outcome model is smoother (larger $\beta$) and the covariate dimension ($d$) is smaller, but worsens when treatment probabilities are highly imbalanced (small $\underline{\pi}$).

In the asymptotic regime, incorporating surrogates has been shown to reduce the semiparametric efficiency bound. However, one might ask whether similar benefits could be achieved by simply including the surrogates as part of the covariates. For example, one could redefine the covariates under control and treatment as
$(\bX_t^{\T},S_t(0))^{\T}\in\eR^{d+1}$ and $(\bX_t^{\T},S_t(1))^{\T}\in\eR^{d+1}$, respectively. In this case, the resulting oracle variance still matches the semiparametric efficiency bound $\varsigma_*^2=T\cdot V_T^*$. 

However, the increased covariate dimension from $d$ to $d+1$ slows down the convergence rates of nonparametric estimators such as Nadaraya-Watson regression, which in turn worsens the regret bound. In fact, including surrogates as part of covariates in Nadaraya-Watson regression yields a slower rate of $
    \widetilde{O}(T^{(d+1)/(2\beta+d+1)})$ (assuming bounded $\underline{\pi}$), in contrast to the $\widetilde{O}(T^{d/(2\beta+d)})$ rate established in Theorem~\ref{thm:regret}.

Therefore, treating $S_t(0)$ and $S_t(1)$ as surrogates rather than as additional covariates is not merely a trivial modeling choice. The surrogate-based formulation preserves the semiparametric efficiency bound while avoiding unnecessary dimension inflation. This insight underscores the \textit{dual benefits} of surrogates: they enhance efficiency in the asymptotic sense and simultaneously preserve the convergence rate in the non-asymptotic sense.

\subsubsection{Comparison with existing results}\label{subsubsec:comparison}

We compare our regret bound with existing studies that derive the finite-sample regret bound for the efficient ATE estimation via adaptive experimental designs.

\cite{dai2023clip} pioneer this effort by introducing the notion of regret under a design-based setting and derive a regret bound of order $\widetilde{O}(\sqrt{T})$ based on the IPW estimator.
\cite{noarov2025stronger} establish a logarithmic regret under a similar setting but impose stronger assumptions.
\cite{neopane2025logarithmic} extend this line of work to the superpopulation setting and also obtain a logarithmic regret.
However, the IPW estimator is known to be suboptimal, and the regrets defined in the above studies are based on the minimum variance achievable by the IPW estimator.
To attain the semiparametric efficiency bound, \cite{neopane2025optimistic} generalize the previous work to the AIPW estimator and also obtain a logarithmic regret.
All of the aforementioned studies focus on models without covariates.

To investigate heterogeneous treatment effects, \cite{li2024optimal} and \cite{kato2025efficient} consider nonparametric models that incorporate covariates.
\cite{kato2025efficient} do not explicitly establish a regret bound in terms of the order of $T$, but provide a heuristic result.
Different from the fully adaptive experiments used in other works, \cite{li2024optimal} adopt a low-switching adaptive design, where the number of batches (i.e., the number of allocation updates) is denoted by $K$.
Assuming that the MSE convergence rates of the nonparametric regressions for the conditional mean and variance functions are $\alpha'$ and $\beta'$, respectively, they obtain a regret bound of order $\widetilde{O}(T^{1-\alpha'} + K T^{1-\beta'})$.
When $K=O(\log T)$, this bound simplifies to $\widetilde{O}(T^{1-\alpha'\wedge\beta'})$, which aligns with the result established in our Theorem~\ref{thm:regret}, since $\alpha'=\beta'=2\beta/(2\beta+d)$ in our framework (see Section~\ref{supsubsubsec:event} of the supplementary material). Furthermore, our proposed adaptive design under the linear Model~\eqref{eq:model_linear} achieves a regret bound of $\widetilde{O}((\log T)^2)$.
This bound further improves to $\widetilde{O}(\log T)$ when both $\mu_0(\cdot,\cdot)$ and $\mu_1(\cdot,\cdot)$ are constant functions, consistent with the setting considered in \cite{neopane2025optimistic}.
Detailed derivations are provided in Section~\ref{supsubsec:bound_linear} of the supplementary material.
A comprehensive comparison is summarized in Table~\ref{tab:comparison_regret}.

\begin{table}[ht]
\centering
\scriptsize
\caption{\small Comparisons between the existing adaptive experimental designs in literature for the ATE estimation and our proposed method.}
\label{tab:comparison_regret}
\vspace{1em}
\begin{tabular}{c|ccccc}
\hline
Method & Setting & Covariate & Estimator & Adaptivity & Regret Bound\\ 
\hline
\cite{dai2023clip} & design-based & not included & IPW & fully & $\widetilde{O}(\sqrt{T})$\\
\cite{noarov2025stronger} & design-based & not included & IPW  & fully & $\widetilde{O}(\log T)$\\
\cite{neopane2025logarithmic} & superpopulation & not included & IPW & fully & $\widetilde{O}(\log T)$\\
\cite{neopane2025optimistic} & superpopulation & not included & AIPW & fully & $\widetilde{O}(\log T)$\\
\cite{li2024optimal} & superpopulation & included & AIPW & low switching & $\widetilde{O}(T^{1-\alpha'}+KT^{1-\beta'})$\\
\cite{kato2025efficient} & superpopulation & included & AIPW & fully & / \\
SLOACI & superpopulation & included & AIPW & fully & $\widetilde{O}(T^{d/(2\beta+d)})$ \\
\hline
\end{tabular}
\end{table}

Compared with prior work, our theoretical results address several technical challenges and offer distinct advantages.
\textit{First}, our regret bound is the first to explicitly incorporate a concrete nonparametric regression method, the Nadaraya-Watson estimator, into the theoretical analysis. This allows us to derive not only the asymptotic order but also the lower-order terms and constant factors, which have not been characterized in the existing literature.
\textit{Second}, in establishing the regret bound, we develop new exponential-tail concentration inequalities for estimating both the linear coefficient and the conditional variance in a partially linear model via Robinson’s transformation. These results are of independent theoretical interest and may be useful beyond the present context.
\textit{Third}, while most existing works focus on parametric models, the only comparable nonparametric result, due to \cite{li2024optimal}, relies on a low-switching and offline design that updates the treatment allocation $\pi_t$ only a limited number of times. This restriction stems from their reliance on polynomial-tail concentration inequalities, which simplify derivations but constrain adaptivity in practice. In contrast, our method operates in a fully online manner and updates $\pi_t$ at each stage, enabled by the sharper exponential-tail concentration results established in this work.

\section{Sequential testing}\label{sec:testing}


\subsection{Sequential testing setup}\label{subsec:testing_setup}

In this section, we propose a comprehensive sequential testing framework that complements our proposed design under both asymptotic and non-asymptotic regimes. We consider testing the following hypothesis defined for some $\tau_{H_0}\in\eR$:
$$
    H_0:\tau_0=\tau_{H_0}\quad{\rm v.s.}\quad H_1:\tau_0\neq\tau_{H_0}.
$$
Note that the confidence interval in Theorem~\ref{thm:CI} provides valid coverage only for a fixed total time horizon $T$. In contrast, in sequential testing \citep[Cf.,][]{siegmund2013sequential}, the experimenter continuously monitors the data and may stop the experiment once the test statistic becomes significant. Such “peeking” can inflate the type-I error rate when the statistic is constructed using a fixed-time method \citep{waudby2024time}.

A common remedy is to apply multiple-testing corrections such as the Bonferroni (BF) or Benjamini-Hochberg procedures. If the total time horizon $T$ is known in advance, the Bonferroni-corrected confidence intervals
\begin{equation}
    \label{eq:bf_ci}
    (L_t^{\BF},U_t^{\BF}):=\big(\hat{\tau}_t\pm\hat{\varsigma}_t\Phi^{-1}(1-T^{-1}\alpha/2)/\sqrt{t}\big)
\end{equation}
for $t\in[T]$ ensure type-I error control by rejecting $H_0$ whenever $\tau_{H_0}\notin(L_t^{\BF},U_t^{\BF})$ for some $t$. However, this approach ``spends'' the type-I error rate evenly across all stages, resulting in overly conservative tests; large $T$ makes the Bonferroni correction even more conservative \citep{kato2025efficient}.  
Although the experimenter may consider using more advanced alpha-spending functions \citep[]{
lan2009further}, they still fail to exploit the martingale structure of $\{\phi_t\}_{t\ge1}$ and are therefore not fully leveraging the data structure.

To improve power while maintaining type-I error control, we leverage the martingale structure of $\{\phi_t\}_{t\ge1}$ and develop time-uniform confidence sequences that are fully adaptive and remain valid at arbitrary stopping times. These sequences enable anytime-valid inference under both asymptotic and non-asymptotic regimes (see Sections~\ref{subsec:asy_cs} and~\ref{subsec:eb_cs}), eliminating the need for fixed-horizon adjustments and avoiding penalties for data monitoring. To facilitate the discussion, we start with defining  ``confidence sequences" in Definition \ref{def:asycs} (as introduced in Definition 2.1 of \cite{waudby2024time}). 

\begin{definition}
    \label{def:asycs}
    We say that the intervals $(L_t,U_t)_{t\ge1}$, centered at the estimators $\{\hat{\tau}_t\}_{t\ge1}$ with lower and upper bounds $L_t,U_t,\forall t\ge1$, form a $(1-\alpha)$ asymptotic confidence sequence for a sequence of real parameters $\{\tau_t\}_{t\ge1}$ if there exists a (typically unknown) $(1-\alpha)$ non-asymptotic confidence sequence $(\widetilde{L}_t,\widetilde{U}_t)_{t\ge1}$ for $\{\tau_t\}_{t\ge1}$, that is, satisfying: (i) $\eP(\forall t\ge1,\tau_t\in[\widetilde{L}_t,\widetilde{U}_t])\ge1-\alpha$; (ii) $\widetilde{L}_t/L_t\to 1$ a.s. and $\widetilde{U}_t/U_t\to 1$ a.s. as $t\to\infty.$
\end{definition}

\subsection{Sequential testing under asymptotic regime}
\label{subsec:asy_cs}


We are now ready to present our sequential testing procedure under the asymptotic regime, where the asymptotic confidence sequence provides time-uniform coverage. For all $t\ge1$, recall that $\hat{\tau}_t=t^{-1}\sum_{j=1}^{t}\phi_j,$ and $\hat{\varsigma}_t^2=t^{-1}\sum_{j=1}^{t}(\phi_j-\hat{\tau}_t)^2$. Theorem \ref{thm:asy_cs} below provides a construction of an asymptotic (ASY) confidence sequence for $\tau_0$.

\begin{theorem}
    \label{thm:asy_cs}
    Under the conditions of Theorem~\ref{thm:CI}, for a pre-specified parameter $\varrho>0$,
    \begin{equation}
        \label{eq:asy_cs}
        (L_t^{\ASY},U_t^{\ASY}):=\left(\hat{\tau}_t\pm\sqrt{\frac{t\hat{\varsigma}_t^2\varrho^2+1}{t^2\varrho^2}\log\left(\frac{t\hat{\varsigma}_t^2\varrho^2+1}{\alpha^2}\right)}\right)
    \end{equation}
    forms a $(1-\alpha)$ asymptotic confidence sequence for $\tau_0.$
\end{theorem}

Given time $t$ and type-I error rate $\alpha,$ one can choose the parameter $\varrho$ to minimize the interval width in Theorem \ref{thm:asy_cs}. As discussed in Appendix B.2 of \cite{waudby2024time}, there exists an exact but implicit solution for the optimal $\varrho$, and a closed-form approximation is given by $\varrho=\sqrt{\{-2\log\alpha+\log(-2\log\alpha+1)\}/(t\hat{\varsigma}_t^2)}$.

\subsection{Sequential testing under non-asymptotic regime} \label{subsec:eb_cs}


We next present a non-asymptotic confidence sequence  (Cf. Definition~\ref{def:asycs}) for the ATE based on recent advances for inferring the means of bounded random variables, the so-called rewards, in sequential settings \citep{waudby2024estimating,waudby2024anytime}.
In our framework, the reward at Stage $t$ can be chosen as $\phi_t$. Without loss of generality, we assume that $Y_t(0),Y_t(1)\in[0,1]$ for all $t\ge1$; otherwise, this can be ensured by a rescaling technique. Accordingly, we also assume that $\hat{\mu}_{0,t-1}(\cdot,\cdot),\hat{\mu}_{1,t-1}(\cdot,\cdot)\in[0,1]$. Let $k_t=\{\pi_t\wedge (1-\pi_t)\}^{-1}+1$. Since $\pi_t\in\cF_{t-1}$, we have that $\{k_t\}_{t\ge1}$ is a nonnegative and predictable sequence. Define
$
\xi_t=\frac{\phi_t}{k_t+1},\quad \hat{\xi}_{t-1}=\left(\frac{1}{t-1}\sum_{j=1}^{t-1}\xi_j\right)\wedge\frac{1}{k_t+1},\quad \widebar{\xi}_t=\left(\frac{1}{t}\sum_{j=1}^{t}\xi_j\right)\wedge\frac{1}{k_t+1}.
$
Then $\xi_t$ is the reweighted reward, and $\hat{\xi}_{t-1}$ and $\widebar{\xi}_t$ are the average reweighted reward with truncation which  serve as estimates of the mean of $\{\xi_t\}_{t\ge1}$.  

We provide the construction of a non-asymptotic  empirical Bernstein (EB)  confidence sequence (Cf. Definition \ref{def:asycs}) in Theorem \ref{thm:finite_cs}.

\begin{theorem}
    \label{thm:finite_cs}
    Suppose that Condition~\ref{cond:sutva} of the supplementary material holds, and $Y_t(0),$ $Y_t(1),$ $\hat{\mu}_{0,t-1}(\cdot,\cdot),$ $\hat{\mu}_{1,t-1}(\cdot,\cdot)\in[0,1]$ for all $t\ge1.$ For a given $\alpha\in(0,1)$, define the sequence $\{\lambda_{\alpha,t}\}_{t\ge1}$ as
    $$
    \lambda_{\alpha,t}=\sqrt{\frac{2\log(1/\alpha)}{\hat{\nu}_{t-1}^2t\log(1+t)}}\wedge c,\quad\hat{\nu}_t^2=\frac{\nu_0^2+\sum_{j=1}^{t}(\xi_j-\widebar{\xi}_j)^2}{t+1},
    $$
    where $c\in(0,1)$ is some truncation parameter, and $\nu_0^2$ is some pre-specified parameter that can be thought of as a prior guess for the variance of $\{\xi_t\}_{t\ge1}$. Then, 
    \begin{equation}
        \label{eq:non_cs}
            (L_t^{\EB},U_t^{\EB}):=\left(\frac{\sum_{j=1}^{t}\lambda_{\alpha,j}\xi_j}{\sum_{j=1}^{t}\lambda_{\alpha,j}/(k_j+1)}\pm\frac{\log(1/\alpha)+\sum_{j=1}^{t}(\xi_j-\hat{\xi}_{j-1})^2\psi_E(\lambda_{\alpha,j})}{\sum_{j=1}^{t}\lambda_{\alpha,j}/(k_j+1)}\right)
    \end{equation}
    forms a $(1-\alpha)$ non-asymptotic confidence sequence  for $\tau_0,$ where $\psi_E(\lambda):=-\log(1-\lambda)-\lambda.$
\end{theorem}

The intuition of Theorem~\ref{thm:finite_cs} is to construct a predictable empirical Bernstein supermartingale, as in \eqref{eq:supermartingale} of the supplementary material, based on the reward $\phi_t$ at each stage, and then apply Ville's inequality to guarantee the coverage of a non-asymptotic confidence sequence for the mean reward $\tau_0=\eE(\phi_t|\cF_{t-1})$.
In Theorem~\ref{thm:finite_cs}, $\hat{\nu}_t^2$  serves as an estimate of the variance of $\{\xi_t\}_{t\ge1}$. The predictable truncation parameter $k_t$ is central to maintaining the validity of the confidence sequence: it, together with the function $\psi_E(\cdot)$, guarantees the construction of a valid supermartingale. $\lambda_{\alpha,j}$ is a tuning parameter, and $(\xi_j-\hat{\xi}_{j-1})^2$ captures the fluctuation of the process, thereby enabling the variance adaptivity of the resulting confidence sequence.

Unlike the approach in \cite{cook2024semiparametric}, which adopts a diverging truncation sequence $k_t=\zeta_t^{-1}\to\infty$, our method determines $k_t$ in a fully data-driven manner. By Proposition~\ref{propos:pi} of the supplementary material, $k_t$ converges in probability to $\underline{\pi}^{-1}+1$ as $t\to\infty$, providing an automatic and stable calibration mechanism. 
Following Remark 2 of \cite{waudby2024anytime}, a larger $k_t$ increases variance adaptivity—enhancing robustness to model misspecification—but also yields wider confidence sequences at early stages. Consequently, when $\underline{\pi}$ is small, the method exhibits reduced power at the early stages, reflecting a trade-off between early precision and long-run adaptivity. Additionally, since the two proposed sequential testing procedures are both variance-adaptive, the resulting confidence intervals are narrower than those of existing methods \citep[e.g.,][]{cook2024semiparametric,kato2025efficient}, owing to the efficiency gain induced by surrogates.


\subsection{Trade-offs and practical guidance}
\label{subsec:tradeoff}


To compare the two sequential testing regimes, we introduce the concept of a stopping time. Given a confidence sequence $(L_t,U_t)_{t\ge1}$, the experimenter rejects $H_0$ once there exists a Stage $t$ such that $\tau_{H_0}\notin(L_t,U_t)$. To mitigate the influence of large estimation errors at early stages, the experimenter may begin inference only after a pre-specified number of stages. Formally, the stopping time (if it exists) is defined as $\vartheta:=\min\{t\ge t_0:\tau_{H_0}\notin(L_t,U_t)\}$, where $t_0$ denotes the initial peeking time \citep[]{waudby2024time, cook2024semiparametric}. In what follows, we shall discuss the trade-offs of the sequential testing procedures developed above.

The asymptotic confidence sequence in \eqref{eq:asy_cs} enables anytime-valid inference under the same conditions required by the CLT. Its width scales naturally with the data variance, mirroring the empirical variance in the CLT, thereby exhibiting strong variance adaptivity. This property further underscores the benefit of our surrogate-leveraged adaptive design in reducing the semiparametric efficiency bound. However, since it is valid only asymptotically, in practice, to maintain proper type-I error control, the experimenter may need to select a sufficiently large initial peeking time.

The non-asymptotic confidence sequence in \eqref{eq:non_cs} provides anytime-valid coverage guarantees that hold uniformly over all stages, ensuring that the type-I error rate is controlled regardless of the chosen initial peeking time. However, the resulting confidence sequence is generally more conservative than its asymptotic counterpart, as illustrated by our numerical experiments. Moreover, its validity relies on the boundedness of the primary outcomes, which can be restrictive in practice. In addition, when $\underline{\pi}$ is very small, the method can experience reduced power in small samples.

The trade-offs between the two sequential testing regimes offer practical insights for implementation. {\it On the one hand}, when the primary outcomes are heavy-tailed or when there is a large discrepancy in conditional variances between the treatment and control groups—corresponding to unbounded outcomes and a small $\underline{\pi}$, respectively—the asymptotic confidence sequence is generally preferred. {\it On the other hand}, when the number of experimental stages is too limited, the empirical Bernstein confidence sequence remains the more reliable choice. 
We evaluate the empirical performance of the proposed adaptive designs and sequential testing procedures via a series of simulation studies, with the results presented in Section~\ref{supsec:sim} of the supplementary material.

\section{A synthetic case study}
\label{sec:case}

We now investigate the performance of SLOACI through a synthetic case study using data from Project STAR (Student-Teacher Achievement Ratio), which has been extensively analyzed in prior research (e.g., \citealp{chetty2011does}; \citealp{athey2025combining}). Project STAR was a four-year randomized experiment conducted in 79 Tennessee public schools to examine the effect of class size on test scores. Following \cite{athey2025combining}, we examine whether assignment to a small class as opposed to a regular-size class significantly improves the reading and math scores via a linear model. 

\subsection{Setup}\label{subsec:case_setup}

Let $Y_t(1)$ and $Y_t(0)$ denote the potential outcomes for third-grade student $t$, representing the total standardized test scores in reading and math if assigned to a small class or a regular-size class, respectively. The observed outcome is standardized to have a mean of zero and unit variance. The covariates $\bX_t$ include (i) students' race and ethnicity, (ii) gender, and (iii) eligibility for free or reduced-price lunch. After excluding observations with missing values, the analysis is based on 3523 observations. To provide a self-contained illustration of how to use synthetic surrogates to improve inference power, we follow \cite{mccaw2024synthetic} by splitting the data into a  training dataset $\cD_1$ and an inference dataset $\cD_2$ with 1523 and 2000 observations, respectively. For $t\in\cD_1$, we use the observed $Y_t(0)$ and $Y_t(1)$ as dependent variables and the nine available independent variables\footnote{Including school type, highest degree of the teacher, teacher’s career ladder level, years of total teaching experience, teacher’s ethnicity, school system ID, and the three covariates in $\bX_t$.}, denoted by $\widetilde{\bX}_t$, to train two separate random forest models $\hat{f}_0(\cdot)$ and $\hat{f}_1(\cdot)$. Then, the synthetic surrogates are constructed as $S_t(0)=\hat{f}_0(\widetilde{\bX}_t)$ and $S_t(1)=\hat{f}_1(\widetilde{\bX}_t)$ for all $t\in\cD_1\cup\cD_2$. We assume that $(\bX_t,S_t(0),S_t(1),Y_t(0),Y_t(1))$ follows the linear Model~\eqref{eq:model_linear}, where model errors are assumed to follow the bivariate Gaussian distributions (see Remark~\ref{rmk:biv_gau} for details).

Since our data were collected from a randomized experiment, we need to impute the counterfactual outcomes to mimic the implementation of an adaptive design and evaluate the results. To address this issue, we propose the following novel approach, which exploits the property of the joint regression Model~\eqref{eq:model_linear} that for $z\in\{0,1\},$
$$
Y_t(z)|(\bX_t,S_t(z))\sim\cN\Big(\balpha_{Y,z}^{\T}\bX_t+\rho_z\sigma_{Y,z}\sigma_{S,z}^{-1}\big(S_t(z)-\balpha_{S,z}^{\T}\bX_t\big),(1-\rho_z^2)\sigma_{Y,z}^2\Big).
$$
\begin{itemize}
    \item[(i)] Using the training dataset $\{(\bX_t,Z_t,Y_t,S_t(0),S_t(1))\}_{t\in\cD_1}$, we fit Model~\eqref{eq:model_linear} and obtain the estimators $\hat{\balpha}_{Y,z},\hat{\balpha}_{S,z},\hat{\sigma}_{Y,z}^2,\hat{\sigma}_{S,z}^2$ and $\hat{\rho}_z$ for each $z\in\{0,1\}$.
    
    \item[(ii)] For each $t\in\cD_2,$ if $Y_t(0)$ is missing, then it is imputed as $Y_t(0)=\hat{\balpha}_{Y,0}^{\T}\bX_t+\hat{\rho}_1\hat{\sigma}_{Y,1}\hat{\sigma}_{S,1}^{-1}\big(S_t(1)-\hat{\balpha}_{S,0}^{\T}\bX_t\big)+\varepsilon_t(0)$, where $\varepsilon_t(0)$ is independently sampled from $\cN\big(0,(1-\hat{\rho}_0^2)\hat{\sigma}_{Y,1}^2\big)$. If $Y_t(1)$ is missing, an analogous formulation is applied.
\end{itemize}

Using the imputed data, we compare our SLOACI with two competing designs: RAR \citep{cook2024semiparametric,kato2025efficient} and complete randomization with a fixed allocation $\pi_t=0.5$ (denoted as ``RCT"). The detailed procedure of RAR is given in Section~\ref{supsubsec:rar} of the supplementary material. The sequential testing procedure involves three valid tests: BF, ASY, and EB.
The true value of the ATE is defined as the difference-in-means estimate based on $\{(Y_t(0),Y_t(1))\}_{t\in\cD_2}$, where counterfactual outcomes are imputed by their condition means without randomness, leading to $\tau_0=0.229$.
Since educational interventions have delayed outcomes, with a one-year delay in our case, a fully adaptive design is infeasible. Therefore, we adopt a batch-wise adaptive design with batch sizes $b=400$ and $b=50$. The former mimics a five-year randomized experiment with $b=400$ observations each year and is consistent with real-world practice, whereas the latter is used to assess the impact of reducing the batch size. 
To implement the adaptive design in Section~\ref{subsec:linear} and the sequential testing procedure extended from Section~\ref{sec:testing} to the linear model, we set $\eta=0.25$, $T_0=b,\tau_{H_0}=0,t_0=100,c=0.5,\nu_0^2=1$, and $\hat{\xi}_0=0$. The nominal type-I error rate is fixed at $\alpha=0.05$. We consider a simulation-based evaluation with 2000 replications.

\subsection{Results}\label{subsec:case_results}

Figure~\ref{fig:case_power} presents the normalized variance and empirical power curves for different designs and tests, with the corresponding means and standard errors of the stopping times reported in Table~\ref{tab:case_stop}. Moreover, a specific replication from the simulations is presented in Section~\ref{supsec:empirical} of the supplementary material to demonstrate the trajectory of confidence sequences against the cumulative sample size. 

\begin{figure}[ht]
\begin{center}
\includegraphics[width=0.8\linewidth]{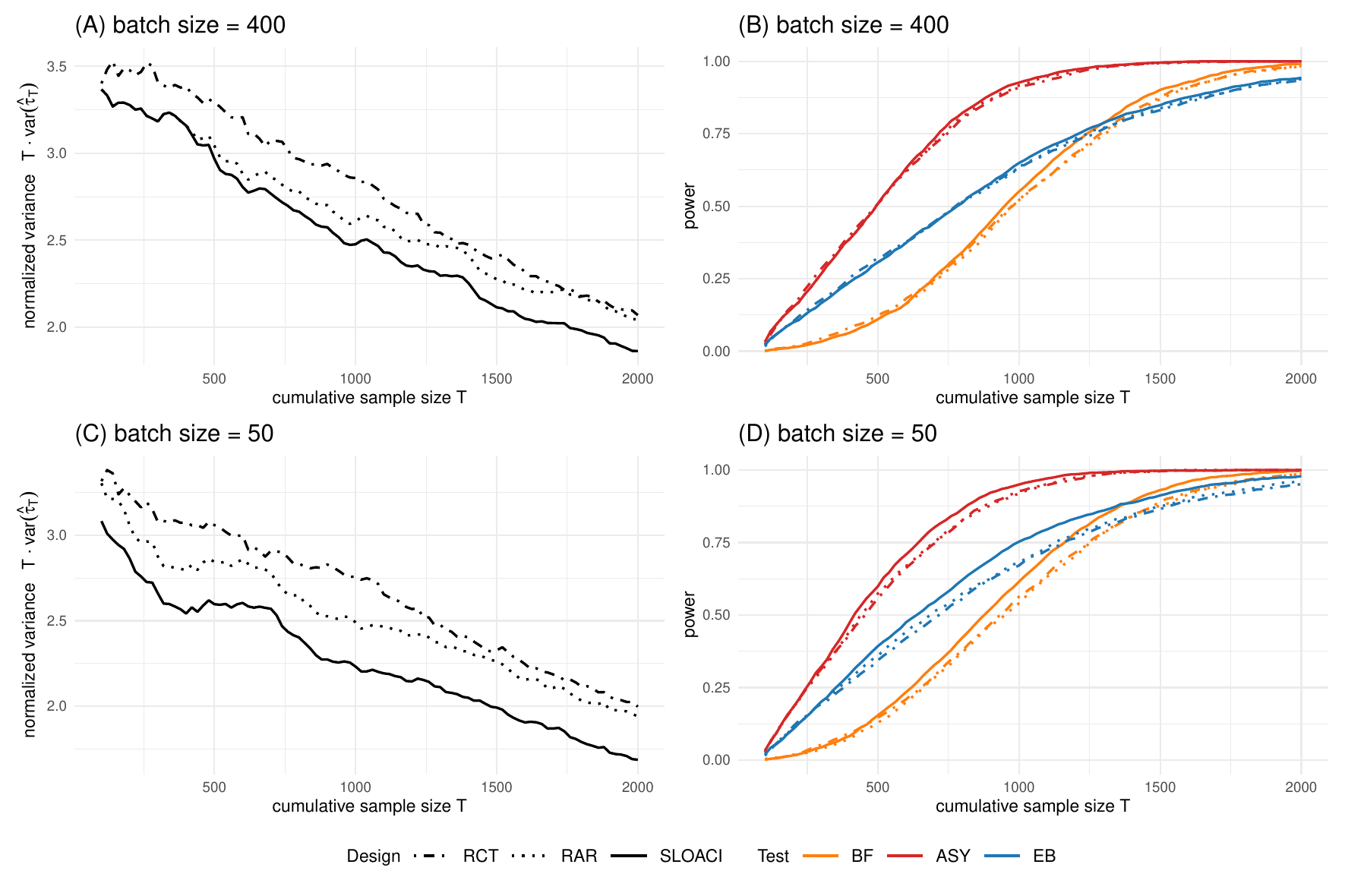}
\end{center}
\vspace{-1em}
\caption{\small Comparison of normalized variances of the (adaptive) AIPW estimator and empirical powers across different designs and tests for the STAR data, with batch size $b=400$ in panels (A)--(B) and $b=50$ in panels (C)--(D).}
\label{fig:case_power}
\end{figure}

\begin{table}[ht]
\footnotesize
\centering
\caption{\small Means and standard errors (in parentheses) of stopping times for the  STAR data. 
}
\label{tab:case_stop}
\vspace{1em}
\begin{tabular}{lcccccc}
\hline
Test & \multicolumn{2}{c}{BF} & \multicolumn{2}{c}{ASY} & \multicolumn{2}{c}{EB} \\
Batch size $b$ & 400 & 50 & 400 & 50 & 400 & 50 \\
\hline
\multicolumn{1}{l}{RCT} & 1000.0 (9.4)  & 960.3 (9.4)  & 536.2 (7.2) & 504.9 (6.9) & 872.7 (12.5) & 816.5 (12) \\
    \multicolumn{1}{l}{RAR} & 1013.8 (9.3) & 975.5 (9.4) & 535.7 (7.0) & 502.3 (6.9) & 874.4 (12.3) & 796.6 (11.8) \\
    \multicolumn{1}{l}{SLOACI} & 977.5 (8.8) & 901.3 (8.5) & 526.6 (6.8) & 470.1 (6.2) & 860.9 (12.1) & 730.7 (10.6) \\
\hline
\end{tabular}
\end{table}

Several patterns can be observed: \textit{First}, the proposed SLOACI consistently achieves the smallest variance, the highest power, and the shortest average stopping time across simulations. Among the three tests, ASY attains the highest power relative to BF and EB. \textit{Second}, when the batch size decreases from 400 to 50 (with a correspondingly shorter initialization phase), SLOACI demonstrates improved performance relative to the other benchmark methods. In particular, when $b=50$, SLOACI significantly reduces the average stopping time, accounting for the reported standard errors. \textit{Third}, sequential testing provides insight into whether the adaptive design and test can effectively shorten the experimental horizon. For example, in the case of $b=400$, where the pre-specified horizon corresponds to 5 years, a stopping time not larger than 1200 implies that the treatment effect can be detected within 3 years. Under the baseline RCT, the proportions of experiments in which the horizon is reduced to 3 years are 68.30\%, 96.05\%, and 72.65\% under the BF, ASY, and EB tests, respectively; under RAR, the corresponding proportions are 67.50\%, 96.50\%, and 73.40\%. By contrast, the proposed SLOACI further increases these proportions to 72.15\%, 97.20\%, and 74.55\%. This highlights that incorporating surrogates may reduce the experimental horizon, thereby yielding economic and social benefits.

\section{Discussion}
\label{sec:discuss}

We discuss several important directions for future research. {\it First}, SLOACI may be extended to the heteroskedastic setting, which generalizes the design to a covariate-adjusted surrogate-leveraged RAR design. {\it Second}, we may incorporate the double/debiased machine learning framework \citep{chernozhukov2018double}, which uses cross-fitting to avoid overfitting and leverages Neyman orthogonality to mitigate bias inflation. {\it Third}, our algorithm involves an exploration phase whose length $T_0$ must be pre-specified. We may consider replacing CLIP with more advanced techniques, such as the Optimistic Policy Tracking Algorithm proposed in \cite{neopane2025optimistic}, to allow $T_0$ to be determined adaptively within the experimentation process. The above topics are beyond the scope of the current paper and will be pursued elsewhere.

\paragraph{Data Availability Statement:} The data that support the findings of this study are openly available in R package ``AER" at \url{https://rdrr.io/cran/AER/man/STAR.html}. 

{\small 
\linespread{1}\selectfont
\bibliographystyle{chicago}
\bibliography{references}

\begin{thebibliography}{}

\bibitem[\protect\citeauthoryear{Anderer, Bastani, and Silberholz}{Anderer et~al.}{2022}]{anderer2022adaptive}
Anderer, A., H.~Bastani, and J.~Silberholz (2022).
\newblock Adaptive clinical trial designs with surrogates: When should we bother?
\newblock {\em Management Science\/}~{\em 68\/}(3), 1982--2002.

\bibitem[\protect\citeauthoryear{Athey, Chetty, and Imbens}{Athey et~al.}{2025}]{athey2025combining}
Athey, S., R.~Chetty, and G.~Imbens (2025).
\newblock The experimental selection correction estimator: Using experiments to remove biases in observational estimates.
\newblock {\em arXiv preprint arXiv:2006.09676v2\/}.

\bibitem[\protect\citeauthoryear{Athey, Chetty, Imbens, and Kang}{Athey et~al.}{2025}]{athey2025surrogate}
Athey, S., R.~Chetty, G.~W. Imbens, and H.~Kang (2025).
\newblock The surrogate index: Combining short-term proxies to estimate long-term treatment effects more rapidly and precisely.
\newblock {\em Review of Economic Studies\/}, in press.

\bibitem[\protect\citeauthoryear{Bertsimas, Korolko, and Weinstein}{Bertsimas et~al.}{2019}]{bertsimas2019covariate}
Bertsimas, D., N.~Korolko, and A.~M. Weinstein (2019).
\newblock Covariate-adaptive optimization in online clinical trials.
\newblock {\em Operations Research\/}~{\em 67\/}(4), 1150--1161.

\bibitem[\protect\citeauthoryear{Cai and Rafi}{Cai and Rafi}{2024}]{cai2024performance}
Cai, Y. and A.~Rafi (2024).
\newblock On the performance of the {N}eyman allocation with small pilots.
\newblock {\em Journal of Econometrics\/}~{\em 242\/}(1), 105793.

\bibitem[\protect\citeauthoryear{Chen and Simchi-Levi}{Chen and Simchi-Levi}{2025}]{chen2025efficient}
Chen, H. and D.~Simchi-Levi (2025).
\newblock Efficient switchback experiments with surrogate variables: Estimation and experimental design.
\newblock {\em Management Science\/}, 1--18.

\bibitem[\protect\citeauthoryear{Chen, Leung, and Qin}{Chen et~al.}{2008}]{chen2008improving}
Chen, S.~X., D.~H. Leung, and J.~Qin (2008).
\newblock Improving semiparametric estimation by using surrogate data.
\newblock {\em Journal of the Royal Statistical Society Series B: Statistical Methodology\/}~{\em 70\/}(4), 803--823.

\bibitem[\protect\citeauthoryear{Chernozhukov, Chetverikov, Demirer, Duflo, Hansen, Newey, and Robins}{Chernozhukov et~al.}{2018}]{chernozhukov2018double}
Chernozhukov, V., D.~Chetverikov, M.~Demirer, E.~Duflo, C.~Hansen, W.~Newey, and J.~Robins (2018).
\newblock Double/debiased machine learning for treatment and structural parameters.
\newblock {\em The Econometrics Journal\/}~{\em 21\/}(1), C1--C68.

\bibitem[\protect\citeauthoryear{Chetty, Friedman, Hilger, Saez, Schanzenbach, and Yagan}{Chetty et~al.}{2011}]{chetty2011does}
Chetty, R., J.~N. Friedman, N.~Hilger, E.~Saez, D.~W. Schanzenbach, and D.~Yagan (2011).
\newblock How does your kindergarten classroom affect your earnings? {E}vidence from {Project STAR}.
\newblock {\em The Quarterly Journal of Economics\/}~{\em 126\/}(4), 1593--1660.

\bibitem[\protect\citeauthoryear{Cook, Mishler, and Ramdas}{Cook et~al.}{2024}]{cook2024semiparametric}
Cook, T., A.~Mishler, and A.~Ramdas (2024).
\newblock Semiparametric efficient inference in adaptive experiments.
\newblock In {\em Causal Learning and Reasoning}, pp.\  1033--1064.

\bibitem[\protect\citeauthoryear{Dai, Gradu, and Harshaw}{Dai et~al.}{2023}]{dai2023clip}
Dai, J., P.~Gradu, and C.~Harshaw (2023).
\newblock {Clip-OGD}: An experimental design for adaptive {N}eyman allocation in sequential experiments.
\newblock In {\em Advances in Neural Information Processing Systems}, Volume~36, pp.\  32235--32269.

\bibitem[\protect\citeauthoryear{Gy{\"o}rfi, Kohler, Krzyzak, and Walk}{Gy{\"o}rfi et~al.}{2006}]{gyorfi2006distribution}
Gy{\"o}rfi, L., M.~Kohler, A.~Krzyzak, and H.~Walk (2006).
\newblock {\em A Distribution-Free Theory of Nonparametric Regression}.
\newblock Springer Science \& Business Media.

\bibitem[\protect\citeauthoryear{Hahn}{Hahn}{1998}]{hahn1998role}
Hahn, J. (1998).
\newblock On the role of the propensity score in efficient semiparametric estimation of average treatment effects.
\newblock {\em Econometrica\/}~{\em 66\/}(2), 315--331.

\bibitem[\protect\citeauthoryear{Hahn, Hirano, and Karlan}{Hahn et~al.}{2011}]{hahn2011adaptive}
Hahn, J., K.~Hirano, and D.~Karlan (2011).
\newblock Adaptive experimental design using the propensity score.
\newblock {\em Journal of Business \& Economic Statistics\/}~{\em 29\/}(1), 96--108.

\bibitem[\protect\citeauthoryear{Howard, Ramdas, McAuliffe, and Sekhon}{Howard et~al.}{2021}]{howard2021time}
Howard, S.~R., A.~Ramdas, J.~McAuliffe, and J.~Sekhon (2021).
\newblock Time-uniform, nonparametric, nonasymptotic confidence sequences.
\newblock {\em The Annals of Statistics\/}~{\em 49\/}(2), 1055--1080.

\bibitem[\protect\citeauthoryear{Hu and Rosenberger}{Hu and Rosenberger}{2006}]{hu2006theory}
Hu, F. and W.~F. Rosenberger (2006).
\newblock {\em The Theory of Response-Adaptive Randomization in Clinical Trials}.
\newblock John Wiley \& Sons.

\bibitem[\protect\citeauthoryear{Hu, Zhang, and He}{Hu et~al.}{2009}]{hu2009efficient}
Hu, F., L.-X. Zhang, and X.~He (2009).
\newblock Efficient randomized-adaptive designs.
\newblock {\em The Annals of Statistics\/}~{\em 37\/}(5A), 2543--2560.

\bibitem[\protect\citeauthoryear{Hu, Zhu, and Hu}{Hu et~al.}{2015}]{hu2015unified}
Hu, J., H.~Zhu, and F.~Hu (2015).
\newblock A unified family of covariate-adjusted response-adaptive designs based on efficiency and ethics.
\newblock {\em Journal of the American Statistical Association\/}~{\em 110\/}(509), 357--367.

\bibitem[\protect\citeauthoryear{Imbens, Kallus, Mao, and Wang}{Imbens et~al.}{2025}]{imbens2025long}
Imbens, G., N.~Kallus, X.~Mao, and Y.~Wang (2025).
\newblock Long-term causal inference under persistent confounding via data combination.
\newblock {\em Journal of the Royal Statistical Society Series B: Statistical Methodology\/}~{\em 87\/}(2), 362--388.

\bibitem[\protect\citeauthoryear{Kallus and Mao}{Kallus and Mao}{2025}]{kallus2025role}
Kallus, N. and X.~Mao (2025).
\newblock On the role of surrogates in the efficient estimation of treatment effects with limited outcome data.
\newblock {\em Journal of the Royal Statistical Society Series B: Statistical Methodology\/}~{\em 87\/}(2), 480--509.

\bibitem[\protect\citeauthoryear{Kato, Ishihara, Honda, and Narita}{Kato et~al.}{2025}]{kato2025efficient}
Kato, M., T.~Ishihara, J.~Honda, and Y.~Narita (2025).
\newblock Efficient adaptive experimental design for average treatment effect estimation.
\newblock {\em arXiv preprint arXiv:2002.05308v7\/}.

\bibitem[\protect\citeauthoryear{Lan and DeMets}{Lan and DeMets}{2009}]{lan2009further}
Lan, K.~G. and D.~DeMets (2009).
\newblock Further comments on the alpha-spending function.
\newblock {\em Statistics in Biosciences\/}~{\em 1\/}(1), 95--111.

\bibitem[\protect\citeauthoryear{Li, Simchi-Levi, and Zhao}{Li et~al.}{2024}]{li2024optimal}
Li, J., D.~Simchi-Levi, and Y.~Zhao (2024).
\newblock Optimal adaptive experimental design for estimating treatment effect.
\newblock {\em arXiv preprint arXiv:2410.05552v3\/}.

\bibitem[\protect\citeauthoryear{McCaw, Gao, Lin, and Gronsbell}{McCaw et~al.}{2024}]{mccaw2024synthetic}
McCaw, Z.~R., J.~Gao, X.~Lin, and J.~Gronsbell (2024).
\newblock Synthetic surrogates improve power for genome-wide association studies of partially missing phenotypes in population biobanks.
\newblock {\em Nature Genetics\/}~{\em 56\/}(7), 1527--1536.

\bibitem[\protect\citeauthoryear{McCaw, Gaynor, Sun, and Lin}{McCaw et~al.}{2023}]{mccaw2023leveraging}
McCaw, Z.~R., S.~M. Gaynor, R.~Sun, and X.~Lin (2023).
\newblock Leveraging a surrogate outcome to improve inference on a partially missing target outcome.
\newblock {\em Biometrics\/}~{\em 79\/}(2), 1472--1484.

\bibitem[\protect\citeauthoryear{Mikusheva and S{\o}lvsten}{Mikusheva and S{\o}lvsten}{2025}]{mikusheva2025linear}
Mikusheva, A. and M.~S{\o}lvsten (2025).
\newblock Linear regression with weak exogeneity.
\newblock {\em Quantitative Economics\/}~{\em 16\/}(2), 367--403.

\bibitem[\protect\citeauthoryear{Neopane, Ramdas, and Singh}{Neopane et~al.}{2025a}]{neopane2025logarithmic}
Neopane, O., A.~Ramdas, and A.~Singh (2025a).
\newblock Logarithmic {N}eyman regret for adaptive estimation of the average treatment effect.
\newblock In {\em International Conference on Artificial Intelligence and Statistics}, Volume 258, pp.\  4303--4311.

\bibitem[\protect\citeauthoryear{Neopane, Ramdas, and Singh}{Neopane et~al.}{2025b}]{neopane2025optimistic}
Neopane, O., A.~Ramdas, and A.~Singh (2025b).
\newblock Optimistic algorithms for adaptive estimation of the average treatment effect.
\newblock In {\em International Conference on Machine Learning}.

\bibitem[\protect\citeauthoryear{Newey}{Newey}{1997}]{newey1997convergence}
Newey, W.~K. (1997).
\newblock Convergence rates and asymptotic normality for series estimators.
\newblock {\em Journal of Econometrics\/}~{\em 79\/}(1), 147--168.

\bibitem[\protect\citeauthoryear{Neyman, Dabrowska, and Speed}{Neyman et~al.}{1990}]{splawa1990application}
Neyman, J., D.~M. Dabrowska, and T.~P. Speed (1990).
\newblock On the application of probability theory to agricultural experiments. {E}ssay on principles. {S}ection 9.
\newblock {\em Statistical Science\/}, 465--472.

\bibitem[\protect\citeauthoryear{Noarov, Fogliato, Bertran, and Roth}{Noarov et~al.}{2025}]{noarov2025stronger}
Noarov, G., R.~Fogliato, M.~Bertran, and A.~Roth (2025).
\newblock Stronger {N}eyman regret guarantees for adaptive experimental design.
\newblock In {\em International Conference on Machine Learning}.

\bibitem[\protect\citeauthoryear{Oprescu, Cho, and Kallus}{Oprescu et~al.}{2025}]{oprescu2025efficient}
Oprescu, M., B.~M. Cho, and N.~Kallus (2025).
\newblock Efficient adaptive experimentation with non-compliance.
\newblock In {\em Advances in Neural Information Processing Systems}.

\bibitem[\protect\citeauthoryear{Perchet and Rigollet}{Perchet and Rigollet}{2013}]{perchet2013multi}
Perchet, V. and P.~Rigollet (2013).
\newblock The multi-armed bandit problem with covariates.
\newblock {\em The Annals of Statistics\/}~{\em 41\/}(2), 693--721.

\bibitem[\protect\citeauthoryear{Qian and Yang}{Qian and Yang}{2016}]{qian2016kernel}
Qian, W. and Y.~Yang (2016).
\newblock Kernel estimation and model combination in a bandit problem with covariates.
\newblock {\em Journal of Machine Learning Research\/}~{\em 17\/}(149), 1--37.

\bibitem[\protect\citeauthoryear{Robinson}{Robinson}{1988}]{robinson1988root}
Robinson, P.~M. (1988).
\newblock Root-{N}-consistent semiparametric regression.
\newblock {\em Econometrica\/}~{\em 56\/}(4), 931--954.

\bibitem[\protect\citeauthoryear{Rubin}{Rubin}{2005}]{rubin2005causal}
Rubin, D.~B. (2005).
\newblock Causal inference using potential outcomes: Design, modeling, decisions.
\newblock {\em Journal of the American Statistical Association\/}~{\em 100\/}(469), 322--331.

\bibitem[\protect\citeauthoryear{Siegmund}{Siegmund}{2013}]{siegmund2013sequential}
Siegmund, D. (2013).
\newblock {\em Sequential Analysis: Tests and Confidence Intervals}.
\newblock Springer Science \& Business Media.

\bibitem[\protect\citeauthoryear{Speckman}{Speckman}{1988}]{speckman1988kernel}
Speckman, P. (1988).
\newblock Kernel smoothing in partial linear models.
\newblock {\em Journal of the Royal Statistical Society Series B: Statistical Methodology\/}~{\em 50\/}(3), 413--436.

\bibitem[\protect\citeauthoryear{Tymofyeyev, Rosenberger, and Hu}{Tymofyeyev et~al.}{2007}]{tymofyeyev2007implementing}
Tymofyeyev, Y., W.~F. Rosenberger, and F.~Hu (2007).
\newblock Implementing optimal allocation in sequential binary response experiments.
\newblock {\em Journal of the American Statistical Association\/}~{\em 102\/}(477), 224--234.

\bibitem[\protect\citeauthoryear{Veitch, Sridhar, and Blei}{Veitch et~al.}{2020}]{veitch2020adapting}
Veitch, V., D.~Sridhar, and D.~Blei (2020).
\newblock Adapting text embeddings for causal inference.
\newblock In {\em Conference on Uncertainty in Artificial Intelligence}, Volume 124, pp.\  919--928.

\bibitem[\protect\citeauthoryear{Waudby-Smith, Arbour, Sinha, Kennedy, and Ramdas}{Waudby-Smith et~al.}{2024}]{waudby2024time}
Waudby-Smith, I., D.~Arbour, R.~Sinha, E.~H. Kennedy, and A.~Ramdas (2024).
\newblock Time-uniform central limit theory and asymptotic confidence sequences.
\newblock {\em The Annals of Statistics\/}~{\em 52\/}(6), 2613--2640.

\bibitem[\protect\citeauthoryear{Waudby-Smith and Ramdas}{Waudby-Smith and Ramdas}{2024}]{waudby2024estimating}
Waudby-Smith, I. and A.~Ramdas (2024).
\newblock Estimating means of bounded random variables by betting.
\newblock {\em Journal of the Royal Statistical Society Series B: Statistical Methodology\/}~{\em 86\/}(1), 1--27.

\bibitem[\protect\citeauthoryear{Waudby-Smith, Wu, Ramdas, Karampatziakis, and Mineiro}{Waudby-Smith et~al.}{2024}]{waudby2024anytime}
Waudby-Smith, I., L.~Wu, A.~Ramdas, N.~Karampatziakis, and P.~Mineiro (2024).
\newblock Anytime-valid off-policy inference for contextual bandits.
\newblock {\em ACM/IMS Journal of Data Science\/}~{\em 1\/}(3), 1--42.

\bibitem[\protect\citeauthoryear{Zhang, Hu, Cheung, and Chan}{Zhang et~al.}{2007}]{zhang2007asymptotic}
Zhang, L., F.~Hu, S.~H. Cheung, and W.~S. Chan (2007).
\newblock Asymptotic properties of covariate-adjusted response-adaptive designs.
\newblock {\em The Annals of Statistics\/}~{\em 35\/}(3), 1166--1182.

\bibitem[\protect\citeauthoryear{Zhao}{Zhao}{2025}]{zhao2025adaptive}
Zhao, J. (2025).
\newblock Adaptive {N}eyman allocation.
\newblock {\em arXiv preprint arXiv:2309.08808v4\/}.

\bibitem[\protect\citeauthoryear{Zhao and Zhou}{Zhao and Zhou}{2025}]{zhao2025pigeonhole}
Zhao, J. and Z.~Zhou (2025).
\newblock Pigeonhole design: Balancing sequential experiments from an online matching perspective.
\newblock {\em Management Science\/}~{\em 71\/}(3), 1889--1908.

\end{thebibliography}
}

\clearpage
\linespread{1.77}\selectfont

\begin{center}
	{\noindent \bf \large Supplementary Material to ``SLOACI: Surrogate-Leveraged Online Adaptive Causal Inference"}\\
\end{center}
	\begin{center}
		{Yingying Fan, Zihan Wang and Waverly Wei}
	\end{center}
\bigskip

\setcounter{page}{1}
\setcounter{section}{0}
\renewcommand{\thesection}{S.\arabic{section}}
\renewcommand{\theHsection}{S.\arabic{section}}
\setcounter{lemma}{0}
\renewcommand{\thelemma}{S.\arabic{lemma}}
\renewcommand{\theHlemma}{S.\arabic{lemma}}
\setcounter{equation}{0}
\renewcommand{\theequation}{S.\arabic{equation}}
\renewcommand{\theHequation}{S.\arabic{equation}}
\setcounter{theorem}{0}
\renewcommand{\thetheorem}{S.\arabic{theorem}}
\renewcommand{\theHtheorem}{S.\arabic{theorem}}
\setcounter{remark}{0}
\renewcommand{\theremark}{S.\arabic{remark}}
\renewcommand{\theHremark}{S.\arabic{remark}}
\setcounter{definition}{0}
\renewcommand{\thedefinition}{S.\arabic{definition}}
\renewcommand{\theHdefinition}{S.\arabic{definition}}
\setcounter{proposition}{0}
\renewcommand{\theproposition}{S.\arabic{proposition}}
\renewcommand{\theHproposition}{S.\arabic{proposition}}
\setcounter{figure}{0}
\renewcommand{\thefigure}{S.\arabic{figure}}
\renewcommand{\theHfigure}{S.\arabic{figure}}
\setcounter{table}{0}
\renewcommand{\thetable}{S.\arabic{table}}
\renewcommand{\theHtable}{S.\arabic{table}}
\setcounter{condition}{0}
\renewcommand{\thecondition}{S.\arabic{condition}}
\renewcommand{\theHcondition}{S.\arabic{condition}}
\setcounter{algorithm}{0}
\renewcommand{\thealgorithm}{S.\arabic{algorithm}}
\renewcommand{\theHalgorithm}{S.\arabic{algorithm}}

This supplementary material contains the asymptotic properties of our method and the regularity conditions in Section~\ref{supsec:asymptotic_theory}, the technical proofs of theoretical results in Section~\ref{supsec:proof}, several extensions in Section~\ref{supsec:method}, all simulation results in Section~\ref{supsec:sim}, and additional empirical results in Section~\ref{supsec:empirical}. We start by introducing some additional notation.
\noindent\textit{Notation}. For a sub-exponential random variable $\varepsilon_1$ and a sub-Gaussian random variable $\varepsilon_2$, denote the corresponding sub-exponential norm by $\|\varepsilon_1\|_{\psi_1}$ and the sub-Gaussian norm by $\|\varepsilon_2\|_{\psi_2}$. For a set $A,$ we use $A^c$ to denote its complement. The symbols `$\to_p$', `$\to_d$' and `$\to_{a.s.}$' mean convergence in probability, in distribution and almost surely, respectively. For random variables $X$ and $Y,$ $X\indep Y$ denotes that $X$ and $Y$ are independent. For $x\in {\mathbb R},$ we use $\lceil x\rceil$ and $\lfloor x\rfloor$ as the ceiling and floor functions of $x$, respectively.

\section{Asymptotic properties of our method}\label{supsec:asymptotic_theory}

The asymptotic properties of the proposed method in this section highlight two benefits of our proposal. \textit{First}, incorporating surrogates yields an efficiency gain compared to the design without surrogates (see Section \ref{ubsec:efficiency_gain}). \textit{Second}, our proposed estimator requires less stringent assumptions than those imposed in the existing literature to establish consistency and asymptotic normality results. 



We consider the following regularity conditions commonly imposed in the  literature.
\begin{condition}
    \label{cond:sutva}
    (i) (SUTVA) The Stable Unit Treatment Value Assumption (Cf. Section 1.6 of \textcolor{blue}{Imbens and Rubin}, \textcolor{blue}{2015}) holds. 
    (ii) (Homogeneity) $(\bX_t^{\T},S_t(0),S_t(1),Y_t(0),Y_t(1))$ is independently and identically drawn from the same superpopulation for each Stage $t$. (iii) (Unconfoundedness) $(S_t(0), S_t(1),Y_t(0),Y_t(1))\indep Z_t|\bX_t,\cF_{t-1}$. (iv) (Positivity) $0<\pi_t<1$.
\end{condition}

We next show that our proposed adaptive AIPW estimator $\hat{\tau}_{T}$ in Section \ref{subsec:procedure}  asymptotically attains the semiparametric efficiency bound.

\begin{theorem}
    \label{thm:variance}
    Suppose that Condition~\ref{cond:sutva} holds. 
    Then, the proposed estimator $\hat{\tau}_T$ defined in \eqref{eq:proposed_non} is unbiased, and has a variance $V_T$ satisfying:
    $$
    \begin{aligned}
    T\cdot V_T=&\frac{1}{T}\sum_{t=1}^{T}\eE\Big[\frac{\sigma_1^2}{\pi_t}+\frac{\sigma_0^2}{1-\pi_t}+\eE_{(\bsX,\bsS)}\{\tau(\bX,\bS)-\tau_0\}^2\\
    &+\frac{1-\pi_t}{\pi_t}\eE_{(\bsX_t,\bsS_t)}\big[\{\mu_1(\bX_t,\bS_t)-\hat{\mu}_{1,t-1}(\bX_t,\bS_t)\}^2|\cF_{t-1}\big]\\
    &+\frac{\pi_t}{1-\pi_t}\eE_{(\bsX_t,\bsS_t)}\big[\{\mu_0(\bX_t,\bS_t)-\hat{\mu}_{0,t-1}(\bX_t,\bS_t)\}^2|\cF_{t-1}\big]\\
    &+2\eE_{(\bsX_t,\bsS_t)}\big[\{\mu_1(\bX_t,\bS_t)-\hat{\mu}_{1,t-1}(\bX_t,\bS_t)\}\{\mu_0(\bX_t,\bS_t)-\hat{\mu}_{0,t-1}(\bX_t,\bS_t)\}|\cF_{t-1}\big]\Big].
    \end{aligned}
    $$
\end{theorem}

Let $u_t:=\phi_t-\tau_0$ with $\phi_t$ defined in \eqref{eq:single-stage-ATE}. Leveraging the property that $\{u_t\}_{t\ge1}$ forms a square integrable martingale difference sequence w.r.t. $\{\cF_t\}_{t\ge1}$, we derive the theoretical properties of $\hat{\tau}_T$ using the law of large numbers and the central limit theorem for martingale difference sequences. Furthermore, to characterize the smoothness of the conditional mean functions, we introduce the class of functions known as the H\"older class, denoted by $\Sigma(\beta,L)$ and we provide the definition below in Definition \ref{def:holder} (as given in Definition 1.2 of \textcolor{blue}{Tsybakov} (\textcolor{blue}{2009})).
\begin{definition}
    \label{def:holder}
    Let $\beta$ and $L$ be two positive numbers. The H\"older class $\Sigma(\beta,L)$ is defined as the set of $k=\lfloor\beta\rfloor$  times differentiable functions $m(\cdot):\eR^d\to\eR$ whose derivative $m^{(k)}(\cdot)$ satisfies $|m^{(k)}(\bx)-m^{(k)}(\bx')|\le L\|\bx-\bx'\|^{\beta-k},\forall\bx,\bx'\in\eR^d.$
\end{definition}

We need some additional conditions for the theoretical results. 

\begin{condition}
    \label{cond:mean_fun}
    The conditional mean functions $m_{Y,0}(\cdot),m_{Y,1}(\cdot),m_{S,0}(\cdot)$ and $m_{S,1}(\cdot)$ belong to the function class $\Sigma(\beta,L)$ for some $0<\beta\le 1.$
\end{condition}

\begin{condition}
    \label{cond:X}
    The covariate domain $\cX$ is a compact set. The distribution of $\bX_t$ is dominated by the Lebesgue measure, with a continuous density $p(\bx)$ uniformly bounded above and away from zero on $\cX$, i.e., $\underline{c}\le p(\bx)\le\bar{c}$ for some positive constants $\underline{c}\le\bar{c}$.
\end{condition}

\begin{condition}
    \label{cond:kernel}
    The multivariate kernel function satisfies: (i) $K(\bx)=0$ for $\|\bx\|_{\infty}>1,$ and there exists some constant $\bar{k}>0$ such that $0\le K(\bx)< \bar{k},\forall\bx\in\eR^d$; (ii) there exists some constant $\check{k}>0$ such that $|K(\bx)-K(\bx')|\le\check{k}\|\bx-\bx'\|_{\infty},\forall\bx,\bx'\in\eR^d$.
\end{condition}


\begin{condition}
    \label{cond:varep}
    The model errors $\{\varepsilon_{Y,t}(z)\}$ and $\{\varepsilon_{S,t}(z)\}$ for $z\in\{0,1\}$ satisfy that: (i) there exist constants $v_{Y,z},v_{S,z}>0$ such that $\eE\{\exp\big(\lambda\varepsilon_{Y,t}(z))|X_t\}\le \exp(v_{Y,z}^2\lambda^2/2\big)$ and $\eE\{\exp\big(\lambda\varepsilon_{S,t}(z))|X_t\}\le \exp(v_{S,z}^2\lambda^2/2\big)$ for all $t\ge1$ and $\lambda\in\eR;$ (ii) there exist constants $K_Y,K_S>0$ such that $\|\varepsilon_{Y,t}(z)\|_{\psi_2}\le K_Y$ and $\|\varepsilon_{S,t}(z)\|_{\psi_2}\le K_S$ for all $t\ge1$.
\end{condition}

    Condition~\ref{cond:mean_fun} is adopted in \textcolor{blue}{Tsybakov} (\textcolor{blue}{2009}); \textcolor{blue}{Perchet and Rigollet} (\textcolor{blue}{2013}), and corresponds to Lipschitz continuity when $\beta=1.$ This smoothness condition leads to a more challenging estimation problem and is weaker than the conventional condition in the kernel regression literature (e.g., \textcolor{blue}{Fan and Gijbels}, \textcolor{blue}{1996}), which typically requires second-order smoothness. 
    Condition~\ref{cond:X} imposes a natural assumption on the support and density of the covariates (\textcolor{blue}{Yang and Zhu}, \textcolor{blue}{2002}; \textcolor{blue}{Perchet and Rigollet}, \textcolor{blue}{2013}; \textcolor{blue}{Qian and Yang}, \textcolor{blue}{2016}). This assumption avoids the use of the trimming technique when estimating linear coefficients $\gamma_0$ and $\gamma_1$, as in \textcolor{blue}{Robinson} (\textcolor{blue}{1988}), where selecting the trimming parameter is practically challenging; see also the discussion in \textcolor{blue}{Linton} (\textcolor{blue}{1995}). 
    Condition~\ref{cond:kernel} concerns a kernel function $K(\bx):\eR^d\to\eR$ that is nonnegative, bounded, compactly supported, and Lipschitz, which is commonly used in the literature (e.g., \textcolor{blue}{Tsybakov}, \textcolor{blue}{2009}; \textcolor{blue}{Qian and Yang}, \textcolor{blue}{2016}).  
Condition \ref{cond:varep} provides sub-Gaussian tail conditions on the model errors that are commonly assumed in the literature (\textcolor{blue}{Cai and Rafi}, \textcolor{blue}{2024}; \textcolor{blue}{Li et al.}, \textcolor{blue}{2024}).
This assumption ensures the uniform consistency of the Nadaraya-Watson regression (\textcolor{blue}{Qian and Yang}, \textcolor{blue}{2016}), and thereby guarantees the consistency of both the conditional variance estimators and our adaptive treatment allocation. We remark that for establishing the $L^2$ convergence of the Nadaraya-Watson regression, a finite second moment condition would already suffice.

For each $z\in\{0,1\}$, to implement the Nadaraya-Watson regressions, the bandwidths are selected as $h_{S,z,t}=c_h t^{-1/\kappa}$ for estimating $\hat{m}_{S,z,t}(\cdot)$, and $h_{Y,z,t}=c_h N_{z,t}^{-1/\kappa}$ for estimating $\hat{m}_{Y,z,t}(\cdot)$. Here, we assume $\kappa>2d$ to ensure the uniform consistency of the Nadaraya-Watson regressions (\textcolor{blue}{Qian and Yang}, \textcolor{blue}{2016}). To balance the bias-variance tradeoff, we take $\kappa=2\beta+d$ when $d<2\beta,$ and $\kappa=2d+\epsilon$ with a sufficiently small positive constant $\epsilon$ when $d\ge2\beta$.

\begin{proposition}
    \label{propos:pi}
    Under Conditions \ref{cond:sutva}--\ref{cond:varep}, it holds that $\pi_t\to_p\pi^*$ as $t\to\infty.$
\end{proposition}


Leveraging Proposition \ref{propos:pi}, we have the following asymptotical normality results for our proposed estimator in Theorem \ref{thm:asy_nor}. 

\begin{theorem}
    \label{thm:asy_nor}
    Suppose that Conditions \ref{cond:sutva}--\ref{cond:varep} hold, and there exists some constant $M_{\gamma}>0$ such that $\max_{z\in\{0,1\}}|\hat{\gamma}_{z,t}|\le M_{\gamma}$ almost surely for $t\ge1$. Then, it holds that: (i) the proposed estimator $\hat{\tau}_T$ defined in \eqref{eq:proposed_non} is strongly consistent; (ii) moreover, if the clipping rate satisfies $0<\eta<2\beta/(2\beta+\kappa)$, then $\hat{\tau}_T$ is asymptotically normal:
    $$
    \sqrt{T}(\hat{\tau}_T-\tau_0)\to_d \cN(0,\varsigma_*^2),
    $$
    as $T\to\infty,$ where $\varsigma_*^2$ is defined in \eqref{eq:ney_var}.
\end{theorem}

Theorem \ref{thm:asy_nor} shows that our proposed estimator $\hat{\tau}_T$ has an asymptotic variance attaining the semiparametric efficiency bound $\varsigma_*^2=T\cdot V_T^*$, demonstrating its optimality in the asymptotic framework.
Our strong consistency result is new to the literature and is crucial for constructing the asymptotic confidence sequence in Theorem~\ref{thm:asy_cs}. 

It is worth noting that some prior works (e.g., \textcolor{blue}{Cook et al.}, \textcolor{blue}{2024}; \textcolor{blue}{Kato et al.}, \textcolor{blue}{2025}) assume a convergence rate for $\pi_t$ when establishing the asymptotic normality of the ATE estimator, while our Proposition \ref{propos:pi} does not provide an explicit rate.  While such convergence rate assumption is indeed necessary under heteroskedastic models, it is not required in our case due to the homoskedasticity assumption. In our non-asymptotic framework, we derive explicitly that $|\pi_t-\pi^*|=\widetilde{O}_p(t^{-2\beta/(2\beta+d)})$;
See Lemma~\ref{lem:pi_rate} below.  For infinitely differentiable conditional mean functions with $\beta\to\infty$, the convergence rate becomes $|\pi_t-\pi^*|=\widetilde{O}_p(t^{-1/2})$, which aligns with Lemma 4.2 of \textcolor{blue}{Neopane et al.} (\textcolor{blue}{2025a}) who studied a simplified model without covariates.

\section{Proofs of theoretical results}
\label{supsec:proof}

\subsection{Proof of Theorem~\ref{thm:variance}}

Recall that $\hat{\tau}_T=T^{-1}\sum_{t=1}^{T}\phi_t$, and $u_t=\phi_t-\tau_0.$ We first show the conditional mean and conditional variance of $u_t$ given the history $\cF_{t-1}.$

\begin{lemma}
    \label{lem:condition_var}
    Suppose that the conditions of Theorem~\ref{thm:variance} hold. Then, we have: (i) $\eE(u_t|\cF_{t-1})=0;$ (ii) the conditional variance of $u_t$ satisfies that
    $$
    \begin{aligned}
        \eE(u_t^2|\cF_{t-1})
        =&\frac{\sigma_1^2}{\pi_t}+\frac{\sigma_0^2}{1-\pi_t}+\eE_{(\bsX,\bsS)}\{\tau(\bX,\bS)-\tau_0\}^2\\
        &+\frac{1-\pi_t}{\pi_t}\eE_{(\bsX_t,\bsS_t)}\big[\{\mu_1(\bX_t,\bS_t)-\hat{\mu}_{1,t-1}(\bX_t,\bS_t)\}^2|\cF_{t-1}\big]\\
        &+\frac{\pi_t}{1-\pi_t}\eE_{(\bsX_t,\bsS_t)}\big[\{\mu_0(\bX_t,\bS_t)-\hat{\mu}_{0,t-1}(\bX_t,\bS_t)\}^2|\cF_{t-1}\big]\\
        &+2\eE_{(\bsX_t,\bsS_t)}\big[\{\mu_1(\bX_t,\bS_t)-\hat{\mu}_{1,t-1}(\bX_t,\bS_t)\}\{\mu_0(\bX_t,\bS_t)-\hat{\mu}_{0,t-1}(\bX_t,\bS_t)\}|\cF_{t-1}\big].
    \end{aligned}
    $$
\end{lemma}
\begin{proof}
    (i) By the definitions of $\phi_t$ in \eqref{eq:single-stage-ATE} and $\pi_t=\eP(Z_t=1|\cF_{t-1})$, and the iterated expectation we can show that
    $$
    \begin{aligned}
    &\eE(\phi_t|\cF_{t-1},\bX_t,\bS_t)\\
    =&\frac{\eE\{I(Z_t=1)|\cF_{t-1}\}\eE\{Y_t(1)-\hat{\mu}_{1,t-1}(\bX_t,\bS_t)|\cF_{t-1},\bX_t,\bS_t\}}{\pi_t}\\
    &-\frac{\eE\{I(Z_t=0)|\cF_{t-1}\}\eE\{Y_t(0)-\hat{\mu}_{0,t-1}(\bX_t,\bS_t)|\cF_{t-1},\bX_t,\bS_t\}}{1-\pi_t}\\
    &+\hat{\mu}_{1,t-1}(\bX_t,\bS_t)-\hat{\mu}_{0,t-1}(\bX_t,\bS_t)\\
    =&\frac{\pi_t\eE\{Y_t(1)-\hat{\mu}_{1,t-1}(\bX_t,\bS_t)|\cF_{t-1},\bX_t,\bS_t\}}{\pi_t}-\frac{(1-\pi_t)\eE\{Y_t(0)-\hat{\mu}_{0,t-1}(\bX_t,\bS_t)|\cF_{t-1},\bX_t,\bS_t\}}{1-\pi_t}\\
    &+\hat{\mu}_{1,t-1}(\bX_t,\bS_t)-\hat{\mu}_{0,t-1}(\bX_t,\bS_t)\\
    =&\eE\{Y_t(1)-Y_t(0)|\bX_t,\bS_t\}=\tau(\bX_t,\bS_t),
    \end{aligned}
    $$
    which implies that
    $$
    \eE(u_t|\cF_{t-1})=\eE_{(\bsX_t,\bsS_t)}\{\eE(\phi_t|\cF_{t-1},\bX_t,\bS_t)\}-\tau_0=\eE_{(\bsX_t,\bsS_t)}\{\tau(\bX_t,\bS_t)\}-\tau_0=0.
    $$
    (ii) By the definition of $\phi_t$ in Section~\ref{subsec:procedure}, we have
    $$
    \begin{aligned}
    u_t=&\frac{I(Z_t=1)\{Y_t-\hat{\mu}_{1,t-1}(\bX_t,\bS_t)\}}{\pi_t}-\frac{I(Z_t=0)\{Y_t-\hat{\mu}_{0,t-1}(\bX_t,\bS_t)\}}{1-\pi_t}\\
    &+\hat{\mu}_1(\bX_t,\bS_t)-\hat{\mu}_0(\bX_t,\bS_t)-\tau_0\\
    =&\frac{I(Z_t=1)\{Y_t-\mu_1(\bX_t,\bS_t)\}}{\pi_t}+\frac{\{I(Z_t=1)-\pi_t\} \cdot \{\mu_1(\bX_t,\bS_t)-\hat{\mu}_{1,t-1}(\bX_t,\bS_t)\}}{\pi_t}\\
    &-\frac{I(Z_t=0)\{Y_t-\mu_0(\bX_t,\bS_t)\}}{1-\pi_t}-\frac{\{I(Z_t=0)-1+\pi_t\} \cdot \{\mu_0(\bX_t,\bS_t)-\hat{\mu}_{0,t-1}(\bX_t,\bS_t)\}}{1-\pi_t}\\
    &+\{\mu_1(\bX_t,\bS_t)-\mu_0(\bX_t,\bS_t)-\tau_0\}\\
    =:&R_{1,t}+R_{2,t}+R_{3,t}+R_{4,t}+R_{5,t}.
    \end{aligned}
    $$
    Then, it can be respectively shown that:\\
    (a) $\eE(R_{1,t}^2|\cF_{t-1})=\pi_t^{-1}\eE_{(\bsX_t,\bsS_t)}\big[\eE\{(Y_t(1)-\mu_1(\bX_t,\bS_t))^2|\bX_t,\bS_t\}\big]=\pi_t^{-1}\sigma_1^2.$\\
    (b) $\eE(R_{2,t}^2|\cF_{t-1})=\pi_t^{-2}(\pi_t-\pi_t^2)\eE_{(\bsX_t,\bsS_t)}\big[\{\mu_1(\bX_t,\bS_t)-\hat{\mu}_{1,t-1}(\bX_t,\bS_t)\}^2|\cF_{t-1}\big]\\=\pi_t^{-1}(1-\pi_t)\eE_{(\bsX_t,\bsS_t)}\big[\{\mu_1(\bX_t,\bS_t)-\hat{\mu}_{1,t-1}(\bX_t,\bS_t)\}^2|\cF_{t-1}\big]$.\\
    (c) $\eE(R_{3,t}^2|\cF_{t-1})=(1-\pi_t)^{-1}\eE_{(\bsX_t,\bsS_t)}\big[\eE\{(Y_t(0)-\mu_0(\bX_t,\bS_t))^2|\bX_t,\bS_t\}\big]=(1-\pi_t)^{-1}\sigma_0^2.$\\
    (d) $\eE(R_{4,t}^2|\cF_{t-1})=(1-\pi_t)^{-2}(\pi_t-\pi_t^2)\eE_{(\bsX_t,\bsS_t)}\big[\{\mu_0(\bX_t,\bS_t)-\hat{\mu}_{0,t-1}(\bX_t,\bS_t)\}^2|\cF_{t-1}\big]\\=(1-\pi_t)^{-1}\pi_t\eE_{(\bsX_t,\bsS_t)}\big[\{\mu_0(\bX_t,\bS_t)-\hat{\mu}_{0,t-1}(\bX_t,\bS_t)\}^2|\cF_{t-1}\big]$.\\
    (e) $\eE(R_{5,t}^2|\cF_{t-1})=\eE_{(\bsX,\bsS)}\{\tau(\bX,\bS)-\tau_0\}^2.$\\
    (f) $\eE(R_{1,t}R_{2,t}|\cF_{t-1})=0$ and $\eE(R_{1,t}R_{4,t}|\cF_{t-1})=0$ by $\eE\{Y_t(1)-\mu_1(\bX_t,\bS_t)|\bX_t,\bS_t\}=0$.\\
    (g) $\eE(R_{3,t}R_{2,t}|\cF_{t-1})=0$ and $\eE(R_{3,t}R_{4,t}|\cF_{t-1})=0$ by $\eE\{Y_t(0)-\mu_0(\bX_t,\bS_t)|\bX_t,\bS_t\}=0$.\\
    (h) $\eE(R_{1,t}R_{3,t}|\cF_{t-1})=0$ by $\eP\{I(Z_t=1)I(Z_t=0)|\cF_{t-1}\}=0.$\\
    (l) $\eE(R_{5,t}R_{1,t}|\cF_{t-1})=0,\eE(R_{5,t}R_{2,t}|\cF_{t-1})=0,\eE(R_{5,t}R_{3,t}|\cF_{t-1})=0,$ and $\eE(R_{5,t}R_{4,t}|\cF_{t-1})=0$ by $\eE_{(\bsX_t,\bsS_t)}\{\mu_1(\bX_t,\bS_t)-\mu_0(\bX_t,\bS_t)-\tau_0\}=0.$\\
    (m) $\eE(R_{2,t}R_{4,t}|\cF_{t-1})=\{\pi_t(1-\pi_t)\}^{-1}\pi_t(1-\pi_t)\eE_{(\bsX_t,\bsS_t)}\big[\{\mu_1(\bX_t,\bS_t)-\hat{\mu}_{1,t-1}(\bX_t,\bS_t)\}\{\mu_0(\bX_t,\bS_t)-\hat{\mu}_{0,t-1}(\bX_t,\bS_t)\}|\cF_{t-1}\big]=\eE_{(\bsX_t,\bsS_t)}\big[\{\mu_1(\bX_t,\bS_t)-\hat{\mu}_{1,t-1}(\bX_t,\bS_t)\}\{\mu_0(\bX_t,\bS_t)-\hat{\mu}_{0,t-1}(\bX_t,\bS_t)\}|\cF_{t-1}\big]$.
    Combining above, we have
    $$
    \begin{aligned}
        \eE(u_t^2|\cF_{t-1})=&\frac{\sigma_1^2}{\pi_t}+\frac{\sigma_0^2}{1-\pi_t}+\eE_{(\bsX,\bsS)}\{\tau(\bX,\bS)-\tau_0\}^2\\
        &+\frac{1-\pi_t}{\pi_t}\eE_{(\bsX_t,\bsS_t)}\big[\{\mu_1(\bX_t,\bS_t)-\hat{\mu}_{1,t-1}(\bX_t,\bS_t)\}^2|\cF_{t-1}\big]\\
        &+\frac{\pi_t}{1-\pi_t}\eE_{(\bsX_t,\bsS_t)}\big[\{\mu_0(\bX_t,\bS_t)-\hat{\mu}_{0,t-1}(\bX_t,\bS_t)\}^2|\cF_{t-1}\big]\\
        &+2\eE_{(\bsX_t,\bsS_t)}\big[\{\mu_1(\bX_t,\bS_t)-\hat{\mu}_{1,t-1}(\bX_t,\bS_t)\}\{\mu_0(\bX_t,\bS_t)-\hat{\mu}_{0,t-1}(\bX_t,\bS_t)\}|\cF_{t-1}\big].
    \end{aligned}
    $$
    which is the desired result.
\end{proof}
\begin{remark}
    \label{rmk:ell_2}
    From the proof above, we have 
    $$
    \ell_2(a_t)=\eE(R_{2,t}^2|\cF_{t-1})+\eE(R_{4,t}^2|\cF_{t-1})+2\eE(R_{2,t}R_{4,t}|\cF_{t-1})=\eE\{(R_{2,t}+R_{4,t})^2|\cF_{t-1}\}\ge0,
    $$ 
    where $\ell_2(a_t)$ is defined in Section~\ref{subsec:regret}. If $\hat{\mu}_{0,t-1}=\mu_0$ and $\hat{\mu}_{1,t-1}=\mu_1$ a.s., then $R_{2,t}=R_{4,t}=0$ a.s., which further implies that $\ell_2(a_t)=0$. 
\end{remark}

Now we are ready to show Theorem~\ref{thm:variance}.\\
(i) According to Lemma~\ref{lem:condition_var}(i), we have 
$$
\eE(\hat{\tau}_T)=T^{-1}\sum_{t=1}^{T}\eE\{\eE(\phi_t|\cF_{t-1})\}=\tau_0,
$$ 
which implies that $\hat{\tau}_T$ is unbiased.\\
(ii) For the variance, note that $T\cdot V_T=T^{-1}\eE\big(\sum_{t=1}^{\T}u_t^2+2\sum_{t=1}^{\T}\sum_{s=1}^{t-1}u_tu_s\big)$. Since $u_s\in\cF_{t-1}$ for $s<t$, by iterated expectations, it can be shown that 
$$
\eE(u_tu_s)=\eE\{u_s\eE(u_t|\cF_{t-1})\}=0,\quad \forall s<t.
$$
Then, by combining the result in Lemma~\ref{lem:condition_var}(ii), the proof is completed.

\subsection{Proof of Proposition~\ref{propos:efficiency_bound}}
The desired result follows immediately from Theorem~\ref{thm:variance} by letting $\hat{\mu}_{1,t-1}(\cdot,\cdot)=\mu_1(\cdot,\cdot),\hat{\mu}_{0,t-1}(\cdot,\cdot)=\mu_0(\cdot,\cdot),$ and $\pi_t=\pi^*$ for $t\in[T].$

\subsection{Proof of Proposition~\ref{propos:pi}}

\begin{lemma}
    \label{lem:N_t_lower}
    Suppose that Conditions \ref{cond:sutva}--\ref{cond:kernel} hold. Recall that $N_{1,t}=\sum_{j=1}^{t}I(Z_j=1)$ and $N_{0,t}=t-N_{1,t}$. Then, we have $N_{1,t}=\Omega(t^{1-\eta})$ a.s., and $N_{0,t}=\Omega(t^{1-\eta})$ a.s.
\end{lemma}
\begin{proof}
    Note that $N_{1,t}=\sum_{j=1}^{t}I(Z_j=1)$, where $\eP(Z_t=1|\cF_{t-1})=\pi_t$. Define the event $\cE_t=\big\{N_{1,t}\le\sum_{j=1}^{t}\pi_j/2\big\}$. By using Lemma~\ref{lem:count}(i) below, we have $\eP(\cE_t)\le\exp\big(-3\sum_{j=1}^{t}\pi_j/28\big)$. Since $\pi_j\ge (1/2)\cdot j^{-\eta}$ for each $j\ge1$ by \eqref{eq:clip}, it can be shown that $\sum_{j=1}^{t}\pi_j=\Theta(t^{1-\eta})$. Then, it can be seen that $\sum_{t=1}^{\infty}\eP(\cE_t)<\infty$, which implies that $N_{1,t}=\Omega(t^{1-\eta})$ a.s. by Borel-Cantelli Lemma. By using similar arguments, we can show that $N_{0,t}=\Omega(t^{1-\eta})$ a.s., by noticing that $\pi_j\le 1-(1/2)\cdot j^{-\eta}$ for each $j\ge1$.
\end{proof}

\begin{lemma}
    \label{lem:mu_hat_asy}
    Suppose that Conditions \ref{cond:sutva}--\ref{cond:kernel} hold. Then, $$\eE_{(\bsX_{t+1},\bsS_{t+1})}\big[\{\mu_z(\bX_{t+1},\bS_{t+1})-\hat{\mu}_{z,t}(\bX_{t+1},\bS_{t+1})\}^2|\cF_t\big]={O_p(N_{z,t}^{-2\beta/\kappa})}$$ for each $z\in\{0,1\}.$
\end{lemma}
\begin{proof}
    Without loss of generality, we only prove the case $z=0$. By the definition of $\hat{\mu}_{0,t}$ in \eqref{eq:mu_hat}, we have $\hat{\mu}_{0,t}(\bX_{t+1},\bS_{t+1})-\mu_0(\bX_{t+1},\bS_{t+1})=\hat{m}_{Y,0,t}(\bX_{t+1})-m_{Y,0}(\bX_{t+1})+\{s_{t+1}(0)-m_{S,0}(\bX_{t+1})\}(\hat{\gamma}_{0,t}-\gamma_0)+\gamma_0\{m_{S,0}(\bX_{t+1})-\hat{m}_{S,0,t}(\bX_{t+1})\}+(\gamma_0-\hat{\gamma}_{0,t})\{\hat{m}_{S,0,t}(\bX_{t+1})-m_{S,0}(\bX_{t+1})\}$. This implies that
    \begin{equation}
        \label{eq:hat_mu_deco_asy}
        \begin{aligned}
            &\eE_{(\bsX_{t+1},\bsS_{t+1})}\big[\{\mu_0(\bX_{t+1},\bS_{t+1})-\hat{\mu}_{0,t}(\bX_{t+1},\bS_{t+1})\}^2|\cF_t\big]\\
            \le&4\gamma_0^2\eE_{\bsX_{t+1}}\big[\{\hat{m}_{S,0,t}(\bX_{t+1})-m_{S,0}(\bX_{t+1})\}^2|\cF_t\big]\\
            &+4(\hat{\gamma}_{0,t}-\gamma_0)^2\eE_{\bsX_{t+1}}\big[\{\hat{m}_{S,0,t}(\bX_{t+1})-m_{S,0}(\bX_{t+1})\}^2|\cF_t\big]\\
            &+4\eE_{\bsX_{t+1}}\big[\{\hat{m}_{Y,0,t}(\bX_{t+1})-m_{Y,0}(\bX_{t+1})\}^2|\cF_t\big]+4\sigma_{S,0}^2(\hat{\gamma}_{0,t}-\gamma_0)^2.
        \end{aligned}
    \end{equation}
    By applying Proposition 1.13 of 
    \textcolor{blue}{Tsybakov} (\textcolor{blue}{2009}) or Theorem 5.2 of \textcolor{blue}{Gy{\"o}rfi et al.} (\textcolor{blue}{2006}), which provide upper bounds on the mean integrated squared error of the Nadaraya-Watson regression, we can derive that there exists some constant $M_S$, independent of $t$ and $h_{S,0,t}$, such that 
    \begin{equation}
        \label{eq:m_S_rate}
        \eE\left[\int_{\cX}\{\hat{m}_{S,0,t}(\bx)-m_{S,0}(\bx)\}^2{\rm d}\bx\right]\le M_S(h_{S,0,t}^{2\beta}+t^{-1}h_{S,0,t}^{-d}).
    \end{equation}
    Then, it follows from Condition~\ref{cond:X} and the Markov inequality that $\eE_{\bsX_{t+1}}\big[\{\hat{m}_{S,0,t}(\bX_{t+1})-m_{S,0}(\bX_{t+1})\}^2|\cF_t\big]=O_p(h_{S,0,t}^{2\beta}+t^{-1}h_{S,0,t}^{-d})$. Similarly, $\eE_{\bsX_{t+1}}\big[\{\hat{m}_{Y,0,t}(\bX_{t+1})-m_{Y,0}(\bX_{t+1})\}^2|\cF_t\big]=O_p(h_{Y,0,t}^{2\beta}+N_{0,t}^{-1}h_{Y,0,t}^{-d})$. These results further imply that
    $$
    \begin{aligned}
        &\eE_{\bsX_{t+1}}\big[\{\hat{m}_{S,0,t}(\bX_{t+1})-m_{S,0}(\bX_{t+1})\}^2|\cF_t\big]=O_p(t^{-2\beta/\kappa}),\\
        &\eE_{\bsX_{t+1}}\big[\{\hat{m}_{Y,0,t}(\bX_{t+1})-m_{Y,0}(\bX_{t+1})\}^2|\cF_t\big]=O_p(N_{0,t}^{-2\beta/\kappa}),
    \end{aligned}
    $$
    since we set {$h_{S,0,t}\asymp t^{-1/\kappa}$ and $h_{Y,0,t}\asymp N_{0,t}^{-1/\kappa}$ with $\kappa>2d$, and in particular, $\kappa=2\beta+d$ when $d<2\beta$.} 

    For the linear coefficient $\gamma_0,$ for each $j\in G_{0,t}=\{j\in[t]:Z_j=0\}$, denote $\widetilde{Y}_j=Y_j-\hat{m}_{Y,0,t}(\bX_j),$ $\widetilde{S}_j(0)=S_j(0)-\hat{m}_{S,0,t}(\bX_j),$ $\epsilon_{Y,0}(\bX_j)=\hat{m}_{Y,0,t}(\bX_j)-m_{Y,0}(\bX_j),$ and $\epsilon_{S,0}(\bX_j)=\hat{m}_{S,0,t}(\bX_j)-m_{S,0}(\bX_j)$, where we omit the subscript $t$ for simplicity. By the definition of $\hat{\gamma}_{0,t}$ in \eqref{eq:gamma_hat}, it follows that
    $$
    \hat{\gamma}_{0,t}-\gamma_0=\frac{\sum_{j\in G_{0,t}}\{\epsilon_{S,0}(\bX_j)+\varepsilon_{S,j}(0)\}\varepsilon_j(0)}{\sum_{j\in G_{0,t}}\widetilde{S}_j^2(0)}-\frac{\sum_{j\in G_{0,t}}\widetilde{S}_j(0)\{\epsilon_{Y,0}(\bX_j)+\epsilon_{S,0}(\bX_j)\gamma_0\}}{\sum_{j\in G_{0,t}}\widetilde{S}_j^2(0)}.
    $$
    Noticing that $\var\{\varepsilon_{S,t}(0)|\bX_t\}=\sigma_{S,0}^2>0$, we have $N_{0,t}^{-1}\sum_{j\in G_{0,t}}\widetilde{S}_j^2(0)$ bounded away from zero with probability tending to 1 by the Weak Law of Large Numbers. By Condition~\ref{cond:X} and \eqref{eq:m_S_rate}, $\eE\epsilon_{S,0}^2(\bX_j)=O(t^{-2\beta/\kappa})$ for each $j\in[t],$ which, together with the Markov inequality, implies that $N_{0,t}^{-1}\sum_{j\in G_{0,t}}\epsilon_{S,0}(\bX_j)\varepsilon_j(0)=O_p(t^{-\beta/\kappa})$. Furthermore, since $\eE\{\varepsilon_{S,j}(0)\varepsilon_j(0)|\bX_t\}=0$ and $\eE\{\varepsilon_{S,j}^2(0)\varepsilon_j^2(0)|\bX_t\}=\sigma_{S,0}^2\sigma_0^2<\infty$, we have $N_{0,t}^{-1}\sum_{j\in G_{0,t}}\varepsilon_{S,j}(0)\varepsilon_j(0)=O_p(N_{0,t}^{-1/2})$ by the Central Limit Theorem. Thus, we have shown that the numerator of the first term in the decomposition above satisfies that $N_{0,t}^{-1}\sum_{j\in G_{0,t}}\{\epsilon_{S,0}(\bX_j)+\varepsilon_{S,j}(0)\}\varepsilon_j(0)=O_p(N_{0,t}^{-\beta/\kappa})$, where we also used the fact that $t> N_{0,t}$ and $\kappa>2\beta$. Similar arguments can apply to the second term of the decomposition. Thus, $|\hat{\gamma}_{0,t}-\gamma_0|=O_p(N_{0,t}^{-\beta/\kappa}).$
    
    Combining above, we have $\eE_{(\bsX_{t+1},\bsS_{t+1})}\big[\{\mu_0(\bX_{t+1},\bS_{t+1})-\hat{\mu}_{0,t}(\bX_{t+1},\bS_{t+1})\}^2|\cF_t\big]=O_p(N_{0,t}^{-2\beta/\kappa})$ by \eqref{eq:hat_mu_deco_asy}. The proof is completed.
\end{proof}

Now we are ready to show Proposition~\ref{propos:pi}. We first show that $\hat{\sigma}_{0,t}^2\to_p\sigma_0^2$ and $\hat{\sigma}_{1,t}^2\to_p\sigma_1^2$ as $t\to\infty$. Recall that $G_{0,t}=\{j\in[t]:Z_j=0\}$. By \eqref{eq:sigma_hat}, it follows that
$$
\begin{aligned}
    \hat{\sigma}_{0,t}^2=&\frac{1}{N_{0,t}}\sum_{j\in G_{0,t}}\{Y_j-\hat{\mu}_{0,t}(\bX_j,\bS_j)\}^2\\
    =&\frac{1}{N_{0,t}}\sum_{j\in G_{0,t}}\{Y_j(0)-\mu_0(\bX_j,\bS_j)\}^2
    +\frac{1}{N_{0,t}}\sum_{j\in G_{0,t}}\{\hat{\mu}_{0,t}(\bX_j,\bS_j)-\mu_0(\bX_j,\bS_j)\}^2\\
    &+\frac{2}{N_{0,t}}\sum_{j\in G_{0,t}}\{Y_j(0)-\mu_0(\bX_j,\bS_j)\}\{\mu_0(\bX_j,\bS_j)-\hat{\mu}_{0,t}(\bX_j,\bS_j)\}\\
    =:&\widetilde{R}_{1,t}+\widetilde{R}_{2,t}+\widetilde{R}_{3,t}.
\end{aligned}
$$

Note that $\varepsilon_j(0)=Y_j(0)-\mu_0(\bX_j,\bS_j)$ and $\var\{\varepsilon_j(0)|\bX_j,\bS_j\}=\sigma_0^2$. By Condition~\ref{cond:sutva} and Lemma~\ref{lem:N_t_lower}, it follows from the SLLN for i.i.d. random variables that $\widetilde{R}_{1,t}\to_{a.s.}\sigma_0^2$ as $t\to\infty$.
{For the second term: (i) since $th_{S,0,t}^{2d}/\log t\to\infty$ and Conditions~\ref{cond:mean_fun}--\ref{cond:varep} hold, Theorem 1 of \textcolor{blue}{Qian and Yang} (\textcolor{blue}{2016}) implies that $\sup_{\bx\in\cS}|\hat{m}_{S,0,t}(\bx)-m_{S,0}(\bx)|=o_p(1)$; (ii) likewise, combined with Lemma~\ref{lem:N_t_lower}, $\sup_{\bx\in\cS}|\hat{m}_{Y,0,t}(\bx)-m_{Y,0}(\bx)|=o_p(1)$; (iii) $|\hat{\gamma}_{0,t}-\gamma_0|=O_p(N_{0,t}^{-\beta/\kappa})$ follows from the proof of Lemma~\ref{lem:mu_hat_asy}; (iv) $\max_{j\in G_{0,t}}\varepsilon_{S,j}^2(0)=O_p(\log N_{0,t})$ by Condition~\ref{cond:varep}. Combining above, using the decomposition similar to \eqref{eq:hat_mu_deco_asy}, we can show that $\widetilde{R}_{2,t}=o_p(1)$.}
The third term can be proved that $\widetilde{R}_{3,t}=o_p(1)$ by using the Cauchy-Schwarz inequality. Thus, $\hat{\sigma}_{0,t}^2\to_{p}\sigma_0^2$ as $t\to\infty$. The second argument $\hat{\sigma}_{1,t}^2\to_{p}\sigma_1^2$ as $t\to\infty$ can be proved in a similar manner. Moreover, it can be seen that $\pi_t\to_p\widetilde{\pi}_t$. By the definitions of the Neyman allocation and the initial allocation, that $\pi^*=\sigma_1/(\sigma_1+\sigma_0)$ and $\widetilde{\pi}_t=\hat{\sigma}_{1,t-1}/(\hat{\sigma}_{1,t-1}+\hat{\sigma}_{0,t-1})$, we can immediately obtain that $\pi_t\to_{p}\pi^*$ as $t\to\infty$ by using the Continuous Mapping Theorem.

\subsection{Proof of Theorem~\ref{thm:asy_nor}}

The strong consistency of the proposed estimator $\hat{\tau}_T=T^{-1}\sum_{t=1}^{T}\phi_t$ follows from the Strong Law of Large Numbers (SLLN) for martingale differences. 

\begin{lemma}
    \label{lem:SLLN}
    (SLLN for martingale differences, Theorem 2.18 of \textcolor{blue}{Hall and Heyde} (\textcolor{blue}{1980})) Let $\{X_t\}_{t\ge1}$ be a martingale sequence w.r.t. filtration $\{\cF_t\}_{t\ge1}$. For $1\le p\le2,n^{-1}\sum_{t=1}^{n}X_t\to0$ as $n\to\infty$ almost surely on the set $\big\{\sum_{t=1}^{\infty}t^{-p}\eE(|X_t|^p|\cF_{t-1})<\infty\big\}$. 
\end{lemma}

The following Lemma~\ref{lem:mclt} generalizes the Lindeberg-Feller central limit theorem for a martingale difference sequence. The statement and proof of the Martingale Central Limit Theorem can be found in Chapter 3 of \textcolor{blue}{Hall and Heyde} (\textcolor{blue}{1980}).

\begin{lemma}
    \label{lem:mclt}
    (Martingale Central Limit Theorem) For each $n\ge1,$ let $\{S_{n,i}\}_{1\le i\le k_n}$ be a square integrable martingale sequence w.r.t. filtration $\{\cF_{n,i}\}_{1\le i\le k_n}$ with corresponding martingale difference $X_{n,i}=S_{n,i}-S_{n,i-1}$ and conditional variance $\sigma_{n,i}^2=\eE(X_{n,i}^2|\cF_{n,i-1})$. Assume that: (i) $\sum_{i=1}^{k_n}\sigma_{n,i}^2\to_p\sigma^2$ as $n\to\infty$ for some $\sigma^2>0;$ and (ii) the conditional Lindeberg condition holds, that is, $\forall\epsilon>0,\sum_{i=1}^{k_n}\eE\{X_{n,i}^2I(|X_{n,i}|>\epsilon)|\cF_{t-1}\}\to_p0$ as $n\to\infty.$ Then, $S_{n,k_n}\to_d\cN(0,\sigma^2)$ as $n\to\infty.$
\end{lemma}

To apply Lemmas~\ref{lem:SLLN} and \ref{lem:mclt} to show Theorem~\ref{thm:asy_nor}, we establish some basic results.

\begin{lemma}
    \label{lem:mart}
    Suppose that the conditions of Theorem~\ref{thm:asy_nor}(i) hold. Then, $\{u_t\}_{t\ge1}$ forms a square integrable martingale difference sequence w.r.t. $\{\cF_t\}_{t\ge1}$. 
\end{lemma}
\begin{proof}
    By Lemma~\ref{lem:condition_var}(i), we have $\eE(u_t|\cF_{t-1})=0.$ It suffices to show that $\eE(u_t^2)<\infty.$ Note that $\pi_t\ge(1/2)\cdot t^{-\eta},1-\pi_t\ge(1/2)\cdot t^{-\eta}$ by CLIP defined in \eqref{eq:clip}. Combining \eqref{eq:hat_mu_deco_asy} and \eqref{eq:m_S_rate}, Condition~\ref{cond:X}, and the almost surely boundedness of $\hat{\gamma}_{z,t-1}$ as assumed in Theorem~\ref{thm:asy_nor}, we obtain that there exists some constant $M_{\mu}>0$ such that $\eE\{\mu_z(\bX_t,\bS_t)-\hat{\mu}_{z,t-1}(\bX_t,\bS_t)\}^2\le M_{\mu}$ for all $t\ge1$ and each $z\in\{0,1\}$. Together with Lemma~\ref{lem:condition_var}, these imply that $\eE(u_t^2)<\infty$, which completes the proof.
\end{proof}

\begin{lemma}
    \label{lem:con_var_as}
    Suppose that the conditions of Theorem~\ref{thm:asy_nor} hold. Then, $\eE(u_t^2|\cF_{t-1})\to_p\varsigma_*^2$ as $t\to\infty$, where $\varsigma_*^2$ is defined in \eqref{eq:ney_var}.
\end{lemma}
\begin{proof}
    By Proposition~\ref{propos:pi} and the Continuous Mapping Theorem, we have 
    $$
    \frac{\sigma_1^2}{\pi_t}+\frac{\sigma_0^2}{1-\pi_t}+\eE_{(\bsX,\bsS)}\{\tau(\bX,\bS)-\tau_0\}^2\to_{p}\frac{\sigma_1^2}{\pi^*}+\frac{\sigma_0^2}{1-\pi^*}+\eE_{(\bsX,\bsS)}\{\tau(\bX,\bS)-\tau_0\}^2=\varsigma_*^2
    $$
    as $t\to\infty.$ Then, define $\widetilde{R}_{4,t}=\eE(u_t^2|\cF_{t-1})-\pi_t^{-1}\sigma_1^2-(1-\pi_t)^{-1}\sigma_0^2-\eE_{(\bsX,\bsS)}\{\tau(\bX,\bS)-\tau_0\}^2,$ which corresponds to the sum of the last three terms in the decomposition of $\eE(u_t^2|\cF_{t-1})$ as given in Lemma~\ref{lem:condition_var}(ii). It follows from the Cauchy-Schwarz inequality that
    $$
    \begin{aligned}
        \widetilde{R}_{4,t}\le&\frac{1}{\pi_t}\eE_{(\bsX_t,\bsS_t)}\big[\{\mu_1(\bX_t,\bS_t)-\hat{\mu}_{1,t-1}(\bX_t,\bS_t)\}^2|\cF_{t-1}\big]\\
        &+\frac{1}{1-\pi_t}\eE_{(\bsX_t,\bsS_t)}\big[\{\mu_0(\bX_t,\bS_t)-\hat{\mu}_{0,t-1}(\bX_t,\bS_t)\}^2|\cF_{t-1}\big]\\
        =&O_p(t^{\eta}\cdot N_{0,t-1}^{-2\beta/\kappa})+ O_p(t^{\eta}\cdot N_{1,t-1}^{-2\beta/\kappa}) \\
        =&O_p(t^{\eta-2\beta(1-\eta)/\kappa})=o_p(1),
    \end{aligned}
    $$
    where the second line follows from CLIP such that $\pi_t\ge(1/2)\cdot t^{-\eta},1-\pi_t\ge(1/2)\cdot t^{-\eta}$ and Lemma~\ref{lem:mu_hat_asy}, and the third line follows from Lemma~\ref{lem:N_t_lower} and the assumption $\eta<2\beta/(2\beta+\kappa).$ Combining above, it follows that $\eE(u_t^2|\cF_{t-1})\to_p \varsigma_*^2$ as $t\to\infty$, which completes the proof. 
\end{proof}

Now we are ready to show Theorem~\ref{thm:asy_nor}. \\
(i) By Lemma~\ref{lem:mart}, $\{u_t\}_{t\ge1}$ forms a square integrable martingale difference sequence w.r.t. $\{\cF_t\}_{t\ge1}$. 
By using Lemma~\ref{lem:lyapunov}(i) below with $\theta=0,$
there exists some constant $M_u>0$ such that $\eE(u_t^2|\cF_{t-1})\le M_ut^{\eta}\log t$ almost surely. Then, since $\eta<1,$ it follows that $\sum_{t=1}^{\infty}t^{-2}\eE(u_t^2|\cF_{t-1})<\infty$ a.s. By using the SLLN for martingale differences in Lemma~\ref{lem:SLLN}, $T^{-1}\sum_{t=1}^{T}u_t\to_{a.s.}0$, which implies that $\hat{\tau}_T\to_{a.s.}\tau_0$ as $T\to\infty.$ \\
(ii) To show $\hat{\tau}_T-\tau_0=T^{-1}\sum_{t=1}^{T}u_t$ is asymptotically normal, we apply the Martingale Central Limit Theorem in Lemma~\ref{lem:mclt}. Let $V_{T,t}=\var(T^{-1/2}u_t|\cF_{t-1})=T^{-1}\eE(u_t^2|\cF_{t-1})$. Following similar arguments as in the proof of Theorem 2 in \textcolor{blue}{Cook et al.} (\textcolor{blue}{2024}), we can show that $\eE(u_t^2|\cF_{t-1})$ is uniformly integrable. Together with the result $\eE(u_t^2|\cF_{t-1})\to_p\varsigma_*^2$ as $t\to\infty$ from Lemma~\ref{lem:con_var_as}, this implies that $\eE(u_t^2|\cF_{t-1})\to\varsigma_*^2$ in $L^1$, i.e., $\eE\big|\eE(u_t^2|\cF_{t-1})-\varsigma_*^2\big|\to0$ as $t\to\infty$. By the Markov inequality and Toeplitz Lemma, for any $\delta>0,$
$$
\eP\left(\Big|\sum_{t=1}^{T}V_{T,t}-\varsigma_*^2\Big|>\delta\right)\le\frac{\sum_{t=1}^{T}\eE\big|\eE(u_t^2|\cF_{t-1})-\varsigma_*^2\big|}{T\delta}=o(1),
$$
which leads to the desired results that  $\sum_{t=1}^{T}V_{T,t}\to_p\varsigma_*^2$ as $T\to\infty.$ 

Therefore, it suffices to verify the conditional Lindeberg condition, namely, that for any $\epsilon>0$, we have $T^{-1}\sum_{t=1}^{T}\eE\{u_t^2I(|u_t|>\epsilon\sqrt{T})|\cF_{t-1}\}\to_p0$ as $T\to\infty.$ To see this, define $b_{T,t}=u_t^2I(|u_t|>\epsilon\sqrt{T})$. According to the proof of Lemma~\ref{lem:mart}, we have $\eE\{\mu_z(\bX_t,\bS_t)-\hat{\mu}_{z,t-1}(\bX_t,\bS_t)\}^2=O(1)$ for each $z\in\{0,1\}.$ It follows from the Markov inequality that, for each $t\in[T],$ $\eP(|u_t|>\epsilon\sqrt{T})\le\eE(u_t^2)/(\epsilon^2T)\to0$ as $T\to\infty,$ since $\pi_t^{-1}\le2 t^\eta$ and $(1-\pi_t)^{-1}\le2 t^\eta$. This further implies that $b_{T,t}=o_p(1)$ for each $t\in[T].$ Since $b_{T,t}\le u_t^2$ and $\eE(u_t^2)<\infty$, it follows from the Dominated Convergence Theorem that $\lim_{T\to\infty}\eE(b_{T,t})=\eE(\lim_{T\to\infty}b_{T,t})=0$. By the Markov inequality and Toeplitz Lemma, for any $\delta>0,$
$$
\eP\left\{T^{-1}\sum_{t=1}^{T}\eE(b_{T,t}|\cF_{t-1})>\delta\right\}\le\frac{\sum_{t=1}^{T}\eE(b_{T,t})}{T\delta }=o(1).
$$
Thus, the conditional Lindeberg condition holds, which further implies that $\sqrt{T}(\hat{\tau}_T-\tau_0)\to_d \cN(0,\varsigma_*^2)$ as $T\to\infty.$ The proof is completed.

\subsection{Proof of Theorem~\ref{thm:CI}}

\begin{lemma}
    \label{lem:lyapunov}
    Suppose that the conditions of Theorem~\ref{thm:CI} hold. Then: (i) for any $\theta\ge0,$ there exists some constant $M_u>0$ such that $\eE(|u_t|^{2+\theta}|\cF_{t-1})\le M_ut^{\eta(1+\theta)}(\log t)^{1+\theta/2}$ a.s.; (ii) for $\theta>2\eta/(1-2\eta)$, it follows that $\sum_{t=1}^{\infty}t^{-1-\theta/2}\eE(|u_t|^{2+\theta}|\cF_{t-1})<\infty$ a.s.
\end{lemma}
\begin{proof}
    (i) {Since sub-Gaussian random variables have finite moments of all orders, by Condition~\ref{cond:varep}, there exists some constant $M_{\varepsilon}>0$ such that 
    $\eE\big\{|\varepsilon_t(z)|^{2+\theta}|\bX_t,S_t(z)\big\}<M_{\varepsilon}$ and $\eE\big\{|\varepsilon_{S,t}(z)|^{2+\theta}|\bX_t\big\}<M_{\varepsilon}$ for all $t\ge1$ and each $z\in\{0,1\}$.} By the definition of $\phi_t$, we have
    $$
    \begin{aligned}
        \phi_t:=&\left[\frac{I(Z_t=1)\{Y_t-\mu_1(\bX_t,\bS_t)\}}{\pi_t}-\frac{I(Z_t=0)\{Y_t-\mu_0(\bX_t,\bS_t)\}}{1-\pi_t}\right]\\
        &+\left[\frac{I(Z_t=1)\{\mu_1(\bX_t,\bS_t)-\hat{\mu}_{1,t-1}(\bX_t,\bS_t)\}}{\pi_t}-\frac{I(Z_t=0)\{\mu_0(\bX_t,\bS_t)-\hat{\mu}_{0,t-1}(\bX_t,\bS_t)\}}{1-\pi_t}\right]\\
        &+\left[\hat{\mu}_{1,t-1}(\bX_t,\bS_t)-\hat{\mu}_{0,t-1}(\bX_t,\bS_t)\right]\\
        =:&\widetilde{R}_{5,t}+\widetilde{R}_{6,t}+\widetilde{R}_{7,t}.
    \end{aligned}
    $$
    Noticing that $|I(Z_t=z)\{Y_t-\mu_z(\bX_t,\bS_t)\}|^{2+\theta}=I(Z_t=z)|\varepsilon_t(z)|^{2+\theta}$, it follows from Condition~\ref{cond:sutva}(iii) and the fact $\pi_t=\eP(Z_t=1|\cF_{t-1})$ that
    $$
    \eE(|\widetilde{R}_{5,t}|^{2+\theta}|\cF_{t-1})\le\frac{\eE\big\{|\varepsilon_t(1)|^{2+\theta}|\bX_t,S_t(1)\big\}}{\pi_t^{1+\theta}}+\frac{\eE\big\{|\varepsilon_t(0)|^{2+\theta}|\bX_t,S_t(0)\big\}}{(1-\pi_t)^{1+\theta}}
        \le2^{2+\theta}M_{\varepsilon}t^{\eta(1+\theta)}
    $$
    almost surely. For the second term $\widetilde{R}_{6,t}$, considering that the Nadaraya-Watson estimator in \eqref{eq:m_S_hat}, it follows that, for each $z\in\{0,1\}$ and any $\bx\in\cX,$
    $$
    \hat{m}_{S,z,t-1}(\bx)-m_{S,z}(\bx)=\frac{\sum_{j=1}^{t-1}K_h(\bx-\bX_j)\{m_{S,z}(\bX_j)-m_{S,z}(\bx)\}}{\sum_{j=1}^{t-1}K_h(\bx-\bX_j)}+\frac{\sum_{j=1}^{t-1}K_h(\bx-\bX_j)\varepsilon_{S,j}(z)}{\sum_{j=1}^{t-1}K_h(\bx-\bX_j)}.
    $$
    If $\sum_{j=1}^{t}K_h(\bx-\bX_j)=0$, then $|\hat{m}_{S,z,t-1}(\bx)-m_{S,z}(\bx)|\le\sup_{\bx\in\cX}|m_{S,z}(\bx)|<\infty$ by Conditions~\ref{cond:mean_fun} and \ref{cond:X}. If $\sum_{j=1}^{t}K_h(\bx-\bX_j)>0,$ then by Conditions~\ref{cond:mean_fun} and \ref{cond:kernel}(i), we have
    $$
    \left|\frac{\sum_{j=1}^{t}K_h(\bx-\bX_j)\{m_{S,z}(\bX_j)-m_{S,z}(\bx)\}}{\sum_{j=1}^{t}K_h(\bx-\bX_j)}\right|
    \le\sqrt{d}Lh_{S,z,t-1}^{\beta}.
    $$
    By the sub-Gaussian tail conditions in Condition~\ref{cond:varep}, there exists some constant $M_{\varepsilon}'>0$ such that $\max_{1\le j\le t-1}|\varepsilon_{S,j}(z)|^{2+\theta}\le M_{\varepsilon}'(\log t)^{1+\theta/2}$ almost surely. Combining above, 
    $$
    \begin{aligned}
        &\eE_{\bsX_t}\big\{|\hat{m}_{S,z,t-1}(\bX_t)-m_{S,z}(\bX_t)|^{2+\theta}|\cF_{t-1}\big\}\\
        \le&\sup_{\bx\in\cX}|m_{S,z}(\bx)|^{2+\theta}+2^{1+\theta}(\sqrt{d}Lh_{S,z,t-1}^{\beta})^{2+\theta}+2^{1+\theta}M_{\varepsilon}'(\log t)^{1+\theta/2}
    \end{aligned}
    $$
    almost surely. Similar arguments apply to $\eE_{\bsX_t}\big\{|\hat{m}_{Y,z,t-1}(\bX_t)-m_{Y,z}(\bX_t)|^{2+\theta}|\cF_{t-1}\big\}$. Thus, using the decomposition in \eqref{eq:hat_mu_deco_asy} together with the almost surely boundedness of $\hat{\gamma}_{z,t-1}$ as assumed in Theorem~\ref{thm:asy_nor}, we obtain that 
    there exists some constant $M_{\mu}'>0$ such that $\eE(|\widetilde{R}_{6,t}|^{2+\theta}|\cF_{t-1})\le M_{\mu}'t^{\eta(1+\theta)}{(\log t)^{1+\theta/2}}$ almost surely. For the third term $\widetilde{R}_{7,t}$, using the basic inequality that $(a+b)^p\le 2^{p-1}(a^p+b^p)$ for $a,b\ge0$ and $p\ge1$, we have that $\eE|\mu_z(\bX_t,\bS_t)|^{2+\theta}\le 2^{1+\theta}\eE|m_{Y,z}(\bX_t)|^{2+\theta}+2^{1+\theta}|\gamma_z|^{2+\theta}\eE\big\{|\varepsilon_{S,t}(z)|^{2+\theta}|\bX_t\big\}\le2^{1+\theta}\sup_{\bx\in\cX}|m_{Y,z}(\bx)|^{2+\theta}+2^{1+\theta}|\gamma_z|^{2+\theta}M_{\varepsilon}$ is uniformly bounded, which further demonstrates that $\eE(|\widetilde{R}_{7,t}|^{2+\theta}|\cF_{t-1})\le 2^{1+\theta}\eE\big\{|\hat{\mu}_{z,t-1}(\bX_t,\bS_t)-\mu_z(\bX_t,\bS_t)|^{2+\theta}|\cF_{t-1}\big\}+2^{1+\theta}\eE|\mu_z(\bX_t,\bS_t)|^{2+\theta}\le M_{\mu}''(\log t)^{1+\theta/2}$ almost surely for some $M_{\mu}''>0$. Combining above, there exists some constant $M_u>0$ such that $\eE(|u_t|^{2+\theta}|\cF_{t-1})\le M_ut^{\eta(1+\theta)}(\log t)^{1+\theta/2}$ almost surely.
    
    \noindent (ii) Using the result in part (i), there exists some constant $M_u'>0$ such that 
    $$
    \begin{aligned}
        \sum_{t=1}^{\infty}\frac{\eE(|u_t|^{2+\theta}|\cF_{t-1})}{t^{1+\theta/2}}\le&\sum_{t=1}^{\infty}\frac{2^{1+\theta}\tau_0^{2+\theta}+2^{1+\theta}\eE(|\phi_t|^{2+\theta}|\cF_{t-1})}{t^{1+\theta/2}}\\
        \le&M_u'\sum_{t=1}^{\infty}t^{\eta(1+\theta)-1-\theta/2}(\log t)^{1+\theta/2}<\infty
    \end{aligned}
    $$
    almost surely, since $\eta<\theta/(2+2\theta)$. The proof is completed. 
\end{proof}

Now we are ready to show Theorem~\ref{thm:CI}. Let $v_t:=u_t^2-\eE(u_t^2|\cF_{t-1})\in\cF_t$. Noticing that $\eE(v_t|\cF_{t-1})=0$ and $\eE|v_t|\le 2\eE(u_t^2)<\infty$ by Lemma~\ref{lem:mart}, we have $\{v_t\}_{t\ge1}$ forms a martingale difference sequence w.r.t. $\{\cF_t\}_{t\ge1}$. By the definition of $\hat{\varsigma}_T^2$, it follows that
$$
\hat{\varsigma}_T^2=\frac{1}{T}\sum_{t=1}^{T}(\phi_t-\hat{\tau}_T)^2=\frac{1}{T}\sum_{t=1}^{T}(\phi_t-\tau_0+\tau_0-\hat{\tau}_T)^2=\frac{1}{T}\sum_{t=1}^{T}u_t^2-(\hat{\tau}_T-\tau_0)^2. 
$$
By Theorem~\ref{thm:asy_nor}, we have $\hat{\tau}_T\to_{a.s.}\tau_0$, and $\sqrt{T}(\hat{\tau}_T-\tau_0)\to_d \cN(0,\varsigma_*^2)$. Then, it suffices to show that $T^{-1}\sum_{t=1}^{T}u_t^2\to_p \varsigma_*^2$ as $T\to\infty$. By the proof of Theorem~\ref{thm:asy_nor} above, it is seen that $T^{-1}\sum_{t=1}^{T}\eE(u_t^2|\cF_{t-1})\to_p\varsigma_*^2$. 

Therefore, it suffices to show that $T^{-1}\sum_{t=1}^{T}v_t^2\to_p 0$. 
To see this, we apply the SLLN for martingale differences in Lemma~\ref{lem:SLLN} to show a stronger result, that $T^{-1}\sum_{t=1}^{T}v_t^2\to_{a.s.}0$. By Jensen's inequality, $\eE(|v_t|^{1+\theta/2}|\cF_{t-1})\le\eE(|u_t|^{2+\theta}|\cF_{t-1})+\eE\big\{|\eE(u_t^2|\cF_{t-1})|^{1+\theta/2}|\cF_{t-1}\big\}\le 2\eE(|u_t|^{2+\theta}|\cF_{t-1})$ a.s, where $\theta$ is specified in Lemma~\ref{lem:lyapunov}(ii). Then, it follows from Lemma~\ref{lem:lyapunov} that
$$
\sum_{t=1}^{\infty}\frac{\eE(|v_t|^{1+\theta/2}|\cF_{t-1})}{t^{1+\theta/2}}\le\sum_{t=1}^{\infty}\frac{2\eE(|u_t|^{2+\theta}|\cF_{t-1})}{t^{1+\theta/2}}<\infty
$$
almost surely. The proof is completed.

\subsection{Proof of Theorem~\ref{thm:regret}}
\label{supsubsec:proof_regret}
This section derives a finite-sample upper bound for the expected regret $\eE\cR_T$.

\subsubsection{Deviation inequalities}
\label{supsubsubsec:dev}
In this section, we drive some deviation inequalities for the non-asymptotic analysis. Before that, we provide some useful technical lemmas. Lemma~\ref{lem:bounded} is a special case of Lemma 1 in \textcolor{blue}{Qian and Yang} (\textcolor{blue}{2016}), obtained by setting $c=0$ in their formulation. In Lemma~\ref{lem:count}, parts (i) and (ii) correspond to Lemma 2 of \textcolor{blue}{Qian and Yang} (\textcolor{blue}{2016}) and Lemma B.5 of \textcolor{blue}{Neopane et al.} (\textcolor{blue}{2025a}), respectively. Thus, their proofs are omitted.
\begin{lemma}
    \label{lem:bounded}
    Suppose $\{\cF_j\}_{j\ge 1}$ is an increasing filtration of $\sigma$-fields. For each $j\ge 1$, let $\varepsilon_j$ be an $\cF_{j+1}$-measurable,  random variable satisfying $\eE(\varepsilon_j|\cF_j)=0,$ and let $Z_j$ be an $\cF_j$-measurable random variable with $|Z_j|\le M$ almost surely. If there exists constant $v>0$ such that for all $j\ge 1$ and $\lambda\in\eR,\eE\{\exp(\lambda\varepsilon_j)|\cF_j\}\le\exp(v^2\lambda^2/2)$, then for every $\epsilon>0$,
    $$
    \eP\left(\sum_{j=1}^{n}Z_j\varepsilon_j\ge n\epsilon\right)\le\exp\left(-\frac{n\epsilon^2}{2M^2v^2}\right).
    $$
\end{lemma}

\begin{lemma}
    \label{lem:count}
    Suppose $\{\cF_j\}_{j\ge 1}$ is an increasing filtration of $\sigma$-fields. For each $j\ge 1,$ let $W_j$ be an $\cF_j$-measurable Bernoulli random variable with $\eP(W_j=1|\cF_{j-1})=p_j$ for some $0\le p_j\le 1.$ Then, we have \\
    (i) $\eP\big(\sum_{j=1}^{t}W_j\le\sum_{j=1}^{t}p_j/2\big)\le\exp\big(-3\sum_{j=1}^{t}p_j/28\big);$\\
    (ii) with probability at least $1-\delta,\big|\sum_{j=1}^{t}W_j-\sum_{j=1}^{t}p_j\big|\le C_{\delta,t,0}\sqrt{t}$ holds for all $t\ge1,$ where $C_{\delta,t,0}=0.85\sqrt{\log\log t+0.72\log(5.2/\delta)}$.
\end{lemma}

\begin{lemma}
    \label{lem:sub_gau_exp}
    (Lemma 2.8.6 of \textcolor{blue}{Vershynin} (\textcolor{blue}{2025})) If $X$ and $Y$ are sub-Gaussian, then $XY$ is sub-exponential, and $\|XY\|_{\psi_1}=\|X\|_{\psi_2}\|Y\|_{\psi_2}.$
\end{lemma}

\begin{lemma}
    \label{lem:exp_ber}
    (Proposition 4.2 of \textcolor{blue}{Zhang and Chen} (\textcolor{blue}{2020})) Let $\{\varepsilon_j\}_{j\ge1}$ be independent, mean zero, sub-exponential random variables. If there exists constant $K>0$ such that $\|\varepsilon_j\|_{\psi_1}\le K$ for all $j\ge1$. Then for every $\epsilon>0,$
    $$
    \eP\left(\sum_{j=1}^{n}\varepsilon_j\ge n\epsilon\right)\le\exp\left\{-\left(\frac{n\epsilon^2}{8K^2}\wedge\frac{n\epsilon}{4K}\right)\right\}.
    $$
\end{lemma}

We first derive the deviation inequality for the Nadaraya-Watson estimator. Suppose that $Y_j=m(\bX_j)+\varepsilon_j$ for $j\ge1$. Let $\hat{m}_n(\bx)=\sum_{j=1}^{n}K_h(\bx-\bX_j)Y_j/\sum_{j=1}^{n}K_h(\bx-\bX_j)$ be the Nadaraya-Watson estimator of $m(\bx)$ for a given $\bx$, where $K_h(\bx-\bX_j)=K\{(\bx-\bX_j)/h\}/h$. We impose the following conditions.\\
(C1): $\{\bX_j\}$ are i.i.d. with support contained in a compact set $\cX\subset\eR^d$. The distribution of $\bX_j$ is dominated by the Lebesgue measure, with a continuous density $p(\bx)$ uniformly bounded above and away from zero on $\cX$, i.e., $\underline{c}\le p(\bx)\le\bar{c}$ for some positive constants $\underline{c}\le\bar{c}$.\\
(C2): The function $m(\cdot)\in\Sigma(\beta,L)$ for some $0<\beta\le1$ (Cf. Definition \ref{def:holder}).\\
(C3): The model errors $\{\varepsilon_j\}$ satisfy that $\eE(\varepsilon_j|\bX_j)=0$, and there exists constant $v>0$ such that $\eE\{\exp(\lambda\varepsilon_j)|\bX_j\}\le\exp(v^2\lambda^2/2)$ for all $\lambda\in\eR.$\\
(C4): The multivariate kernel function satisfies $K(\bx)=0$ for $\|\bx\|_{\infty}>1,$ and there exists some constant $\bar{k}>0$ such that $0\le K(\bx)< \bar{k},\forall\bx\in\eR^d$. Moreover, $K(\cdot)$ is replaced with the uniform kernel $K(\bu)=\underline{k}\cdot I(\|\bu\|_{\infty}\le 1)$ whenever $\sum_{j=1}^{n}K_{h}(\bx-\bX_j)\le\underline{k}\sum_{j=1}^{n}I(\|\bx-\bX_j\|_{\infty}\le h)/h$ for some small constant $\underline{k}>0.$

\begin{proposition}
    \label{propos:NW_est}
    Suppose that Conditions (C1)--(C4) hold. For a given $\bx\in\cX$ and every small $\epsilon>\sqrt{d}Lh^{\beta}$, we have
    $$
    \eP(|\hat{m}_n(\bx)-m(\bx)|\ge\epsilon)\le3\exp\left(-\frac{2^d\underline{k}^2\underline{c}nh^d(\epsilon-\sqrt{d}Lh^{\beta})^2}{\bar{k}^2v^2}\right).
    $$
\end{proposition}
\begin{proof}
    Define $m_n(\bx)=\sum_{j=1}^{n}K_h(\bx-\bX_j)m(\bX_j)/\sum_{j=1}^{n}K_h(\bx-\bX_j),$ and $B_n(\bx)=|\{j\in[n]:\|\bx-\bX_j\|_{\infty}\le h\}|$. Note that for each $x\in\eR^d,|\hat{m}_n(\bx)-m(\bx)|\le |m_n(\bx)-m(\bx)|+|\hat{m}_n(\bx)-m_n(\bx)|.$
    For the first term, by Condition (C2),
    \begin{equation}
        \label{eq:m_bias}
        \begin{aligned}
            |m_n(\bx)-m(\bx)|=&\left|\frac{\sum_{j=1}^{n}K_h(\bx-\bX_j)\{m(\bX_j)-m(\bx)\}}{\sum_{j=1}^{n}K_h(\bx-\bX_j)}\right|\\
            \le&\sqrt{d}Lh^{\beta}I\{B_n(\bx)\ge1\}+ |m(\bx)|\cdot I\{B_n(\bx)=0\}.
        \end{aligned}
    \end{equation}
    For the second term, by Condition (C4),
    $$
    |\hat{m}_n(\bx)-m_n(\bx)|=\left|\frac{\sum_{j=1}^{n}K_h(\bx-\bX_j)(Y_j-m(\bX_j))}{\sum_{j=1}^{n}K_h(\bx-\bX_j)}\right|\le\left|\frac{1}{\underline{k}B_n(\bx)}\sum_{j=1}^{n}K\Big(\frac{\bx-\bX_j}{h}\Big)\varepsilon_j\right|.
    $$
    Note that there are $B_n(\bx)$ nonzero terms in $\sum_{j=1}^{n}K\{(\bx-\bX_j)/h\}\varepsilon_j$. Then, by Lemma~\ref{lem:bounded} and Conditions (C3)--(C4),
    \begin{equation}
        \label{eq:m2}
        \eP\big(|\hat{m}_n(\bx)-m_n(\bx)|\ge \epsilon-\sqrt{d}Lh^{\beta}\big|\{\bX_j\}_{j\in[n]}\big)\le 2\exp\left(-\frac{B_n(\bx)\underline{k}^2(\epsilon-\sqrt{d}Lh^{\beta})^2}{2\bar{k}^2v^2}\right),
    \end{equation}
    where $\eP(\cdot|\{\bX_j\}_{j\in[n]})$ denotes the conditional probability given $\{\bX_j\}_{j\in[n]}.$ 
    By Condition (C1), $\eP(\|\bx-\bX_j\|_{\infty}\le h)\ge\underline{c}(2h)^d$ for $j\in[n],$ which together with Lemma~\ref{lem:count}(i) implies that
    \begin{equation}
        \label{eq:B}
    \eP\Big(B_n(\bx)\le\frac{\underline{c}n(2h)^d}{2}\Big)\le\exp\left(-\frac{3\underline{c}n(2h)^d}{28}\right). 
    \end{equation}
    Thus, we have
    $$
    \begin{aligned}
        &\eP(|\hat{m}_n(\bx)-m(\bx)|\ge\epsilon)\\
        \le&\eP\Big(B_n(\bx)\le\frac{\underline{c}n(2h)^d}{2}\Big)+\eP\Big(|\hat{m}_n(\bx)-m(\bx)|\ge\epsilon,B_n(\bx)>\frac{\underline{c}n(2h)^d}{2}\Big)\\
        \le&\exp\left(-\frac{3\underline{c}n(2h)^d}{28}\right)+\eP\Big(|\hat{m}_n(\bx)-m_n(\bx)|\ge\epsilon-\sqrt{d}Lh^{\beta},B_n(\bx)>\frac{\underline{c}n(2h)^d}{2}\Big)\\
        \le&\exp\left(-\frac{3\underline{c}n(2h)^d}{28}\right)+2\exp\left(-\frac{2^d\underline{k}^2\underline{c}nh^d(\epsilon-\sqrt{d}Lh^{\beta})^2}{\bar{k}^2v^2}\right)\\
        \le&3\exp\left(-\frac{2^d\underline{k}^2\underline{c}nh^d(\epsilon-\sqrt{d}Lh^{\beta})^2}{\bar{k}^2v^2}\right),
    \end{aligned}
    $$
    where the second inequality follows from \eqref{eq:m_bias} and \eqref{eq:m2}, and the third inequality follows from \eqref{eq:B}. The proof is completed. 
\end{proof}

Next, we provide the deviation inequalities for the estimations of $\gamma_z$ and the conditional variance of $\varepsilon_t(z)$ for $z\in\{0,1\}$ in Model~\eqref{eq:model_varep}. For simplicity, we suppose 
$$
Y_j=m_Y(\bX_j)-\gamma m_S(\bX_j)+\gamma S_j+\varepsilon_j=\mu(\bX_j,S_j)+\varepsilon_j,\quad j\ge1,
$$
where $Y_j=m_Y(\bX_j)+\varepsilon_{Y,j},S_j=m_S(\bX_j)+\varepsilon_{S,j},$ and $\mu(\bX_j,S_j)=m_Y(\bX_j)-\gamma m_S(\bX_j)+\gamma S_j$. Assume that $\{\varepsilon_{Y,j}\}$ and $\{\varepsilon_{S,j}\}$ satisfy Condition (C3) with positive constants $v_Y$ and $v_S$, respectively. Additional, we impose the following conditions.\\
(C5): $\var(\varepsilon_{S,j}|\bX_j)=\sigma_S^2,\eE(\varepsilon_j|\varepsilon_{S,j})=0,\var(\varepsilon_j|\varepsilon_{S,j})=\sigma^2$, and there exist constants $K_Y,K_S>0$ such that $\|\varepsilon_{Y,j}\|_{\psi_2}\le K_Y$ and $\|\varepsilon_{S,j}\|_{\psi_2}\le K_S$ for all $j\ge1$.\\
(C6): Let $\hat{m}_{Y,n}(\cdot)$ and $\hat{m}_{S,n'}(\cdot)$ be some estimators of $m_Y(\cdot)$ and $m_S(\cdot)$, with sample sizes $n$ and $n'$, respectively. For each given $\bx$ and every small $\epsilon_1,\epsilon_2>0$, $\eP(|\hat{m}_{Y,n}(\bx)-m_Y(\bx)|\ge \epsilon_1)\le\delta_{1n}(\epsilon_1)$, and $\eP(|\hat{m}_{S,n'}(\bx)-m_S(\bx)|\ge \epsilon_2)\le\delta_{2n'}(\epsilon_2)$. 

\begin{proposition}
    \label{propos:sigma_est}
    Suppose that Conditions (C5)--(C6) hold. \\
    (i) Let $\hat{\gamma}_n=\sum_{j=1}^{n}\{Y_j-\hat{m}_{Y,n}(\bX_j)\}\{S_j-\hat{m}_{S,n'}(\bX_j)\}/\sum_{j=1}^{n}\{S_j-\hat{m}_{S,n'}(\bX_j)\}^2$ be the least square estimate of $\gamma$. Then, for every small $\epsilon_3>4^{-1}\sigma_S^2\{3\sigma\epsilon_2+2\sigma_S(\epsilon_1+|\gamma|\epsilon_2)\},$
    $$
    \begin{aligned}
        \delta_{3n}(\epsilon_3):=&\eP(|\hat{\gamma}_n-\gamma|\ge\epsilon_3)\\
        \le& n\delta_{1n}(\epsilon_1)+n\delta_{2n'}(\epsilon_2)+3\exp\left(-\frac{n\{4\epsilon_3/\sigma_S^2-3\sigma\epsilon_2-2\sigma_S(\epsilon_1+|\gamma|\epsilon_2)\}^2}{8K_S^2(K_Y+|\gamma|K_S)^2}\right).
    \end{aligned}
    $$
    (ii) Let $\hat{\mu}_n(\bx,s)=\hat{m}_{Y,n}(\bx)-\hat{\gamma}_n\hat{m}_{S,n'}(\bx)+\hat{\gamma}_ns$ be the plug-in estimator of $\mu(\cdot,\cdot).$ Then, for every small $\epsilon_4>\epsilon_1+2|\gamma|\epsilon_2$, we have
    $$
    \begin{aligned}
        \delta_{4n}(\epsilon_4):=&\eP\left(\max_{j\in[n]}|\hat{\mu}_n(\bX_j,S_j)-\mu(\bX_j,S_j)|\ge\epsilon_4\right)\\
        \le&\delta_{3n}(\epsilon_3)
        +2n\exp\left(-\frac{(\epsilon_4-\epsilon_1-2|\gamma|\epsilon_2)^2}{2\epsilon_3^2v_S^2}\right).
    \end{aligned}
    $$
    (iii) Let $\hat{\sigma}_n^2=n^{-1}\sum_{j=1}^{n}\{Y_j-\hat{\mu}_n(\bX_j,S_j)\}^2$ be the estimate of the conditional variance $\sigma^2$ for $\varepsilon_j$. Then, for every small $\epsilon_5>3\epsilon_4$, we have
    $$
    \eP\left(|\hat{\sigma}_n-\sigma|>\epsilon_5\right)\le \delta_{4n}(\epsilon_4)+3\exp\left(-\frac{n(\sigma\epsilon_5-3\sigma\epsilon_4)^2}{8(K_Y+|\gamma|K_S)^4}\right).
    $$
\end{proposition}
\begin{proof}
    (i) Denote $U_{j,n}=m_Y(\bX_j)-\hat{m}_{Y,n}(\bX_j)$ and $V_{j,n'}=m_S(\bX_j)-\hat{m}_{S,n'}(\bX_j)$ for $j\in[n].$ Note that $Y_j-\hat{m}_{Y,n}(\bX_j)=U_{j,n}+\varepsilon_{Y,j},S_j-\hat{m}_{S,n'}(\bX_j)=V_{j,n'}+\varepsilon_{S,j}$ and $\varepsilon_j=\varepsilon_{Y,j}-\gamma\varepsilon_{S,j}$.
    By the triangle inequality and the Cauchy-Schwarz inequality, we respectively have $\|\varepsilon_j\|_{\psi_2}\le\|\varepsilon_{Y,j}\|_{\psi_2}+|\gamma|\|\varepsilon_{S,j}\|_{\psi_2}\le K_Y+|\gamma|K_S$.
    Then, we write
    $$
    \begin{aligned}
        \hat{\gamma}_n-\gamma=&\frac{n^{-1}\sum_{j=1}^{n}(U_{j,n}+\varepsilon_{Y,j})(V_{j,n'}+\varepsilon_{S,j})}{n^{-1}\sum_{j=1}^{n}(V_{j,n'}+\varepsilon_{S,j})^2}-\gamma\\
        =&\frac{n^{-1}\sum_{j=1}^{n}(V_{j,n'}+\varepsilon_{S,j})(U_{j,n}-\gamma V_{j,n'}+\varepsilon_j)}{n^{-1}\sum_{j=1}^{n}(V_{j,n'}+\varepsilon_{S,j})^2}.
    \end{aligned}
    $$
    Define the event $A(\epsilon_1,\epsilon_2):=\big\{\max_{j\in[n]}|U_{j,n}|\ge\epsilon_1\big\}\bigcup\big\{\max_{j\in[n]}|V_{j,n'}|\ge\epsilon_2\big\}$. By Condition (C6), it follows that
    $$
    \eP\Big(\max_{j\in[n]}|U_{j,n}|\ge\epsilon_1\Big)\le\sum_{j=1}^{n}\eP(|U_{j,n}|\ge\epsilon_1)\le n\delta_{1n}(\epsilon_1),
    $$ 
    and also $\eP\big(\max_{j\in[n]}|V_{j,n'}|\ge\epsilon_2\big)\le n\delta_{2n'}(\epsilon_2),$ which together imply that $\eP(A(\epsilon_1,\epsilon_2))\le n\delta_{1n}(\epsilon_1)+n\delta_{2n'}(\epsilon_2)$. Under $A(\epsilon_1,\epsilon_2)^c,\big|n^{-1}\sum_{j=1}^{n}V_{j,n'}(U_{j,n}-\gamma V_{j,n'})\big|\le\epsilon_1\epsilon_2+|\gamma|\epsilon_2^2<\sigma\epsilon_2$, and $n^{-1}\sum_{j=1}^{n}V_{j,n'}^2\le\epsilon_2^2<\sigma_S^2/8$. By Lemma~\ref{lem:sub_gau_exp}, $\|\varepsilon_j^2\|_{\psi_1}\le \|\varepsilon_j\|_{\psi_2}^2$ and $\|\varepsilon_{S,j}^2\|_{\psi_1}\le \|\varepsilon_{S,j}\|_{\psi_2}^2.$ Then, combining Conditions (C5)--(C6), the Cauchy-Schwarz inequality and Lemma~\ref{lem:exp_ber}, we have
    {\small
    \begin{align}
            &\eP\left(\Big|n^{-1}\sum_{j=1}^{n}V_{j,n'}\varepsilon_j\Big|\ge2\sigma\epsilon_2,A(\epsilon_1,\epsilon_2)^c\right)\le\eP\left(n^{-1}\sum_{j=1}^{n}\varepsilon_j^2\ge4\sigma^2\right)\label{eq:V_var_1}\\
            =&\eP\left(n^{-1}\sum_{j=1}^{n}(\varepsilon_j^2-\sigma^2)\ge3\sigma^2\right)
            \le\exp\left\{-\left(\frac{9n\sigma^4}{8(K_Y+|\gamma|K_S)^4}\wedge\frac{3n\sigma^2}{4(K_Y+|\gamma|K_S)^2}\right)\right\},\nonumber\\
            &\eP\left(\Big|n^{-1}\sum_{j=1}^{n}(U_{j,n}-\gamma V_{j,n'})\varepsilon_{S,j}\Big|
            \ge2\sigma_S(\epsilon_1+|\gamma|\epsilon_2),A(\epsilon_1,\epsilon_2)^c\right)
            \le\eP\left(n^{-1}\sum_{j=1}^{n}\varepsilon_{S,j}^2\ge4\sigma_S^2\right)\label{eq:V_var_2}\\
            =&\eP\left(n^{-1}\sum_{j=1}^{n}(\varepsilon_{S,j}^2-\sigma_S^2)\ge3\sigma_S^2\right)
            \le\exp\left\{-\left(\frac{9n\sigma_S^4}{8K_S^4}\wedge\frac{3n\sigma_S^2}{4K_S^2}\right)\right\},\nonumber\\
            &\eP\left(\Big|n^{-1}\sum_{j=1}^{n}V_{j,n'}\varepsilon_{S,j}\Big|\ge\frac{\sigma_S^2}{16},A(\epsilon_1,\epsilon_2)^c\right)\le\eP\left(n^{-1}\sum_{j=1}^{n}\varepsilon_{S,j}^2\ge\frac{\sigma_S^4}{256\epsilon_2^2}\right)\nonumber\\
            \le&\eP\left(n^{-1}\sum_{j=1}^{n}(\varepsilon_{S,j}^2-\sigma_S^2)\ge\frac{\sigma_S^4-256\epsilon_2^2\sigma_S^2}{256\epsilon_2^2}\right)\le\exp\left(-\frac{n(\sigma_S^4-256\epsilon_2^2\sigma_S^2)}{1024\epsilon_2^2K_S^2}\right).\nonumber
    \end{align}
    }
    By Lemma~\ref{lem:sub_gau_exp}, $\|\varepsilon_{S,j}\varepsilon_j\|_{\psi_1}\le\|\varepsilon_{S,j}\|_{\psi_2}\|\varepsilon_j\|_{\psi_2}$. Then it follows from Lemma~\ref{lem:exp_ber} that
    $$
    \begin{aligned}
        &\eP\left(\Big|n^{-1}\sum_{j=1}^{n}\varepsilon_{S,j}\varepsilon_j\Big|\ge4\epsilon_3/\sigma_S^2-3\sigma\epsilon_2-2\sigma_S(\epsilon_1+|\gamma|\epsilon_2)\right)\\
        \le &2\exp\left(-\frac{n\{4\epsilon_3/\sigma_S^2-3\sigma\epsilon_2-2\sigma_S(\epsilon_1+|\gamma|\epsilon_2)\}^2}{8K_S^2(K_Y+|\gamma|K_S)^2}\right),\\
        &\eP\left(n^{-1}\sum_{j=1}^{n}(\varepsilon_{S,j}^2-\sigma_S^2)\le-\frac{\sigma_S^2}{2}\right)\le\exp\left\{-\left(\frac{n\sigma_S^4}{32K_S^4}\wedge\frac{n\sigma_S^2}{16K_S^2}\right)\right\}.
    \end{aligned}
    $$
    Combining above, we have
    {\small$$
    \begin{aligned}
        &\eP(|\hat{\gamma}_n-\gamma|\ge\epsilon_3)\\
        \le&\eP(A(\epsilon_1,\epsilon_2))+\eP\left(n^{-1}\sum_{j=1}^{n}(V_{j,n'}+\varepsilon_{S,j})^2<\frac{\sigma_S^2}{4},A(\epsilon_1,\epsilon_2)^c\right)\\
        &+\eP\left(n^{-1}\sum_{j=1}^{n}(V_{j,n'}+\varepsilon_{S,j})(U_{j,n}-\gamma V_{j,n'}+\varepsilon_j)\ge 4\epsilon_3/\sigma_S^2,n^{-1}\sum_{j=1}^{n}(V_{j,n'}+\varepsilon_{S,j})^2\ge\frac{\sigma_S^2}{4},A(\epsilon_1,\epsilon_2)^c\right)\\
        \le&n\delta_{1n}(\epsilon_1)+n\delta_{2n'}(\epsilon_2)+3\exp\left(-\frac{n\{4\epsilon_3/\sigma_S^2-3\sigma\epsilon_2-2\sigma_S(\epsilon_1+|\gamma|\epsilon_2)\}^2}{8K_S^2(K_Y+|\gamma|K_S)^2}\right),
    \end{aligned}
    $$
    }which completes the proof.\\
    (ii) Consider the following decomposition:
    $$
    \hat{\mu}_n(\bX_j,S_j)-\mu(\bX_j,S_j)=\{\hat{m}_{Y,n}(\bX_j)-m_Y(\bX_j)\}-\hat{\gamma}_n\{\hat{m}_{S,n}(\bX_j)-m_S(\bX_j)\}+(\hat{\gamma}_n-\gamma)\varepsilon_{S,j}.
    $$
    Define the event $A(\epsilon_1,\epsilon_2,\epsilon_3):=A(\epsilon_1,\epsilon_2)\bigcup\big\{|\hat{\gamma}_n-\gamma|\ge\epsilon_3\big\}.$ By the proof of part (i), $\eP(A(\epsilon_1,\epsilon_2,\epsilon_3))\le\delta_{3n}(\epsilon_3)$. Under $A(\epsilon_1,\epsilon_2,\epsilon_3)^c,$ we have $\max_{j\in[n]}|\hat{m}_{Y,n}(\bX_j)-m_Y(\bX_j))|\le\epsilon_1,$ and $\max_{j\in[n]}|\hat{\gamma}_n\{\hat{m}_{S,n'}(\bX_j)-m_S(\bX_j)\}|\le2|\gamma|\epsilon_2.$ By using the tail probability of the sub-Gaussian random variable $\varepsilon_{S,j}$ with proxy variance $v_S^2,$ it follows that
    $$
    \eP\left(\max_{j\in[n]}|(\hat{\gamma}_n-\gamma)\varepsilon_{S,j}|\ge \epsilon_4-\epsilon_1-2|\gamma|\epsilon_2,A(\epsilon_1,\epsilon_2,\epsilon_3)^c\right)\le 2n\exp\left(-\frac{(\epsilon_4-\epsilon_1-2|\gamma|\epsilon_2)^2}{2\epsilon_3^2v_S^2}\right),
    $$
    which implies the desired result. \\
    (iii) Consider the following decomposition:
    $$
    \begin{aligned}
        \hat{\sigma}_n^2-\sigma^2=&n^{-1}\sum_{j=1}^{n}(\varepsilon_j^2-\sigma^2)+n^{-1}\sum_{j=1}^{n}\{\mu(\bX_j,S_j)-\hat{\mu}_n(\bX_j,S_j)\}^2\\
        &+2n^{-1}\sum_{j=1}^{n}\{\mu(\bX_j,S_j)-\hat{\mu}_n(\bX_j,S_j)\}\varepsilon_j.
    \end{aligned}
    $$
    Define the event $A(\epsilon_4):=\big\{\max_{j\in[n]}|\mu(\bX_j,S_j)-\hat{\mu}_n(\bX_j,S_j)|\ge\epsilon_4\big\}.$ By part (ii), we have $\eP(A(\epsilon_4))\le \delta_{4n}(\epsilon_4).$ Under $A(\epsilon_4)^c,\big|n^{-1}\sum_{j=1}^{n}\{\mu(\bX_j,S_j)-\hat{\mu}_n(\bX_j,S_j)\}^2\big|\le \epsilon_4^2<\sigma\epsilon_4.$ Then, Condition (C5), the Cauchy-Schwarz inequality and Lemma~\ref{lem:exp_ber}, we have
    $$
    \begin{aligned}
        &\eP\left(\Big|n^{-1}\sum_{j=1}^{n}\{\mu(\bX_j,S_j)-\hat{\mu}_n(\bX_j,S_j)\}\varepsilon_j\Big|\ge2\sigma\epsilon_4,A(\epsilon_4)^c\right)\le \eP\left(n^{-1}\sum_{j=1}^{n}\varepsilon_j^2\ge4\sigma^2\right)\\
        =&\eP\left(n^{-1}\sum_{j=1}^{n}(\varepsilon_j^2-\sigma^2)\ge3\sigma^2\right)
        \le\exp\left\{-\left(\frac{9n\sigma^4}{8(K_Y+|\gamma|K_S)^4}\wedge\frac{3n\sigma^2}{4(K_Y+|\gamma|K_S)^2}\right)\right\}.
    \end{aligned}
    $$
    By Lemma~\ref{lem:sub_gau_exp},  $\|\varepsilon_j^2\|_{\psi_1}\le \|\varepsilon_j\|_{\psi_2}^2.$ Then by Lemma~\ref{lem:exp_ber},
    $$
    \eP\left(\Big|n^{-1}\sum_{j=1}^{n}(\varepsilon_j^2-\sigma^2)\Big|\ge\sigma\epsilon_5-3\sigma\epsilon_4\right)\le2\exp\left(-\frac{n(\sigma\epsilon_5-3\sigma\epsilon_4)^2}{8(K_Y+|\gamma|K_S)^4}\right).
    $$
    Combining above, we have
    $$
    \eP\left(|\hat{\sigma}_n^2-\sigma^2|>\sigma\epsilon_5\right)\le \delta_{4n}(\epsilon_4)+3\exp\left(-\frac{n(\sigma\epsilon_5-3\sigma\epsilon_4)^2}{8(K_Y+|\gamma|K_S)^4}\right),
    $$
    which, together with the fact that $|\hat{\sigma}_n-\sigma|\le|\hat{\sigma}_n^2-\sigma^2|/\sigma$, implies the desired result.
\end{proof}
\begin{remark}
    \label{rmk:corss}
    Note that Condition (C6) does not require $\hat{m}_{Y,n}(\cdot)$ and $\hat{m}_{S,n'}(\cdot)$ to be independent of $\{(X_j,S_j,Y_j)\}_{j\in[n]}$. Motivated by the cross-fitting technique of \textcolor{blue}{Chernozhukov et al.} (\textcolor{blue}{2018}), if $\hat{m}_{Y,n}(\cdot)$ and $\hat{m}_{S,n'}(\cdot)$ are estimated using data independent of the evaluation data $\{(X_j,S_j,Y_j)\}_{j\in[n]}$, then the convergence rates established in Proposition~\ref{propos:sigma_est} can be improved. Specifically, under such independence, when deriving the deviation inequalities~\eqref{eq:V_var_1} and \eqref{eq:V_var_2}, Lemma~\ref{lem:bounded} can be applied instead of Lemma~\ref{lem:exp_ber}. Consequently, the biases associated with $n^{-1}\sum_{j=1}^{n}V_{j,n'}\varepsilon_j$ and $n^{-1}\sum_{j=1}^{n}(U_{j,n}-\gamma V_{j,n'})\varepsilon_{S,j}$ are both dominated by that of $n^{-1}\sum_{j=1}^{n}V_{j,n'}(U_{j,n}-\gamma V_{j,n'})$. This implies that the bias term of $\hat{\gamma}_n$ in Proposition~\ref{propos:sigma_est}(i) is asymptotically of the second order, i.e., dominated by quadratic terms of the first-step nonparametric regression bias, namely $\epsilon_1\epsilon_2+|\gamma|\epsilon_2^2$, in the same spirit as the double machine learning framework where first-step biases affect the second-step estimator only through quadratic terms. A similar refinement applies to $\hat{\sigma}_n$ in Proposition~\ref{propos:sigma_est}(iii). As a result, the convergence rates of $\hat{\gamma}_n$ and $\hat{\sigma}_n$ established in the following section are improved from $\widetilde{O}_p(n^{-\beta/(2\beta+d)})$ to $\widetilde{O}_p(n^{-2\beta/(2\beta+d)})+O_p(n^{-1/2})$, attaining the parametric rate $O_p(n^{-1/2})$ when $d<2\beta$, and $\widetilde{O}_p(n^{-1/2})$ when $d\le2\beta$. Nevertheless, in this paper we refrain from employing cross-fitting, since it does not ultimately improve the asymptotic order of the expected regret bound w.r.t $T$. The reason is that, although cross-fitting may sharpen the finite-sample bound up to a constant factor, the gain from the new rates is dominated by other terms, yielding the same asymptotic order as in Theorem~\ref{thm:regret}. Moreover, cross-fitting is not used in the classical Robinson’s method either.
\end{remark}

\subsubsection{Good event}
\label{supsubsubsec:event}

To facilitate an explicit formulation of the regret bound in Theorem~\ref{thm:regret}, we define the following constants: $v_{\max}=v_{Y,0}^2\vee v_{Y,1}^2\vee v_{S,0}^2\vee v_{S,1}^2,$ $\gamma_{\max}=|\gamma_0|\vee|\gamma_1|,$ $\sigma_{S,\max}^2=\sigma_{S,1}^2\vee\sigma_{S,0}^2$, 
$\sigma_{\max}^2=\sigma_1^2\vee\sigma_0^2$, $M_{Y,0}=\sup_{\bsx\in\cX}|m_{Y,0}(\bx)|,$ $M_{Y,1}=\sup_{\bsx\in\cX}|m_{Y,1}(\bx)|,$ $M_{S,0}=\sup_{\bsx\in\cX}|m_{S,0}(\bx)|,$ $M_{S,1}=\sup_{\bsx\in\cX}|m_{S,1}(\bx)|.$ 
Given $0<\delta<1/2$ and for $t\in[T]$, we then define the following quantities:
$$
\begin{aligned}
    &C_{\delta,t,0}=0.85\sqrt{\log\log t+0.72\log(5.2/\delta)},\\
    &C_{\delta,t,1}'=\sqrt{\frac{\bar{k}^2v_{\max}^2\log(30Tt/\delta)}{2^d\underline{k}^2\underline{c}}},\quad C_{\delta,t,1}=C_{\delta,t,1}'+\sqrt{d}L,\\
    &C_{\delta,t,2}'=\frac{\sigma_{S,\max}^2}{4}\sqrt{8K_S^2(K_Y+\gamma_{\max}K_S)^2\log(30T/\delta)},\\ 
    &C_{\delta,t,2}=C_{\delta,t,2}'+\frac{\sigma_{S,\max}^2}{4}\{3\sigma_{\max}+2\sigma_{S,\max}(1+\gamma_{\max})\}C_{\delta,t,1},\\
    &C_{\delta,t,3}'=\sqrt{2v_{\max}^2\log(20Tt/\delta)},\quad C_{\delta,t,3}=(1+2\gamma_{\max})C_{\delta,t,1}+C_{\delta,t,3}'C_{\delta,t,2},\\
    &C_{\delta,t,4}'=\sqrt{8(K_Y+\gamma_{\max}K_S)^4\log(30T/\delta)},\quad C_{\delta,t,4}=C_{\delta,t,4}'+3\sigma_{\max}C_{\delta,t,3}.
\end{aligned}
$$
It can be shown that $C_{\delta,T,4}\asymp \log \delta^{-1}T$ and $C_{\delta,T,0}\asymp\sqrt{\log\log T+\log(1/\delta)}$. Next, we define the good event by using the deviation inequalities obtained in Section~\ref{supsubsubsec:dev}.

{\bf Step 1}. By Proposition~\ref{propos:NW_est}, with probability at least $1-2\delta/(5T),$
$$
\begin{aligned}
    &\max_{j\in G_{1,t}}\big|\hat{m}_{Y,1,t}(\bX_j)-m_{Y,1}(\bX_j)\big|\le\frac{\sqrt{d}L}{N_{1,t}^{\beta/(2\beta+d)}}+\frac{C_{\delta,t,1}'}{N_{1,t}^{\beta/(2\beta+d)}}=\frac{C_{\delta,t,1}}{N_{1,t}^{\beta/(2\beta+d)}}=:\Psi_{Y,1}(t,\delta),\\
    &\max_{j\in G_{0,t}}\big|\hat{m}_{Y,0,t}(\bX_j)-m_{Y,0}(\bX_j)\big|\le\frac{\sqrt{d}L}{N_{0,t}^{\beta/(2\beta+d)}}+\frac{C_{\delta,t,1}'}{N_{0,t}^{\beta/(2\beta+d)}}=\frac{C_{\delta,t,1}}{N_{0,t}^{\beta/(2\beta+d)}}=:\Psi_{Y,0}(t,\delta),\\
    &\max_{j\in[t]}\big|\hat{m}_{S,1,t}(\bX_j)-m_{S,1}(\bX_j)\big|\le\frac{\sqrt{d}L}{t^{\beta/(2\beta+d)}}+\frac{C_{\delta,t,1}'}{t^{\beta/(2\beta+d)}}=\frac{C_{\delta,t,1}}{t^{\beta/(2\beta+d)}}=:\Psi_{S,1}(t,\delta),\\
    &\max_{j\in[t]}\big|\hat{m}_{S,0,t}(\bX_j)-m_{S,0}(\bX_j)\big|\le\frac{\sqrt{d}L}{t^{\beta/(2\beta+d)}}+\frac{C_{\delta,t,1}'}{t^{\beta/(2\beta+d)}}=\frac{C_{\delta,t,1}}{t^{\beta/(2\beta+d)}}=:\Psi_{S,0}(t,\delta).
\end{aligned}
$$

{\bf Step 2}. By Proposition~\ref{propos:sigma_est}(i), with probability at least $1-3\delta/(5T),$
$$
\begin{aligned}
    |\hat{\gamma}_{1,t}-\gamma_1|\le&\frac{\sigma_{S,1}^2}{4}\big[3\sigma_1\Psi_{S,1}(t,\delta)+2\sigma_{S,1}\{\Psi_{Y,1}(t,\delta)+|\gamma_1|\Psi_{S,1}(t,\delta)\}\big]\\
    &+\frac{C_{\delta,t,2}'}{\sqrt{N_{1,t}}}
    \le\frac{C_{\delta,t,2}}{N_{1,t}^{\beta/(2\beta+d)}}=:\Psi_{\gamma,1}(t,\delta),\\
    |\hat{\gamma}_{0,t}-\gamma_0|\le&\frac{\sigma_{S,0}^2}{4}\big[3\sigma_0\Psi_{S,0}(t,\delta)+2\sigma_{S,0}\{\Psi_{Y,0}(t,\delta)+|\gamma_1|\Psi_{S,0}(t,\delta)\}\big]\\
    &+\frac{C_{\delta,t,2}'}{\sqrt{N_{0,t}}}
    \le\frac{C_{\delta,t,2}}{N_{0,t}^{\beta/(2\beta+d)}}=:\Psi_{\gamma,0}(t,\delta).
\end{aligned}
$$

{\bf Step 3}. By Proposition~\ref{propos:sigma_est}(ii), with probability at least $1-4\delta/(5T),$
$$
\begin{aligned}
    &\max_{j\in G_{1,t}}\big|\hat{\mu}_{1,t}(\bX_j,\bS_j)-\mu_1(\bX_j,\bS_j)\big|\\
    \le&\Psi_{Y,1}(t,\delta)+2\gamma_{\max}\Psi_{S,1}(t,\delta)+C_{\delta,t,3}'\Psi_{\gamma,1}(t,\delta)
    \le\frac{C_{\delta,t,3}}{N_{1,t}^{\beta/(2\beta+d)}}=:\Psi_{\mu,1}(t,\delta),\\
    &\max_{j\in G_{0,t}}\big|\hat{\mu}_{0,t}(\bX_j,\bS_j)-\mu_0(\bX_j,\bS_j)\big|\\
    \le&\Psi_{Y,0}(t,\delta)+2\gamma_{\max}\Psi_{S,0}(t,\delta)+C_{\delta,t,3}'\Psi_{\gamma,0}(t,\delta)
    \le\frac{C_{\delta,t,3}}{N_{0,t}^{\beta/(2\beta+d)}}=:\Psi_{\mu,0}(t,\delta).
\end{aligned}
$$

{\bf Step 4}. By Proposition~\ref{propos:sigma_est}(iii), with probability at least $1-\delta/T,$
\begin{equation}
    \label{eq:C_3}
    \begin{aligned}
        |\hat{\sigma}_{1,t}-\sigma_1|\le& 3\Psi_{\mu,1}(t,\delta)+\frac{C_{\delta,t,4}'}{\sigma_1\sqrt{N_{1,t}}}\le\frac{C_{\delta,t,4}}{\sigma_1N_{1,t}^{\beta/(2\beta+d)}}:=\Psi_{\sigma,1}(t,\delta),\\
        |\hat{\sigma}_{0,t}-\sigma_0|\le& 3\Psi_{\mu,0}(t,\delta)+\frac{C_{\delta,t,4}'}{\sigma_0\sqrt{N_{0,t}}}\le\frac{C_{\delta,t,4}}{\sigma_0N_{0,t}^{\beta/(2\beta+d)}}:=\Psi_{\sigma,0}(t,\delta).
    \end{aligned}
\end{equation}

Then, define the following events:
\begin{align}
    &\cE_1(\delta)=\bigcap_{t=1}^{T}\left\{N_{1,t}\in\left[\sum_{j=1}^{t}\pi_j-C_{\delta,t,0}\sqrt{t},\sum_{j=1}^{t}\pi_j+C_{\delta,t,0}\sqrt{t}\right]\right\}, \label{eq:good_pi}\\
    &\cE_2(\delta)=\bigcap_{z\in\{0,1\}}\bigcap_{t=1}^{T}\left\{\sigma_z\in\left[\hat{\sigma}_{z,t}-\frac{C_{\delta,t,4}}{\sigma_zN_t(z)^{\beta/(2\beta+d)}},\hat{\sigma}_{z,t}+\frac{C_{\delta,t,4}}{\sigma_zN_t(z)^{\beta/(2\beta+d)}}\right]\right\},\label{eq:good_sigma}
\end{align}
where $C_{\delta,t,0}$ is defined in Lemma~\ref{lem:count}(ii). The definition of $\cE_2(\delta)$ in \eqref{eq:good_sigma} is introduced for notational convenience; in fact, it includes all the relations from {\bf Step 1} to {\bf Step 4} outlined above. By Lemma~\ref{lem:count}(ii) and \eqref{eq:C_3}, the good event $\cE=\cE_1(\delta)\bigcap\cE_2(\delta)$ occurs with probability at least $1-2\delta.$ Throughout the following Section~\ref{supsubsubsec:high_regret}, we assume the good event $\cE$ always holds.

\subsubsection{High-probability regret bound}
\label{supsubsubsec:high_regret}

In this section, we first establish a high-probability bound for the regret $\cR_T$, which serves as the basis for the expected regret bound derived in Section~\ref{supsubsubsec:proof_regret}. The analysis in this section proceeds under the good event $\cE$ defined above.

\begin{lemma}
    \label{lem:low_bound_pi}
    Suppose that the conditions of Theorem~\ref{thm:regret} hold. For $T_0\le t<T,$
    $$\underline{\pi}/2\le\pi_{t+1}\le1-\underline{\pi}/2.$$
\end{lemma}
\begin{proof}
    Under the event $\cE,$ we have $|\hat{\sigma}_{1,t}-\sigma_1|\le\Psi_{\sigma,1}(t,\delta)$ and $|\hat{\sigma}_{0,t}-\sigma_0|\ge\Psi_{\sigma,0}(t,\delta)$ by \eqref{eq:good_sigma}. Then for $t\ge T_0,$ we have
    \begin{equation}
        \label{eq:range_pi}
        \widetilde{\pi}_{t+1}\in\left[\frac{\sigma_1-\Psi_{\sigma,1}(t,\delta)}{\sigma_0+\Psi_{\sigma,0}(t,\delta)+\sigma_1-\Psi_{\sigma,1}(t,\delta)},\frac{\sigma_1+\Psi_{\sigma,1}(t,\delta)}{\sigma_0-\Psi_{\sigma,0}(t,\delta)+\sigma_1+\Psi_{\sigma,1}(t,\delta)}\right].
    \end{equation}
 Without loss of generality, we assume $\pi^*\le 1/2$, which also implies that $\underline{\pi}=\pi^*,\sigma_1\le\sigma_0$ and $\Psi_{\sigma,1}(t,\delta)\ge\Psi_{\sigma,0}(t,\delta)$. By the definition of $T_0$ in \eqref{eq:T0}, it follows that, for $t\ge T_0,$
    \begin{equation}
        \label{eq:l_delta_1}
        \Psi_{\sigma,1}(t,\delta)\le\frac{C_{\delta,t,4}}{\sigma_1(T_0/2)^{\beta/(2\beta+d)}}\le \frac{\sigma_1}{2},
    \end{equation}
    which implies that
    $$
    \frac{\sigma_1-\Psi_{\sigma,1}(t,\delta)}{\sigma_0+\Psi_{\sigma,0}(t,\delta)+\sigma_1-\Psi_{\sigma,1}(t,\delta)}\ge\frac{\sigma_1-\sigma_1/2}{\sigma_0+\sigma_1}\ge\frac{\pi^*}{2}=\frac{\underline{\pi}}{2}.
    $$
    Similar arguments can show that $\widetilde{\pi}_{t+1}\le1-\underline{\pi}/2$. By the definition of CLIP in \eqref{eq:clip}, we have $\underline{\pi}/2\le\pi_{t+1}\le1-\underline{\pi}/2$ also holds. The proof is completed. 
\end{proof}

\begin{lemma}
    \label{lem:pi_rate}
    Suppose that the conditions of Theorem~\ref{thm:regret} hold. For $T_0\le t<T,$
    $$
    |\pi_{t+1}-\pi^*|\le\frac{2^{2\beta/(2\beta+d)+2}}{(\sigma_0+\sigma_1)^2\underline{\pi}^{\beta/(2\beta+d)+1}}C_{\delta,t,4}t^{-\beta/(2\beta+d)}=\Theta(\underline{\pi}^{-\beta/(2\beta+d)-1}C_{\delta,t,4}t^{-\beta/(2\beta+d)}).
    $$
\end{lemma}
\begin{proof}
    By \eqref{eq:good_pi} and Lemma~\ref{lem:low_bound_pi}, we have $N_{1,t}\ge\sum_{j=1}^{t}\pi_j-C_{\delta,t,0}\sqrt{t}\ge\underline{\pi}t/2-C_{\delta,t,0}\sqrt{t}$ and also $N_{1,t}\ge\underline{\pi}t/2-C_{\delta,t,0}\sqrt{t}$. Since $T_0>16C_{\delta,T,0}^2/\underline{\pi}^2,$ it can be shown that $N_{1,t}\ge\underline{\pi}t/4$ and $N_{0,t}\ge\underline{\pi}t/4$. Moreover, the clipping threshold $\zeta_{t+1}=(1/2)\cdot (t+1)^{-\eta}<\underline{\pi}$ for $t\ge T_0$ as $T_0\ge (2\underline{\pi})^{-1/\eta}$. Combining above, it follows that
    $$
    \begin{aligned}
        \pi^*-\pi_{t+1}\le&\pi^*-\widetilde{\pi}_{t+1} \le\frac{\sigma_1\sigma_0^{-1}C_{\delta,t,4}N_{0,t}^{-\beta/(2\beta+d)}+\sigma_0\sigma_1^{-1}C_{\delta,t,4}N_{1,t}^{-\beta/(2\beta+d)}}{(\sigma_0+\sigma_1)(\sigma_0+\sigma_0^{-1}C_{\delta,t,4}N_{0,t}^{-\beta/(2\beta+d)}+\sigma_1-\sigma_1^{-1}C_{\delta,t,4}N_{1,t}^{-\beta/(2\beta+d)})}\\
        \le&\frac{(1-\underline{\pi})/\underline{\pi}+\underline{\pi}/(1-\underline{\pi})}{(\sigma_0+\sigma_1)(\sigma_0+\sigma_1/2)}C_{\delta,t,4}\left(\frac{\underline{\pi}t}{4}\right)^{-\beta/(2\beta+d)}
        \le\frac{4C_{\delta,t,4}}{(\sigma_0+\sigma_1)^2\underline{\pi}}\left(\frac{\underline{\pi}t}{4}\right)^{-\beta/(2\beta+d)}\\
        =&\frac{2^{2\beta/(2\beta+d)+2}}{(\sigma_0+\sigma_1)^2\underline{\pi}^{\beta/(2\beta+d)+1}}C_{\delta,t,4}t^{-\beta/(2\beta+d)}=\Theta(\underline{\pi}^{-\beta/(2\beta+d)-1}C_{\delta,t,4}t^{-\beta/(2\beta+d)}),
    \end{aligned}
    $$
    where the second inequality follows from the lower bound of $\pi_{t+1}$ in \eqref{eq:range_pi}, and the third inequality follows from \eqref{eq:l_delta_1}. Similarly, by the upper bound of $\pi_{t+1}$ in \eqref{eq:range_pi}, we have
    $$
    \begin{aligned}
        \pi_{t+1}-\pi^*\le&\frac{\sigma_1\sigma_0^{-1}C_{\delta,t,4}N_{0,t}^{-\beta/(2\beta+d)}+\sigma_0\sigma_1^{-1}C_{\delta,t,4}N_{1,t}^{-\beta/(2\beta+d)}}{(\sigma_0+\sigma_1)(\sigma_0-\sigma_0^{-1}C_{\delta,t,4}N_{0,t}^{-\beta/(2\beta+d)}+\sigma_1+\sigma_1^{-1}C_{\delta,t,4}N_{1,t}^{-\beta/(2\beta+d)})}\\
        \le&\frac{2^{2\beta/(2\beta+d)+2}}{(\sigma_0+\sigma_1)^2\underline{\pi}^{\beta/(2\beta+d)+1}}C_{\delta,t,4}t^{-\beta/(2\beta+d)}=\Theta(\underline{\pi}^{-\beta/(2\beta+d)-1}C_{\delta,t,4}t^{-\beta/(2\beta+d)}).
    \end{aligned}
    $$
    The proof is completed.
\end{proof}

\begin{lemma}
    \label{lem:pi_squa}
    Let $\epsilon=|\pi-\pi^*|$ for some $0\le\epsilon\le(\sigma_0\vee\sigma_1)/\{2(\sigma_0+\sigma_1)\}.$ Then, 
    $$
    \frac{\sigma_1^2}{\pi}+\frac{\sigma_0^2}{1-\pi}-\frac{\sigma_1^2}{\pi^*}-\frac{\sigma_0^2}{1-\pi^*}\le\frac{2(\sigma_0+\sigma_1)^2}{\underline{\pi}(1-\underline{\pi})}\epsilon^2=\Theta(\underline{\pi}^{-1}\epsilon^2).
    $$
\end{lemma}
\begin{proof}
    Without loss of generality, we assume that $\pi\le\pi^*$ and thus $\pi=\pi^*-\epsilon$. Then, it follows that
    $$
    \begin{aligned}
        &\frac{\sigma_1^2}{\pi}+\frac{\sigma_0^2}{1-\pi}-\frac{\sigma_1^2}{\pi^*}-\frac{\sigma_0^2}{1-\pi^*}\\
        =&\frac{\sigma_1^2}{\pi^*-\epsilon}+\frac{\sigma_0^2}{1-\pi^*+\epsilon}-\frac{\sigma_1^2}{\pi^*}-\frac{\sigma_0^2}{1-\pi^*}
        =\epsilon\left\{\frac{\sigma_1^2}{\pi^*(\pi^*-\epsilon)}-\frac{\sigma_0^2}{(1-\pi^*)(1-\pi^*+\epsilon)}\right\}\\
        =&\epsilon(\sigma_0+\sigma_1)^2\left\{\frac{\sigma_1^2}{\sigma_1^2-\sigma_1(\sigma_0+\sigma_1)\epsilon}-\frac{\sigma_0^2}{\sigma_0^2+\sigma_0(\sigma_0+\sigma_1)\epsilon}\right\}\\
        =&\epsilon^2(\sigma_0+\sigma_1)^3\left\{\frac{1}{\sigma_0+(\sigma_0+\sigma_1)\epsilon}+\frac{1}{\sigma_1-(\sigma_0+\sigma_1)\epsilon}\right\}\\
        =&\epsilon^2(\sigma_0+\sigma_1)^4\frac{1}{\{\sigma_0+(\sigma_0+\sigma_1)\epsilon\}\{\sigma_1-(\sigma_0+\sigma_1)\epsilon\}}\\
        \le&\epsilon^2(\sigma_0+\sigma_1)^2\frac{2(\sigma_0+\sigma_1)^2}{\sigma_0\sigma_1}=\frac{2(\sigma_0+\sigma_1)^2}{\underline{\pi}(1-\underline{\pi})}\epsilon^2
        =\Theta(\underline{\pi}\epsilon^2),
    \end{aligned}
    $$
    where the last line both follows from the fact $\underline{\pi}(1-\underline{\pi})=\sigma_0\sigma_1/(\sigma_0+\sigma_1)^2$.
\end{proof}

Now we are ready to drive a high-probability regret bound. By Theorem~\ref{thm:variance}, it can be seen that the regret has the decomposition:
{\small \begin{equation}
    \label{eq:regret_decomp}
    \begin{aligned}
        \cR_T=&\cR_{T_0}\\
        &+\sum_{t=T_0+1}^{T}\frac{\sigma_1^2}{\pi_t}+\frac{\sigma_0^2}{1-\pi_t}-\frac{\sigma_1^2}{\pi^*}-\frac{\sigma_0^2}{1-\pi^*}\\
        &+\sum_{t=T_0+1}^{T}\Big[2\eE_{(\bsX_t,\bsS_t)}\big[\{\mu_1(\bX_t,\bS_t)-\hat{\mu}_{1,t-1}(\bX_t,\bS_t)\}\{\mu_0(\bX_t,\bS_t)-\hat{\mu}_{0,t-1}(\bX_t,\bS_t)\}|\cF_{t-1}\big]\\
        &+\frac{1-\pi_t}{\pi_t}\eE_{(\bsX_t,\bsS_t)}\big[\{\mu_1(\bX_t,\bS_t)-\hat{\mu}_{1,t-1}(\bX_t,\bS_t)\}^2|\cF_{t-1}\big]\\
        &+\frac{\pi_t}{1-\pi_t}\eE_{(\bsX_t,\bsS_t)}\big[\{\mu_0(\bX_t,\bS_t)-\hat{\mu}_{0,t-1}(\bX_t,\bS_t)\}^2|\cF_{t-1}\big]\Big]\\
        =:&\cR_{T_0}+\cR_T'+\cR_T'',
    \end{aligned}
\end{equation}
}where $\cR_{T_0}=\sum_{t=1}^{T_0}\ell(a_t)-\sum_{t=1}^{T_0}\ell(a^*),\cR_T'=\sum_{t=T_0+1}^{T}\ell_1(\pi_t)-\sum_{t=T_0+1}^{T}\ell_1(\pi^*),$ and $\cR_T''=\sum_{t=T_0+1}^{T}\ell_2(a_t)$, where $a_t$ and $a^*$ are defined in Section \ref{subsec:regret}.  We derive the upper bound for $\cR_T$ in the following three steps. \\
(i) \textit{First}, for the cumulative loss during the initialization phase $t\in[T_0]$, we have $\pi_t=1/2,\hat{\mu}_{1,t-1}(\cdot,\cdot)=0$, and $\hat{\mu}_{0,t-1}(\cdot,\cdot)=0$. Then, it follows from $\pi^* = \sigma_1/(\sigma_1+\sigma_0)$ that
$$
\begin{aligned}
    \cR_{T_0}\le&T_0\{2\sigma_1^2+2\sigma_0^2-(\sigma_1+\sigma_0)^2\}\\
    &+2\sum_{t=1}^{T_0}\sum_{z\in\{0,1\}}\eE_{(\bsX_t,\bsS_t)}\big[\{\mu_z(\bX_t,\bS_t)-\hat{\mu}_{z,t-1}(\bX_t,\bS_t)\}^2|\cF_{t-1}\big]\\
    \le&T_0(\sigma_1-\sigma_0)^2+2\sum_{t=1}^{T_0}\sum_{z\in\{0,1\}}\eE_{\bsX_t}|m_{Y,z}(\bX_t)|^2+\gamma_z^2\eE\{\varepsilon_{S,t}^2(z)|\bX_t\}\\
    \le&T_0(\sigma_1-\sigma_0)^2+2T_0(\gamma_0^2\sigma_{S,0}^2+\gamma_1^2\sigma_{S,1}^2+M_{Y,0}^2+M_{Y,1}^2)=:\widetilde{C}_1T_0,
\end{aligned}
$$
where $\widetilde{C}_1=(\sigma_1-\sigma_0)^2+2(\gamma_0^2\sigma_{S,0}^2+\gamma_1^2\sigma_{S,1}^2+M_{Y,0}^2+M_{Y,1}^2).$ \\
(ii) \textit{Second}, for the cumulative loss from adaptive treatment allocation during the concentration phase, we notice that $(\sigma_0+\sigma_1)(\sigma_0\vee\sigma_1)\underline{\pi}=\sigma_0\sigma_1$. Then by Lemma~\ref{lem:pi_rate} and the definition of $T_0$ in \eqref{eq:T0}, for $t\ge T_0$, we have
$$
|\pi_{t+1}-\pi^*|\le\frac{2^{2\beta/(2\beta+d)+2}}{(\sigma_0+\sigma_1)^2\underline{\pi}^{\beta/(2\beta+d)+1}}C_{\delta,t,4}t^{-\beta/(2\beta+d)}\le \frac{\sigma_0\vee\sigma_1}{2(\sigma_0+\sigma_1)},
$$
which, together with Lemma \ref{lem:pi_squa}, implies that
$$
\begin{aligned}
    \cR_T'=&\sum_{t=T_0+1}^{T}\frac{\sigma_1^2}{\pi_t}+\frac{\sigma_0^2}{1-\pi_t}-\frac{\sigma_1^2}{\pi^*}-\frac{\sigma_0^2}{1-\pi^*}
    \le\frac{2^{4\beta/(2\beta+d)+5}}{(\sigma_0+\sigma_1)^2\underline{\pi}^{2\beta/(2\beta+d)+3}}\sum_{t=T_0+1}^{T}C_{\delta,t,4}^2t^{-2\beta/(2\beta+d)}\\
    \le&\widetilde{C}_{\delta,T,2}\underline{\pi}^{-2\beta/(2\beta+d)-3}T^{d/(2\beta+d)},
\end{aligned}
$$
where $\widetilde{C}_{\delta,T,2}=2^{4\beta/(2\beta+d)+5}d^{-1}(2\beta+d)(\sigma_0+\sigma_1)^{-2}C_{\delta,T,4}^2$.\\
(iii) \textit{Third}, for the cumulative estimation loss during the concentration phase, by Lemma~\ref{lem:low_bound_pi}, we have $(1-\pi_t)/\pi_t\le2/\underline{\pi}-1$ and $\pi_t/(1-\pi_t)\le2/\underline{\pi}-1$, which, together with the Cauchy-Schwarz inequality, implies
\begin{equation}
    \label{eq:R_T_prime}
    \cR_T''\le\frac{2}{\underline{\pi}}\sum_{t=T_0+1}^{T}\sum_{z\in\{0,1\}}\eE_{(\bsX_t,\bsS_t)}\big[\{\mu_z(\bX_t,\bS_t)-\hat{\mu}_{z,t-1}(\bX_t,\bS_t)\}^2|\cF_{t-1}\big].
\end{equation}
Then, similar to \eqref{eq:hat_mu_deco_asy}, it follows that
\begin{equation}
    \label{eq:hat_mu_deco}
    \begin{aligned}
        &\eE_{(\bsX_t,\bsS_t)}\big[\{\mu_1(\bX_t,\bS_t)-\hat{\mu}_{1,t-1}(\bX_t,\bS_t)\}^2|\cF_{t-1}\big]\\
        \le&4\gamma_1^2\eE_{\bsX_t}\big[\{\hat{m}_{S,1,t-1}(\bX_t)-m_{S,1}(\bX_t)\}^2|\cF_{t-1}\big]\\
        &+4(\hat{\gamma}_{1,t-1}-\gamma_1)^2\eE_{\bsX_t}\big[\{\hat{m}_{S,1,t-1}(\bX_t)-m_{S,1}(\bX_t)\}^2|\cF_{t-1}\big]\\
        &+4\eE_{\bsX_t}\big[\{\hat{m}_{Y,1,t-1}(\bX_t)-m_{Y,1}(\bX_t)\}^2|\cF_{t-1}\big]+4\sigma_{S,1}^2(\hat{\gamma}_{1,t-1}-\gamma_1)^2,
    \end{aligned}
\end{equation}
and $\eE_{(\bsX_t,\bsS_t)}\big[\{\mu_0(\bX_t,\bS_t)-\hat{\mu}_{0,t-1}(\bX_t,\bS_t)\}^2|\cF_{t-1}\big]$ can be decomposed analogously. According to Lemma~\ref{lem:pi_rate}, $N_{1,t}\ge\underline{\pi}t/4$ and $N_{0,t}\ge\underline{\pi}t/4$ under the event $\cE$. Then by {\bf Step 2} of Section~\ref{supsubsubsec:event}, it follows that:
$$
\begin{aligned}
        (\hat{\gamma}_{1,t-1}-\gamma_1)^2\le& \frac{C_{\delta,t,2}^2}{N_{1,t}^{2\beta/(2\beta+d)}}\le C_{\delta,t,2}^2(4/\underline{\pi})^{2\beta/(2\beta+d)}t^{-2\beta/(2\beta+d)},\\
        (\hat{\gamma}_{0,t-1}-\gamma_1)^2\le& \frac{C_{\delta,t,2}^2}{N_{0,t}^{2\beta/(2\beta+d)}}\le C_{\delta,t,2}^2(4/\underline{\pi})^{2\beta/(2\beta+d)}t^{-2\beta/(2\beta+d)}.
\end{aligned}
$$
By the definitions of $T_0$ and $C_{\delta,t,2}$, we can also show that $(\hat{\gamma}_{z,t-1}-\gamma_z)^2\le2^{-6}\sigma_0^2\sigma_1^2$ for $z\in\{0,1\}$ and $t>T_0$. Denote $\widetilde{\ell}_{S,z,t}=\eE_{\bsX_t}[\{\hat{m}_{S,z,t-1}(\bX_t)-m_{S,z}(\bX_t)\}^2|\cF_{t-1}]$ and $\widetilde{\ell}_{Y,z,t}=\eE_{\bsX_t}[\{\hat{m}_{Y,z,t-1}(\bX_t)-m_{Y,z}(\bX_t)\}^2|\cF_{t-1}]$ for $z\in\{0,1\}$ and $t>T_0$. Then, combine the above results, by basic algebraic calculations, we have
$$
\begin{aligned}
    \cR_T''\le&2\underline{\pi}^{-1}\cdot 2\cdot 4\sigma_{S,\max}^2C_{\delta,T,2}^2(4/\underline{\pi})^{2\beta/(2\beta+d)}\sum_{t=1}^{T}t^{-2\beta/(2\beta+d)}\\
    &+2\underline{\pi}^{-1}\cdot 4(\gamma_1^2+2^{-6}\sigma_0^2\sigma_1^2)\sum_{t=T_0+1}^{T}\widetilde{\ell}_{S,1,t}+2\underline{\pi}^{-1}\cdot 4(\gamma_0^2+2^{-6}\sigma_0^2\sigma_1^2)\sum_{t=T_0+1}^{T}\widetilde{\ell}_{S,0,t}\\
    &+2\underline{\pi}^{-1}\cdot 4\sum_{t=T_0+1}^{T}\widetilde{\ell}_{Y,1,t}+2\underline{\pi}^{-1}\cdot 4\sum_{t=T_0+1}^{T}\widetilde{\ell}_{Y,0,t}\\
    \le&\widetilde{C}_{\delta,T,3}\underline{\pi}^{-2\beta/(2\beta+d)-1}T^{d/(2\beta+d)}\\
    &+8\underline{\pi}^{-1}(\gamma_{\max}^2+2^{-6}\sigma_0^2\sigma_1^2)\sum_{z\in\{0,1\}}\sum_{t=T_0+1}^{T}\widetilde{\ell}_{S,z,t}
    +8\underline{\pi}^{-1}\sum_{z\in\{0,1\}}\sum_{t=T_0+1}^{T}\widetilde{\ell}_{Y,z,t}.
\end{aligned}
$$
where $\widetilde{C}_{\delta,T,3}=2^{4\beta/(2\beta+d)+4}d^{-1}(2\beta+d)\sigma_{S,\max}^2C_{\delta,T,2}^2$.

Combining above, we achieve the final result:
\begin{equation}
    \label{eq:high_bound}
    \begin{aligned}
        \cR_T\le&\widetilde{C}_1T_0
        +\widetilde{C}_{\delta,T,2}\underline{\pi}^{-2\beta/(2\beta+d)-3}T^{d/(2\beta+d)}
        +\widetilde{C}_{\delta,T,3}\underline{\pi}^{-2\beta/(2\beta+d)-1}T^{d/(2\beta+d)}\\
        &+8\underline{\pi}^{-1}(\gamma_{\max}^2+2^{-6}\sigma_0^2\sigma_1^2)\sum_{z\in\{0,1\}}\sum_{t=T_0+1}^{T}\widetilde{\ell}_{S,z,t}
        +8\underline{\pi}^{-1}\sum_{z\in\{0,1\}}\sum_{t=T_0+1}^{T}\widetilde{\ell}_{Y,z,t},
    \end{aligned}
\end{equation}
which holds with probability at least $1-2\delta$.

\subsubsection{Finite-sample regret bound}
\label{supsubsubsec:proof_regret}

The following lemma derives an upper bound for the mean squared error implied by an exponential-tail concentration inequality.

\begin{lemma}
    \label{lem:MSE_bound}
    Suppose that $X$ is a random variable satisfying the tail inequality: $\eP(|X|\ge\epsilon)\le 2\exp\{-A(\epsilon-B)^2\}$ for every $\epsilon>B$ and for some $A,B>0.$ Then, $\eE X^2\le 2B^2+(\pi+2)/A$.
\end{lemma}
\begin{proof}
    Note that $\eE X^2=\int_{0}^{\infty}\eP(X^2\ge u){\rm d}u=\int_{0}^{\infty}\eP(|X|\ge\sqrt{u}){\rm d}u=\int_{0}^{\infty}2\epsilon\eP(|X|\ge\epsilon){\rm d}\epsilon=\int_{0}^{B}2\epsilon\eP(|X|\ge\epsilon){\rm d}\epsilon+\int_{B}^{\infty}2\epsilon\eP(|X|\ge\epsilon){\rm d}\epsilon=:I_1+I_2.$ For the first term, $I_1\le\int_{0}^{B}2\epsilon{\rm d}\epsilon=B^2.$ For the second term, given the exponential tail inequality, we have
    $$
    \begin{aligned}
        I_2\le&\int_{B}^{\infty}2\epsilon\cdot 2\exp\{-A(\epsilon-B)^2\}{\rm d}\epsilon=4\int_{0}^{\infty}(v+B)\exp(-Av^2){\rm d}v\\
        =&4B\int_{0}^{\infty}\exp(-Av^2){\rm d}v+4\int_{0}^{\infty}v\exp(-Av^2){\rm d}v\\
        =&4B\cdot\frac{1}{2}\sqrt{\frac{\pi}{A}}+4\cdot\frac{1}{2A}=2B\sqrt{\pi/A}+2/A\\
        \le&B^2+(\pi+2)/A,
    \end{aligned}
    $$
    where the third equality follows from the property of Gaussian integral, and the last inequality follows from the AM-GM inequality. Combining above, we have the desired result $\eE X^2\le 2B^2+(\pi+2)/A$.
\end{proof}

Now we are ready to show Theorem~\ref{thm:regret}. By combining Proposition~\ref{propos:NW_est}, 
Lemma~\ref{lem:MSE_bound} and Condition~\ref{cond:X}, it follows that
{\footnotesize $$
\begin{aligned}
    \eE\big(\widetilde{\ell}_{S,1,t+1}|\cE\big)=&\eE\big[\{\hat{m}_{S,1,t}(\bX_{t+1})-m_{S,1}(\bX_{t+1})\big\}^2|\cE\big]\le \widetilde{C}_4't^{-2\beta/(2\beta+d)},\\
    \eE\big(\widetilde{\ell}_{Y,1,t+1}|\cE\big)=&\eE\big[\{\hat{m}_{Y,1,t}(\bX_{t+1})-m_{Y,1}(\bX_{t+1})\}^2|\cE\big]\le \widetilde{C}_4'N_{1,t}^{-2\beta/(2\beta+d)}\le \widetilde{C}_4'(4/\underline{\pi})^{2\beta/(2\beta+d)}t^{-2\beta/(2\beta+d)},\\
    \eE\big(\widetilde{\ell}_{S,0,t+1}|\cE\big)=&\eE\big[\{\hat{m}_{S,0,t}(\bX_{t+1})-m_{S,0}(\bX_{t+1})\}^2|\cE\big]\le \widetilde{C}_4't^{-2\beta/(2\beta+d)},\\
    \eE\big(\widetilde{\ell}_{Y,0,t+1}|\cE\big)=&\eE\big[\{\hat{m}_{Y,0,t}(\bX_{t+1})-m_{Y,0}(\bX_{t+1})\}^2|\cE\big]\le \widetilde{C}_4'N_{0,t}^{-2\beta/(2\beta+d)}\le \widetilde{C}_4'(4/\underline{\pi})^{2\beta/(2\beta+d)}t^{-2\beta/(2\beta+d)},
\end{aligned}
$$}almost surely, where $\widetilde{C}_4'=2dL^2+3\cdot 2^{1-d}\underline{k}^{-2}\underline{c}^{-1}\bar{k}^2v_{\max}^2$. Note that $\cR_T\le T\{\ell_{\max}-\ell(a^*)\}$ almost surely, and $\eP(\cE^c)\le2\delta=2T^{-2\beta/(2\beta+d)}$. Then, by combining \eqref{eq:high_bound}, we have
$$
\begin{aligned}
    \eE\cR_T\le&\eE(\cR_T|\cE)+\eE(\cR_T|\cE^c)\\
    \le&\widetilde{C}_1T_0+\widetilde{C}_{\delta,T,2}\underline{\pi}^{-2\beta/(2\beta+d)-3}T^{d/(2\beta+d)}
    +\widetilde{C}_{\delta,T,3}\underline{\pi}^{-2\beta/(2\beta+d)-1}T^{d/(2\beta+d)}\\
    &+16\underline{\pi}^{-1}(\gamma_{\max}^2+2^{-6}\sigma_0^2\sigma_1^2)\widetilde{C}_4'\sum_{t=1}^{T}t^{-2\beta/(2\beta+d)}
    +16\underline{\pi}^{-1}(4/\underline{\pi})^{2\beta/(2\beta+d)}\widetilde{C}_4'\sum_{t=1}^{T}t^{-2\beta/(2\beta+d)}\\
    &+T\{\ell_{\max}-\ell(a^*)\}\cdot 2T^{-2\beta/(2\beta+d)}\\
    =&\widetilde{C}_1T_0+\widetilde{C}_{\delta,T,2}\underline{\pi}^{-2\beta/(2\beta+d)-3}T^{d/(2\beta+d)}\\
    &+\widetilde{C}_{\delta,T,3}\underline{\pi}^{-2\beta/(2\beta+d)-1}T^{d/(2\beta+d)}
    +\widetilde{C}_4\underline{\pi}^{-1}T^{d/(2\beta+d)}+\widetilde{C}_5\underline{\pi}^{-2\beta/(2\beta+d)-1}T^{d/(2\beta+d)}\\
    &+2\{\ell_{\max}-\ell(a^*)\}T^{d/(2\beta+d)},
\end{aligned}
$$
where $\widetilde{C}_1=(\sigma_1-\sigma_0)^2+2(\gamma_0^2\sigma_{S,0}^2+\gamma_1^2\sigma_{S,1}^2+M_{Y,0}^2+M_{Y,1}^2),$ 
$\widetilde{C}_{\delta,T,2}=2^{4\beta/(2\beta+d)+5}d^{-1}(2\beta+d)(\sigma_0+\sigma_1)^{-2}C_{\delta,T,4}^2,$
$\widetilde{C}_{\delta,T,3}=2^{4\beta/(2\beta+d)+4}d^{-1}(2\beta+d)\sigma_{S,\max}^2C_{\delta,T,2}^2,$ 
$\widetilde{C}_4=2^4d^{-1}(2\beta+d)(\gamma_{\max}^2+2^{-6}\sigma_0^2\sigma_1^2)(2dL^2+3\cdot 2^{1-d}\underline{k}^{-2}\underline{c}^{-1}\bar{k}^2v_{\max}^2),$ and $\widetilde{C}_5=2^{4\beta/(2\beta+d)+4}d^{-1}(2\beta+d)(2dL^2+3\cdot 2^{1-d}\underline{k}^{-2}\underline{c}^{-1}\bar{k}^2v_{\max}^2)$. Hence, $\widetilde{C}_1,\widetilde{C}_4$ and $\widetilde{C}_5$ are constants independent of $\delta$ and $T$, $\widetilde{C}_{\delta,T,2}\asymp(\log \delta^{-1}T)^2,$ and $\widetilde{C}_{\delta,T,3}\asymp\log \delta^{-1}T$. The proof is completed.

\subsection{Proof of Theorem~\ref{thm:asy_cs}}

To show Theorem~\ref{thm:asy_cs}, we first introduce the uncertainty quantification tools developed in \textcolor{blue}{Waudby-Smith et al.} (\textcolor{blue}{2024a}). Suppose $\{R_t\}_{t\ge1}$ is a sequence of random variables with conditional means and variances given by $\mu_t:=\eE(R_t|\cF_{t-1})$ and $\sigma_t^2:=\var(R_t|\cF_{t-1})$, respectively. We require that the following three conditions hold. \\
(L1): $\widetilde{V}_t:=\sum_{j=1}^{t}\sigma_t^2\to\infty$ a.s. as $t\to\infty$.\\
(L2): The Lindeberg-type uniform integrability holds, which means that there exists some $0<\iota<1$ such that $\sum_{t=1}^{\infty}\widetilde{V}_t^{-\iota}\eE\big[(R_t-\mu_t)^2I\{(R_t-\mu_t)^2>\widetilde{V}_t^{\iota}\}|\cF_{t-1}\big]<\infty$ a.s.\\
(L3): Let $\hat{\sigma}_t^2$ be an estimator of $\widetilde{\sigma}_t^2:=t^{-1}\sum_{j=1}^{t}\sigma_j^2$ constructed using $R_1,\dots,R_t$ such that $\hat{\sigma}_t^2/\widetilde{\sigma}_t^2\to_{a.s.}1$ as $t\to\infty$.

Let $\hat{\mu}_t:=t^{-1}\sum_{j=1}^{t}R_j$ be the estimate of the running average conditional mean $\widetilde{\mu}_t:=t^{-1}\sum_{j=1}^{t}\mu_j.$ The following lemma provides a Lindeberg-type Gaussian mixture martingale asymptotic confidence sequence for $\widetilde{\mu}_t$. 

\begin{lemma}
    \label{lem:asycs}
    (Proposition 2.5 of \textcolor{blue}{Waudby-Smith et al.} (\textcolor{blue}{2024a})) Suppose that conditions (L1)--(L3) hold. Then, for any pre-specified parameter $\varrho>0$,
    $$
    \left(\hat{\mu}_t\pm\sqrt{\frac{t\hat{\sigma}_t^2\varrho^2+1}{t^2\varrho^2}\log\left(\frac{t\hat{\sigma}_t^2\varrho^2+1}{\alpha^2}\right)}\right)
    $$
    forms a $(1-\alpha)$ asymptotic confidence sequence for $\widetilde{\mu}_t.$
\end{lemma}

Now we are ready to show Theorem~\ref{thm:asy_cs}. Denote $\widetilde{V}_t=\sum_{j=1}^{t}\eE(u_j^2|\cF_{j-1})$ and $\widetilde{\varsigma}_t^2=t^{-1}\widetilde{V}_t$. Then we check Conditions (L1)--(L3) above. \textit{First}, by Theorem~\ref{thm:variance} and noting that $\pi^*=\sigma_1/(\sigma_1+\sigma_0)$, we have $\eE(u_j^2|\cF_{j-1})\ge(\sigma_0+\sigma_1)^2$ a.s. for all $j\ge1,$ which implies that $\widetilde{V}_t\ge t(\sigma_0+\sigma_1)^2\to\infty$ a.s., and thus Condition (L1) is satisfied. \textit{Second}, by Lemma~\ref{lem:lyapunov}, it follows that, for $\theta>2\eta/(1-2\eta)$,
$$
\sum_{t=1}^{\infty}\frac{\eE(|u_t|^{2+\theta}|\cF_{t-1})}{\widetilde{V}_t^{1+\theta/2}}\le\frac{1}{(\sigma_0+\sigma_1)^{2+\theta}}\sum_{t=1}^{\infty}\frac{\eE(|u_t|^{2+\theta}|\cF_{t-1})}{t^{1+\theta/2}}<\infty
$$
almost surely, which corresponds to the Lyapunov-type condition given in Section 2.4 of \textcolor{blue}{Waudby-Smith et al.} (\textcolor{blue}{2024a}). As shown in Appendix B.5 of \textcolor{blue}{Waudby-Smith et al.} (\textcolor{blue}{2024a}), the Lyapunov-type condition implies the Lindeberg-type condition, so Condition (L2) is satisfied. For Condition (L3), note that
$$
\frac{\hat{\varsigma}_t^2}{\widetilde{\varsigma}_t^2}=\frac{\sum_{j=1}^{t}(\phi_j-\hat{\tau}_t)^2}{\sum_{j=1}^{t}\eE(u_j^2|\cF_{t-1})}=\frac{t^{-1}\sum_{j=1}^{t}u_j^2+(\hat{\tau}_t-\tau_0)^2}{t^{-1}\sum_{j=1}^{t}\eE(u_j^2|\cF_{t-1})}.
$$
By the proof of Theorem~\ref{thm:CI}, we have $t^{-1}\sum_{j=1}^{t}u_j^2\to_{a.s.}t^{-1}\sum_{j=1}^{t}\eE(u_j^2|\cF_{t-1})$ as $t\to\infty$, which, together with $\hat{\tau}_t\to_{a.s.}\tau_0$ as $t\to\infty$ by Theorem~\ref{thm:asy_nor}, implies that $\hat{\varsigma}_t^2/\widetilde{\varsigma}_t^2\to_{a.s.}1$ as $t\to\infty$. Hence, Condition (L3) is satisfied. Combining above, by using Lemma~\ref{lem:asycs}, the proof is completed.

\subsection{Proof of Theorem~\ref{thm:finite_cs}}
To prove Theorem~\ref{thm:finite_cs}, we first introduce two technical lemmas. 
\begin{lemma}
    \label{lem:supermart}
    (Lemma 1 of \textcolor{blue}{Waudby-Smith et al.} (\textcolor{blue}{2024b})) Let $X$ and $\hat{X}$ be $\cF$-adapted processes such that $X_t-\hat{X}_{t-1}\ge-1$ almost surely for all $t.$ Denoting $\mu_t:=\eE(X_t|\cF_{t-1})$, we have that for any $(0,1)$-valued predictable process $\{\lambda_t\}_{t\ge1},$
    $$
    M_t:=\exp\left\{\sum_{j=1}^{t}\lambda_j(X_t-\mu_t)-\sum_{j=1}^{t}\psi_E(\lambda_j)\big(X_t-\hat{X}_{t-1}\big)^2\right\},
    $$
     forms a test supermartingale, where $\psi_E(\lambda):=-\log(1-\lambda)-\lambda.$
\end{lemma}

\begin{lemma}
    \label{lem:ville}
    (Ville's inequality, Exercise 4.8.2 of \textcolor{blue}{Durrett} (\textcolor{blue}{2019})) Let $\{X_n\}_{n\ge1}$ be a nonnegative supermartingale. Then, for any $\epsilon>0,\eP\big(\sup_{n\ge1}X_n\ge\epsilon\big)\le\eE(X_1)/\epsilon.$
\end{lemma}

Now we are ready to show Theorem~\ref{thm:finite_cs}. Note that for each $z\in\{0,1\}$, $Y_t,$ $S_t(z),$ $\hat{m}_{S,z,t-1}(\bX_t),$ and $\hat{\mu}_{z,t-1}(\bX_t,\bS_t)\in[0,1]$, which together imply that $\phi_t\ge -k_t$ a.s. Thus,
$$
\xi_t-\hat{\xi}_{t-1}\ge\frac{\phi_t}{k_t+1}-\frac{1}{k_t+1}\ge-1
$$
almost surely. Given the facts that $\eE(\xi_t|\cF_{t-1})=\tau_0/(k_t+1)$ and $\{\lambda_{\alpha,t}\}_{t\ge1}$ is a $(0,1)$-valued predictable sequence, Lemma~\ref{lem:supermart} implies that
\begin{equation}
    \label{eq:supermartingale}
    M_t:=\exp\left\{\sum_{j=1}^{t}\lambda_{\alpha,j}\left(\xi_j-\frac{\tau_0}{k_j+1}\right)-\sum_{j=1}^{t}\big(\xi_j-\hat{\xi}_{j-1}\big)^2\psi_E(\lambda_{\alpha,j})\right\}
\end{equation}
forms a test supermartingale. Hence, by Lemma~\ref{lem:ville}, we have 
$$
\eP(\exists t\ge1:M_t\ge 1/\alpha)\le\alpha.
$$ 
By arguments similar to those in the proof of Proposition 2 of \textcolor{blue}{Waudby-Smith et al.} (\textcolor{blue}{2024b}), we can show that, with probability at least $(1-\alpha)$ and for all $t\ge1,$
$$
M_t<1/\alpha\iff \tau_0>\frac{\sum_{j=1}^{t}\lambda_{\alpha,j}\xi_j}{\sum_{j=1}^{t}\lambda_{\alpha,j}/(k_j+1)}-\frac{\log(1/\alpha)+\sum_{j=1}^{t}(\xi_j-\hat{\xi}_{j-1})^2\psi_E(\lambda_{\alpha,j})}{\sum_{j=1}^{t}\lambda_{\alpha,j}/(k_j+1)}.
$$
Finally, using the mirroring trick as discussed in Remark 3 of \textcolor{blue}{Waudby-Smith et al.} (\textcolor{blue}{2024b}) completes the proof.

\section{Extensions}
\label{supsec:method}

\subsection{Profile least squares estimation}
\label{supsubsec:profile}

In this section, we provide the profile least squares method to estimate the partially linear models in Section~\ref{subsec:procedure}. Recall that $m_0(\bx)=m_{Y,0}(\bx)-\gamma_0m_{S,0}(\bx)$ and $m_1(\bx)=m_{Y,1}(\bx)-\gamma_1m_{S,1}(\bx)$. The main idea is that, for any given values of $\gamma_0$ and $\gamma_1$, we treat $m_0(\cdot)$ and $m_1(\cdot)$ as nuisance functions. We first estimate $m_0(\cdot)$ and $m_1(\cdot)$ using a nonparametric method, and then plug the estimates into the model to obtain residuals. The coefficients $\gamma_0$ and $\gamma_1$ are then estimated by minimizing the sum of squared residuals. The detailed procedure is outlined below for each $z\in\{0,1\}$.

\begin{itemize}
    \item[(i)] For a fixed value of $\gamma_z$, a Nadaraya-Watson estimator of $m_z(\cdot)$ is
    $$
    \hat{m}_{z,\gamma_z,t}(\bx)=\frac{\sum_{j=1}^{t}K_h(\bx-\bX_j)\{Y_j-\gamma_zS_j(z)\}I(Z_j=z)}{\sum_{j=1}^{t}K_h(\bx-\bX_j)I(Z_j=z)},\quad\bx\in\eR^d,
    $$
    where $K_h(\bx)=K(\bx/h)/h$ is a kernel function with bandwidth $h>0$. In particular, define $\hat{m}_{z,\gamma_z,t}(\bx)=0$ if $\sum_{j=1}^{t}K_h(\bx-\bX_j)I(Z_j=z)=0.$
    
    \item[(ii)] Calculate the residuals in each $G_{z,t}=\{j\in[t]:Z_j=z\}$:
    $$\hat{\varepsilon}_{j,\gamma_z,t}=Y_j-\gamma_zS_j(z)-\hat{m}_{z,\gamma_z,t}(\bX_j),\quad j\in G_{z,t},
    $$  
    and define the profile least squares criterion
    $L_{z,t}(\gamma_z)=\sum_{j=1}^{t}\hat{\varepsilon}_{j,\gamma_z,t}^2I(Z_j=z).$
    
    \item[(iii)] Find the value of $\gamma_z$ that minimizes the profile criterion:
    $$
    \hat{\gamma}_{z,t}=\argmin_{\gamma_z}L_{z,t}(\gamma_z),
    $$
    which can be computed using optimization algorithms such as BFGS or L-BFGS. 
    
    \item[(iv)] The conditional mean functions and variances can then be estimated as:
    $$
    \begin{aligned}
        &\hat{\mu}_{z,t}(\bx,\bs)=\hat{m}_{z,\hat{\gamma}_{z,t},t}(\bx)+\hat{\gamma}_{z,t}s(z),\\
        &\hat{\sigma}_{z,t}^2=\frac{1}{\sum_{j=1}^{t}I(Z_j=z)}\sum_{j=1}^{t}\{Y_j-\hat{\mu}_{z,t}(\bX_j,\bS_j)\}^2I(Z_j=z).
    \end{aligned}
    $$
\end{itemize}

It is worth noting that the computational complexity of the profile least squares method is substantially higher than that of the residual-based method introduced in Section~\ref{subsec:procedure}, since the nonparametric component must be re-estimated at each optimization step when updating $\gamma_0$ and $\gamma_1$. Table~\ref{tab:comparison_method} provides a side-by-side comparison of the two estimation methods, and Figure~\ref{fig:profile} in Section~\ref{supsec:sim} presents their numerical performance across different nonparametric regression estimators.

\renewcommand{\arraystretch}{1}
\begin{table}[H]
\centering
\caption{\small Comparisons between the two estimation methods.}
\label{tab:comparison_method}
\vspace{1em}
\begin{tabular}{c|cc}
\hline
& residual-based & profile least squares \\ 
\hline
method to estimate $\hat{\gamma}_{z,t}$ & linear regression & profile least squares \\
need optimization? & no & yes \\
computational complexity & low & high  \\
need to estimate $m_{Y,z}$ and $m_{S,z}$? & yes & no \\
directly estimate $m_z$ & no & yes \\
\hline
\end{tabular}
\end{table}

\subsection{Regret bound for the linear model}
\label{supsubsec:bound_linear}

This section first derives the estimators for the parameters in linear Model~\eqref{eq:model_linear_cond_XS}, serving as a supplement to Section~\ref{subsec:linear}. Specifically, at each Stage $t$, denote $G_{z,t}=\{j\in[t]:Z_j=z\}$, $\widecheck{\bX}_t=(\bX_1,\dots,\bX_t)^{\T}\in\eR^{t\times d}$, $\widecheck{\bS}_{z,t}=\big(S_1(z),\dots,S_t(z)\big)^{\T}\in\eR^{t\times 1}$, $\widecheck{\bX}_{G_{z,t},t}=(\bX_j)_{j\in G_{z,t}}^{\T}\in\eR^{|G_{z,t}|\times d}$, and $\widecheck{\bY}_{G_{z,t}}=(Y_j)_{j\in G_{z,t}}\in\eR^{|G_{z,t}|\times 1}.$ The least squares estimators of $\balpha_{S,z},\balpha_{Y,z},$ and $\gamma_z$ are given by 
\begin{align*}
    &\hat{\balpha}_{S,z,t}=\big(\widecheck{\bX}_t^{\T}\widecheck{\bX}_t\big)^{-1}\widecheck{\bX}_t^{\T}\widecheck{\bS}_{z,t}, \quad \hat{\balpha}_{Y,z,t}=\big(\widecheck{\bX}_{G_{z,t},t}^{\T}\widecheck{\bX}_{G_{z,t},t}\big)^{-1}\widecheck{\bX}_{G_{z,t},t}^{\T}\widecheck{\bY}_{G_{z,t}},\\
    &\hat{\gamma}_{z,t}=\sum_{j\in G_{z,t}}\{Y_j(z)-\hat{\balpha}_{Y,z,t}^{\T}\bX_j\}\{S_j(z)-\hat{\balpha}_{S,z,t}^{\T}\bX_j\}/\sum_{j\in G_{z,t}}\{S_j(z)-\hat{\balpha}_{S,z,t}^{\T}\bX_j\}^2.
\end{align*}

Then, we establish a finite-sample regret bound for our proposed method under the linear model. The regret $\cR_T$ is defined similar to that in Section~\ref{subsec:regret}. The assumption on the model errors $\{\varepsilon_{Y,t}(z)\}$ and $\{\varepsilon_{S,t}(z)\}$ still follows Condition \ref{cond:varep} for the nonparametric nonlinear model. Regarding the covariates $\bX_t$, Condition~\ref{cond:X} is replaced by the following assumption.

\begin{condition}
    \label{cond:X_linear}
    The covariate domain $\cX$ is a compact set. $\underline{\lambda}:=\lambda_{\min}(\eE\bX_t\bX_t^{\T})>0,$ where $\lambda_{\min}(\cdot)$ denotes the smallest eigenvalue of a symmetric matrix.
\end{condition}

Similar to Section~\ref{supsubsubsec:event}, we first define a good event. Note that
$$
\begin{aligned}
    \hat{\balpha}_{S,0,t}-\balpha_{S,0}=&\big(\widecheck{\bX}_t^{\T}\widecheck{\bX}_t\big)^{-1}\widecheck{\bX}_t^{\T}\widecheck{\bS}_{0,t}-\big(\widecheck{\bX}_t^{\T}\widecheck{\bX}_t\big)^{-1}\widecheck{\bX}_t^{\T}\widecheck{\bX}_t\balpha_{S,0}\\
    =&\big(\widecheck{\bX}_t^{\T}\widecheck{\bX}_t\big)^{-1}\widecheck{\bX}_t^{\T}\big(\widecheck{\bS}_{0,t}-\widecheck{\bX}_t\balpha_{S,0}\big)=\big(\widecheck{\bX}_t^{\T}\widecheck{\bX}_t\big)^{-1}\widecheck{\bX}_t^{\T}\widetilde{\bveps}_{S,t}(0),
\end{aligned}
$$
where $\widetilde{\bveps}_{S,t}(0)=\big(\varepsilon_{S,1}(0),\dots,\varepsilon_{S,t}(0)\big)\in\eR^{t\times 1}$. Since the covariate domain $\cX$ is compact by Condition~\ref{cond:X_linear}, there exists some constant $\bar{\lambda}>0$ such that $\lambda_{\max}(\bX_t\bX_t)\le \bar{\lambda}$ a.s. for each $t\in[T],$ where $\lambda_{\max}(\cdot)$ denotes the largest eigenvalue of a symmetric matrix. Then, by the matrix Chernoff inequality (Theorem 5.1.1 of \textcolor{blue}{Tropp} (\textcolor{blue}{2015})), 
$$
\eP\left\{\lambda_{\min}\Big(t^{-1}\sum_{j=1}^{t}\bX_j\bX_j^{\T}\Big)\le \underline{\lambda}/2\right\}\le d\exp\left(-\frac{t\underline{\lambda}}{8\bar{\lambda}}\right).
$$
Thus, by combining Conditions~\ref{cond:varep} and \ref{cond:X_linear}, and Lemma~\ref{lem:bounded}, there exists some constant $c_1,M_1>0$ such that, for a given $\bx\in\cX$ and every small $\epsilon>0$, we have $\eP(|\hat{\balpha}_{S,0,t}^{\T}\bx-\balpha_{S,0}^{\T}\bx|\ge\epsilon)\le c_1\exp\left(-t\epsilon^2M_1\right)$. Similar arguments apply to $\hat{\balpha}_{S,1,t},\hat{\balpha}_{Y,0,t}$, and $\hat{\balpha}_{Y,1,t}$. Then, by applying Proposition~\ref{propos:sigma_est} 
with $\epsilon_1\asymp N_{z,t}^{-1/2}$ and $\epsilon_2\asymp t^{-1/2}$, we can derive the deviation inequalities for $\hat{\gamma}_{z,t}$ and $\hat{\sigma}_{z,t}$ with $z\in\{0,1\}$. Combining the above results, we can define quantities $C_{\delta,t,1},C_{\delta,t,2},C_{\delta,t,3},C_{\delta,t,4}$ such that with probability at least $1-\delta/T$, the following relations hold:
$$
\begin{aligned}
    &\max_{j\in G_{z,t}}|\hat{\balpha}_{Y,z,t}^{\T}\bX_j-\balpha_{Y,z}^{\T}\bX_j|\le C_{\delta,t,1}N_{z,t}^{-1/2},\quad
    \max_{j\in [t]}|\hat{\balpha}_{S,z,t}^{\T}\bX_j-\balpha_{S,z}^{\T}\bX_j|\le C_{\delta,t,1}t^{-1/2},\\
    &|\hat{\gamma}_{z,t}-\gamma_z|\le C_{\delta,t,2}N_{z,t}^{-1/2},\quad
    \max_{j\in G_{z,t}}|\hat{\mu}_{z,t}(\bX_j,\bS_j)-\mu_{z}(\bX_j,\bS_j)|\le  C_{\delta,t,3}N_{z,t}^{-1/2},\\
    &|\hat{\sigma}_{z,t}-\sigma_z|\le C_{\delta,t,4}\sigma_z^{-1}N_{z,t}^{-1/2},\quad z\in\{0,1\}.
\end{aligned}
$$
Then, define the following events:
\begin{align}
    &\cE_1(\delta)=\bigcap_{t=1}^{T}\left\{N_{1,t}\in\left[\sum_{j=1}^{t}\pi_j-C_{\delta,t,0}\sqrt{t},\sum_{j=1}^{t}\pi_j+C_{\delta,t,0}\sqrt{t}\right]\right\},\\
    &\cE_2(\delta)=\bigcap_{z\in\{0,1\}}\bigcap_{t=1}^{T}\left\{\sigma_z\in\left[\hat{\sigma}_{z,t}-\frac{C_{\delta,t,4}}{\sigma_z\sqrt{N_t(z)}},\hat{\sigma}_{z,t}+\frac{C_{\delta,t,4}}{\sigma_z\sqrt{N_t(z)}}\right]\right\},
\end{align}
where $C_{\delta,t,0}\asymp\sqrt{\log\log t+\log(1/\delta)}$ and $C_{\delta,t,4}\asymp\sqrt{\log t+\log(T/\delta)}$. The good event $\cE=\cE_1(\delta)\bigcap\cE_2(\delta)$ occurs with probability at least $1-2\delta.$ The length of initialization phase is determined by
\begin{equation}
    \label{eq:T0_linear}
    T_0=\frac{8C_{\delta,T,4}^2}{\sigma_0^4\wedge\sigma_1^4\wedge(2^{-5}\sigma_0^2\sigma_1^2\underline{\pi})}\vee \frac{16C_{\delta,T,0}^2}{\underline{\pi}^2}\vee\frac{1}{(2\underline{\pi})^{1/\eta}}\asymp \log (T/\delta).
\end{equation}

\begin{theorem}
    \label{thm:regret_linear}
    Suppose that Conditions~\ref{cond:sutva}, \ref{cond:varep} and \ref{cond:X_linear} hold, $\delta=T^{-1}$, $T_0$ is given by~\eqref{eq:T0_linear}, and there exists some constant $\ell_{\max}\ge\ell(a^*)$ such that $\max_{t\in[T]}\ell(a_t)\le\ell_{\max}$. Then, for the linear Model~\eqref{eq:model_linear}, it follows that $\eE\cR_T=\widetilde{O}((\log T)^2).$
\end{theorem}
\begin{proof}
The regret can be decomposed into three parts, i.e., $\cR_T=\cR_{T_0}+\cR_T'+\cR_T''$, with definitions similar to those in Section~\ref{supsubsubsec:high_regret}. \\
(i) According to \eqref{eq:T0_linear}, we have $T_0\asymp\log(T/\delta)$. Then, using similar arguments as in Section~\ref{supsubsubsec:high_regret}, it can be shown that $\eE\cR_{T_0}=\Theta(T_0)=\Theta(\log T)$. \\
(ii) Similar to Lemmas~\ref{lem:pi_rate} and \ref{lem:pi_squa}, we can show that $|\pi_t-\pi^*|=\Theta(C_{\delta,t,4}t^{-1/2})$ under the good event $\cE$, and 
$$
\frac{\sigma_1^2}{\pi_t}+\frac{\sigma_0^2}{1-\pi_t}-\frac{\sigma_1^2}{\pi^*}-\frac{\sigma_0^2}{1-\pi^*}=\Theta(|\pi_t-\pi^*|^2)=\Theta(C_{\delta,t,4}^2t^{-1}).
$$
Thus, under the good event $\cE,$
$$
\cR_T'\le\sum_{t=T_0+1}^{T}\frac{\sigma_1^2}{\pi_t}+\frac{\sigma_0^2}{1-\pi_t}-\frac{\sigma_1^2}{\pi^*}-\frac{\sigma_0^2}{1-\pi^*}=O(C_{\delta,T,4}^2\log T)=\Theta((\log T)^2),
$$
which further implies that $\eE\cR_T'=O((\log T)^2).$ \\
(iii) Note that
$$
\begin{aligned}
    \cR_T'\le&\frac{2}{\underline{\pi}}\sum_{z\in\{0,1\}}\sum_{t=T_0+1}^{T}\eE_{(\bsX,\bsS)}\big[\{\mu_z(\bX_t,\bS_t)-\hat{\mu}_{z,t-1}(\bX_t,\bS_t)\}^2\big|\cF_{t-1}]\\
    \lesssim&\sum_{z\in\{0,1\}}\sum_{t=T_0+1}^{T}\eE_{\bsX}\big\{(\hat{\balpha}_{S,z,t-1}^{\T}\bX-\balpha_{S,z}^{\T}\bX)^2|\cF_{t-1}\big\}\\
    &+\sum_{z\in\{0,1\}}\sum_{t=T_0+1}^{T}(\hat{\gamma}_{z,t-1}-\gamma_z)^2\eE_{\bsX}\big\{(\hat{\balpha}_{S,z,t-1}^{\T}\bX-\balpha_{S,z}^{\T}\bX)^2|\cF_{t-1}\big\}\\
    &+\sum_{z\in\{0,1\}}\sum_{t=T_0+1}^{T}\eE_{\bsX}\big\{(\hat{\balpha}_{Y,z,t-1}^{\T}\bX-\balpha_{Y,z}^{\T}\bX)^2|\cF_{t-1}\big\}
    +\sum_{z\in\{0,1\}}\sum_{t=T_0+1}^{T}(\hat{\gamma}_{z,t-1}-\gamma_z)^2.
\end{aligned}
$$
Under the good event $\cE,(\hat{\gamma}_{z,t-1}-\gamma_z)^2\lesssim C_{\delta,t,2}^2t^{-1}$. Combining Lemma~\ref{lem:MSE_bound} and the deviation inequalities for $|\hat{\balpha}_{S,z,t}^{\T}\bx-\balpha_{S,z}^{\T}\bx|$ and $|\hat{\balpha}_{Y,z,t}^{\T}\bx-\balpha_{Y,z}^{\T}\bx|$, we have $\eE(\hat{\balpha}_{S,z,t-1}^{\T}\bX-\balpha_{S,z}^{\T}\bX)^2|\cE)\lesssim t^{-1}$ and $\eE(\hat{\balpha}_{Y,z,t-1}^{\T}\bX-\balpha_{Y,z}^{\T}\bX)^2|\cE)\lesssim t^{-1}$. Thus, we have $\eE\cR_T''\lesssim C_{\delta,T,2}^2\sum_{t=1}^{T}t^{-1}=\Theta((\log T)^2)$. Combining above, $\eE\cR_T\lesssim \eE\cR_{T_0}+\eE\cR_T'+\eE\cR_T''\lesssim(\log T)^2.$ The proof is completed. 
\end{proof}

To ensure $\cR_T=\widetilde{O}(\log T),$ we need $C_{\delta,t,4}\asymp\sqrt{\log\log t+\log(1/\delta)}.$ We consider using the uniform Bernstein bound proposed in \textcolor{blue}{Balsubramani and Ramdas} (\textcolor{blue}{2016}), which provide the law of the iterated logarithm (LIL)-based concentration inequalities. This uniform Bernstein bound only applies to bounded martingales. Notice that
$$
\begin{aligned}
    \hat{\sigma}_{0,t}^2-\sigma_0^2=&N_{0,t}^{-1}\sum_{j\in G_{0,t}}\{Y_j-\hat{\mu}_{0,t}(\bX_j,\bS_j)\}^2-\sigma_0^2\\
    =&N_{0,t}^{-1}\sum_{j\in G_{0,t}}\{Y_j-\mu_0(\bX_j,\bS_j)+\mu_0(\bX_j,\bS_j)-\hat{\mu}_{0,t}(\bX_j,\bS_j)\}^2-\sigma_0^2\\
    =&N_{0,t}^{-1}\sum_{j\in G_{0,t}}\{\varepsilon_j^2(0)-\sigma_0^2\}+N_{0,t}^{-1}\sum_{j\in G_{0,t}}\{\mu_0(\bX_j,\bS_j)-\hat{\mu}_{0,t}(\bX_j,\bS_j)\}^2\\
    &+N_{0,t}^{-1}\sum_{j\in G_{0,t}}\{Y_j-\mu_0(\bX_j,\bS_j)\}\cdot \{\mu_0(\bX_j,\bS_j)-\hat{\mu}_{0,t}(\bX_j,\bS_j)\}\\
    =:&N_{0,t}^{-1}\sum_{j\in G_{0,t}}M_j-N_{0,t}^{-1}\sum_{j\in G_{0,t}}\{\mu_0(\bX_j,\bS_j)-\hat{\mu}_{0,t}(\bX_j,\bS_j)\}^2\\
    &+2N_{0,t}^{-1}\sum_{j\in G_{0,t}}\{Y_j-\hat{\mu}_{0,t}(\bX_j,\bS_j)\}\cdot \{\mu_0(\bX_j,\bS_j)-\hat{\mu}_{0,t}(\bX_j,\bS_j)\}\\
    \le&N_{0,t}^{-1}\sum_{j\in G_{0,t}}M_j+2N_{0,t}^{-1}\sum_{j\in G_{0,t}}\{Y_j-\hat{\mu}_{0,t}(\bX_j,\bS_j)\}\cdot \{\mu_0(\bX_j,\bS_j)-\hat{\mu}_{0,t}(\bX_j,\bS_j)\}.
\end{aligned}
$$
If $\{M_j,\cF_j\}$ forms a bounded martingale sequence, then with probability at least $1-\delta,$
$$
N_{0,t}^{-1}\sum_{j\in G_{0,t}}M_j\lesssim \sqrt{\log\log t+\log (1/\delta)}.
$$
However, the second term above cannot be bounded under the law of the iterated logarithm. In \textcolor{blue}{Neopane et al.} (\textcolor{blue}{2025a},\textcolor{blue}{b}), the second term is zero. In particular, when both $\mu_0(\cdot,\cdot)$ and $\mu_1(\cdot,\cdot)$ are assumed to be constants functions, we may use the average of outcomes in each group to estimate them. Specifically, letting $\hat{\mu}_0(\cdot,\cdot)\equiv N_{0,t}^{-1}\sum_{j\in G_{0,t}}Y_j$, we have $\sum_{j\in G_{0,t}}\{Y_j-\hat{\mu}_{0,t}(\bX_j,\bS_j)\}\cdot \{\mu_0(\bX_j,\bS_j)-\hat{\mu}_{0,t}(\bX_j,\bS_j)\}=0.$ Combining above, it follows that $C_{\delta,t,4}\asymp\sqrt{\log\log t+\log(1/\delta)},$ and thus $\eE\cR_T'=\widetilde{O}(\log T).$ Similarly, we can also show that $\eE\cR_T''=\widetilde{O}(\log T)$ under this setting, which leads to $\eE\cR_T=\widetilde{O}(\log T).$

\subsection{Batch-wise adaptive experimental design}
\label{supsubsec:batch_adaptive}

In practice, when the covariate dimension $d$ is high or the total time horizon $T$ is large, implementing the SLOACI algorithm under a fully adaptive design can be computationally intensive. Furthermore, such designs are often infeasible in real-world applications due to delayed outcomes, ethical restrictions, and other resource constraints. Therefore, we consider a batch-wise adaptive experimental design and extend Algorithm~\ref{alg:main} to a batched version. The main idea is that, during the concentration phase of the algorithm, the conditional mean functions and variances are not estimated, nor is the allocation updated, at every stage. Instead, these estimations and updates are performed only at specified stages, or more specifically, after each batch of size $b$ is observed. The batch-wise adaptive experimental design, summarized in Algorithm~\ref{alg:batch}, is employed to accelerate the simulations in Sections \ref{sec:case} and \ref{supsec:sim}. It is straightforward to show that this batch-wise procedure can still yield an ATE estimator that attains the semiparametric efficiency bound, as presented in Theorem~\ref{thm:asy_nor}, and admits a regret bound similar to that established in Theorem~\ref{thm:regret}.

\begin{algorithm}[H]
\spacingset{1.2}
\caption{Batch-wise SLOACI for nonparametric nonlinear model}
\begin{algorithmic}[1]
\label{alg:batch}
\STATE \textbf{Input}: Batch size $b$, clipping rate $\eta$, the length of initialization phase $T_0$, the total time horizon $T$, and the kernel function $K_h(\cdot)$.

\STATE \textbf{Initialization}: For $t=1,\dots,T_0,$ assign the treatment $Z_t$ to 0 for odd $t$ and to 1 for even $t$, observe the covariates $\bX_t$, the surrogates $\bS_t=(S_t(0),S_t(1))^{\T}$ and the outcome $Y_t=Y_t(1)Z_t+Y_t(0)(1-Z_t)$, and obtain the conditional mean function and variance estimators $\hat{\mu}_{0,0}(\bx,\bs),$ $\hat{\mu}_{1,0}(\bx,\bs),$ $\hat{\sigma}_{0,0}^2$, and $\hat{\sigma}_{1,0}^2$ based on observations $\{(\bX_j,Z_j,\bS_j,Y_j)\}_{j=1}^{T_0}$ according to \eqref{eq:m_S_hat}--\eqref{eq:sigma_hat}.

\FOR{$k=1,\dots,\lceil(T-T_0)/b\rceil,$}
    \STATE Define the $k$-th batch as $B_k=\{T_0+(k-1)b+1,\dots,(T_0+kb)\wedge T\}$.
    \STATE Calculate the initial allocation as $\widetilde{\pi}_k=\hat{\sigma}_{1,k-1}/(\hat{\sigma}_{1,k-1}+\hat{\sigma}_{0,k-1}).$
    \STATE Set the clipping threshold as $\zeta_k=(1/2)\cdot \{T_0+(k-1)b+1\}^{-\eta}$.
    \STATE Obtain the adaptive treatment allocation by clipping as $\pi_k={\rm CLIP}(\widetilde{\pi}_k,\zeta_k,1-\zeta_k)$.
    \FOR{$t\in[B_k],$}
        \STATE Assign the treatment $Z_t$ as 1 with probability $\pi_k$ and 0 with $1-\pi_k$.
        \STATE Observe the covariates $\bX_t$, the surrogates $\bS_t=(S_t(0),S_t(1))^{\T}$, and the outcome $Y_t=Y_t(1)Z_t+Y_t(0)(1-Z_t)$.
    \ENDFOR
    \STATE Calculate the estimated $\hat{m}_{S,0,k}(\cdot),\hat{m}_{Y,0,k}(\cdot),\hat{\gamma}_{0,k},\hat{m}_{S,1,k}(\cdot),\hat{m}_{Y,1,k}(\cdot)$, and $\hat{\gamma}_{1,k}$ based on observations $\{(\bX_j,Z_j,\bS_j,Y_j)\}_{j=1}^{(T_0+kb)\wedge T}$ according to \eqref{eq:m_S_hat}--\eqref{eq:gamma_hat}.
    \STATE Obtain the conditional mean function and variance estimators $\hat{\mu}_{0,k}(\bx,\bs),$ $\hat{\mu}_{1,k}(\bx,\bs),$ $\hat{\sigma}_{0,k}^2$, and $\hat{\sigma}_{1,k}^2$ according to~\eqref{eq:mu_hat} and \eqref{eq:sigma_hat}.
\ENDFOR
\STATE \textbf{Estimate}: 
{\footnotesize $$
\begin{aligned}
    \hat{\tau}_T=&\frac{1}{T_0}\sum_{t=1}^{T_0}\frac{I(Z_t=1)Y_t}{1/2}-\frac{I(Z_t=0)Y_t}{1/2}
    +\frac{1}{T-T_0}\sum_{k=1}^{\lceil(T-T_0)/b\rceil}\sum_{t\in[B_k]}\bigg[\hat{\mu}_{1,k-1}(\bX_t,\bS_t)-\hat{\mu}_{0,k-1}(\bX_t,\bS_t)\\
    &+\frac{I(Z_t=1)\{Y_t-\hat{\mu}_{1,k-1}(\bX_t,\bS_t)\}}{\pi_k}-\frac{I(Z_t=0)\{Y_t-\hat{\mu}_{0,k-1}(\bX_t,\bS_t)\}}{1-\pi_k}\bigg].
\end{aligned}
$$}
\STATE \textbf{Output}: The adaptive AIPW estimator $\hat{\tau}_T$ under the proposed batch-wise SLOACI design.
\end{algorithmic}
\end{algorithm}

\subsection{RAR design in empirical studies}
\label{supsubsec:rar}

The following Algorithm~\ref{alg:rar} outlines the RAR design used for comparison with our proposed SLOACI in Sections~\ref{sec:case} and \ref{supsec:sim}. Notably, unlike the CARA design studied in \textcolor{blue}{Cook et al.} (\textcolor{blue}{2024}) and \textcolor{blue}{Kato et al.} (\textcolor{blue}{2025}), which is a special case of RAR, we assume that the variances of the error terms $\varepsilon_{Y,t}(0)$ and $\varepsilon_{Y,t}(1)$ in Model \eqref{eq:model} are independent of the covariate $\bX_t$, so the allocation does not need to be adjusted based on $\bX_t$. For the RARS design as defined in Section \ref{supsec:sim}, the new covariates are taken as $\big(\bX_t, S_t(0)\big)$ for the control group and $\big(\bX_t, S_t(1)\big)$ for the treatment group, while the remaining steps are similar to those of the RAR design. In addition, the batch-wise versions of RAR and RARS can also be extended in a similar manner to that in Section~\ref{supsubsec:batch_adaptive}, which we omit here.

\begin{algorithm}[H]
\spacingset{1.2}
\caption{RAR for nonparametric nonlinear model}
\begin{algorithmic}[1]
\label{alg:rar}
\STATE \textbf{Input}: Clipping rate $\eta$, the length of initialization phase $T_0$, the total time horizon $T$, and the kernel function $K_h(\cdot)$.

\STATE \textbf{Set}: $\pi_0 \leftarrow 1/2$, $\hat{m}_{Y,0,0}(\bx)\leftarrow 0$, $\hat{m}_{Y,1,0}(\bx)\leftarrow 0$, $\hat{\sigma}_{Y,0,0}^2\leftarrow 0$, and $\hat{\sigma}_{Y,1,0}^2\leftarrow 0$.

\STATE \textbf{Initialization}: For $t=1,\dots,T_0,$ set $\pi_t \leftarrow 1/2$, $\hat{m}_{Y,0,t}(\bx)\leftarrow 0$, and $\hat{m}_{Y,1,t}(\bx)\leftarrow 0$, assign the treatment $Z_t$ to 0 for odd $t$ and to 1 for even $t$, observe the covariates $\bX_t$, and the outcome $Y_t=Y_t(1)Z_t+Y_t(0)(1-Z_t)$, and obtain the conditional mean function and variance estimators $\hat{m}_{Y,0,T_0}(\bx),$ $\hat{m}_{Y,1,T_0}(\bx),$ $\hat{\sigma}_{Y,0,T_0}^2$, and $\hat{\sigma}_{Y,1,T_0}^2$.

\FOR{$t=T_0+1,\dots,T,$}
    \STATE Calculate the initial allocation as $\widetilde{\pi}_t=\hat{\sigma}_{Y,1,t-1}/(\hat{\sigma}_{Y,1,t-1}+\hat{\sigma}_{Y,0,t-1}).$
    \STATE Set the clipping threshold as $\zeta_t=(1/2)\cdot t^{-\eta}$.
    \STATE Obtain the adaptive treatment allocation by clipping as $\pi_t={\rm CLIP}(\widetilde{\pi}_t,\zeta_t,1-\zeta_t)$.
    \STATE Assign the treatment $Z_t$ as 1 with probability $\pi_t$ and 0 with probability $1 - \pi_t$.
    \STATE Observe the covariates $\bX_t$, and the outcome $Y_t=Y_t(1)Z_t+Y_t(0)(1-Z_t)$.
    \STATE Calculate the estimated $\hat{m}_{Y,0,t}(\cdot)$ and $\hat{m}_{Y,1,t}(\cdot)$ according to \eqref{eq:m_Y_hat}.
    \STATE Obtain the conditional variance estimators by
    $$
    \hat{\sigma}_{Y,z,t}^2=\frac{1}{\sum_{j=1}^{t}I(Z_j=z)}\sum_{j=1}^{t}\{Y_j-\hat{m}_{Y,z,t}(\bX_j)\}^2I(Z_j=z),\quad z\in\{0,1\}.
    $$
\ENDFOR
\STATE \textbf{Estimate}: 
$$
\begin{aligned}
    \hat{\tau}_T=\frac{1}{T}\sum_{t=1}^{T}&\bigg[\hat{m}_{Y,1,t-1}(\bX_t)-\hat{m}_{Y,0,t-1}(\bX_t)\\
    &+\frac{I(Z_t=1)\{Y_t-\hat{m}_{Y,1,t-1}(\bX_t)\}}{\pi_t}-\frac{I(Z_t=0)\{Y_t-\hat{m}_{Y,0,t-1}(\bX_t)\}}{1-\pi_t}\bigg].
\end{aligned}
$$
\STATE \textbf{Output}: The adaptive AIPW estimator $\hat{\tau}_T$ under the RAR design.
\end{algorithmic}
\end{algorithm}

\subsection{Heteroskedasticity setting}
\label{supsubsec:heter}

This section considers extending the proposed method in Section~\ref{subsec:procedure} to the heteroskedastic setting, under which we allow the conditional variances of $Y_t(0)$ and $Y_t(1)$ to depend on $\big(\bX_t, S_t(0)\big)$ and $\big(\bX_t, S_t(1)\big)$, respectively. Define $\sigma_z^2(\bx,\bs):=\var\{\varepsilon_t(z)|\bX_t=\bx,\bS_t=\bs\},$ and $e_z(\bx,\bs):=\eE\{\varepsilon_t^2(z)|\bX_t=\bx,\bS_t=\bs\}$ for each $z\in\{0,1\}.$ The conditional variance estimator in \eqref{eq:sigma_hat} can be replaced by the following two-step procedure. \textit{First}, a Nadaraya-Watson estimator of $e_z(\bx,\bs)$ at Stage $t$ is
$$
\hat{e}_{z,t}(\bx,\bs)=\frac{\sum_{j=1}^{t}\widetilde{K}_h(\widetilde{\bw}_z-\widetilde{\bW}_{z,j})Y_j^2I(Z_j=z)}{\sum_{j=1}^{t}\widetilde{K}_h(\widetilde{\bw}_z-\widetilde{\bW}_{z,j})I(Z_j=z)},
$$
where $\widetilde{\bW}_{z,j}=\big(\bX_j^{\T},S_j(z)\big)^{\T},\widetilde{\bw}_z=\big(\bx^{\T},s(z)\big)^{\T}\in\eR^{d+1},$ and $\widetilde{K}_h:\eR^{d+1}\to1$ is a multivariate kernel function with bandwidth $h>0$. In particular, define $\hat{e}_{z,t}(\bx,\bs)=0$ if $\sum_{j=1}^{t}\widetilde{K}_h(\widetilde{\bw}_z-\widetilde{\bW}_{z,j})I(Z_j=z)=0.$
\textit{Second}, a natural estimator of $\sigma_z^2(\bx,\bs)$ is 
$$
\hat{\sigma}_{z,t}^2(\bx,\bs)=
\begin{cases}
    \{\hat{e}_{z,t}(\bx,\bs)-\hat{\mu}_{z,t}^2(\bx,\bs)\}, & {\rm if}\ \hat{e}_{z,t}(\bx,\bs)>\hat{\mu}_{z,t}^2(\bx,\bs), \\
    \epsilon, & {\rm if}\ \hat{e}_{z,t}(\bx,\bs)\le\hat{\mu}_{z,t}^2(\bx,\bs),
\end{cases}
$$
where $\epsilon>0$ is a small constant to ensure positivity. The allocation at each Stage $t\in[T]$, defined as $\pi_t=\eP(Z_t=1|\bX_t,\bS_t,\cF_{t-1})$, is then updated (before CLIP) as:
$$
\pi_t=\frac{\hat{\sigma}_{1,t-1}(\bX_t,\bS_t)}{\hat{\sigma}_{1,t-1}(\bX_t,\bS_t)+\hat{\sigma}_{0,t-1}(\bX_t,\bS_t)}.
$$

This extension generalizes the surrogate-leveraged RAR design developed in Section~\ref{subsec:procedure} to a covariate-adjusted surrogate-leveraged RAR design. Analogous asymptotic and non-asymptotic theoretical results can be established accordingly.

\section{Simulation studies}
\label{supsec:sim}

In this section, we conduct a series of simulations to evaluate the empirical performance of the proposed adaptive designs and sequential testing procedures. 

\subsection{Setup}

Suppose the covariate $X_t$ is independently sampled from ${\rm Uniform}(0,2).$ The mean functions of the primary outcomes and the surrogates are specified as $m_{Y,1}(x)=x^2+3x+1,m_{Y,0}(x)=x^2+x+2,m_{S,1}(x)=\sin(2x)$, and $m_{S,0}(x)=\sin(x).$ The model errors in \eqref{eq:model} are assumed to follow bivariate Gaussian distributions: $\big(\varepsilon_{Y,t}(0),\varepsilon_{S,t}(0)\big)^{\T}\sim\cN(\bzero,\bSigma_0)$ and $\big(\varepsilon_{Y,t}(1),\varepsilon_{S,t}(1)\big)^{\T}\sim\cN(\bzero,\bSigma_1)$, where $\bSigma_0=\sigma_{Y,0}^2
\left(
\begin{smallmatrix}
1 & \rho_0 \\
\rho_0 & 1
\end{smallmatrix}
\right)$ 
and $\bSigma_1=\sigma_{Y,1}^2
\left(
\begin{smallmatrix}
1 & \rho_1 \\
\rho_1 & 1
\end{smallmatrix}
\right)$, with $\sigma_{Y,0}+\sigma_{Y,1}=1$. Under this setup, the ATE is $\tau_0=1$, and the semiparametric efficiency bound without surrogates is $T\cdot\widetilde{V}^*=2.333$. We consider the following three scenarios:
\begin{itemize}
    \item[(i)] Scenario 1: $\rho_0=0.8,\rho_1=0.8,\sigma_{Y,0}=0.4,\sigma_{Y,1}=0.6$;

    \item[(ii)] Scenario 2: $\rho_0=0.5,\rho_1=-0.5,\sigma_{Y,0}=0.4,\sigma_{Y,1}=0.6$;

    \item[(iii)] Scenario 3: $\rho_0=0.8,\rho_1=0.8,\sigma_{Y,0}=0.2,\sigma_{Y,1}=0.8$.
\end{itemize}
In Scenarios 1, 2 and 3, the surrogate-leveraged semiparametric efficiency bounds are $T\cdot V_T^*=2.026,2.213$ and $2.128$, respectively, and the optimal allocations are $\pi^*=0.6,0.6$ and 0.8, respectively. Scenario 1 represents a standard setting where the conditional correlations between the primary outcomes and the surrogates are relatively high. Scenario 2 reduces the conditional correlations, while Scenario 3 increases the difference in conditional variances between the control and treatment groups. By comparing the results of Scenario 1 with those of Scenarios 2 and 3, we can assess the impact of $\rho_z$ and $\underline{\pi}$ on the performance of our proposed method, respectively.

We compare the proposed SLOACI in Algorithm~\ref{alg:main} with three competing designs: the RAR design, the RAR design incorporating surrogates (denoted as RARS), and the optimal design (denoted as OPT). 
RAR relies only on covariate $X_t$, excluding the surrogates $S_t(0)$ and $S_t(1)$, to estimate conditional mean functions and variances of the primary outcomes, which are then used to determine the allocation and estimate the ATE. Other components, including the use of the Nadaraya-Watson regression to estimate $m_{Y,0}(\cdot)$ and $m_{Y,1}(\cdot)$, as well as the application of CLIP, remain the same as in SLOACI.
The detailed procedure of RAR is given in Section~\ref{supsubsec:rar}.
RARS follows similar procedures to RAR, but differs in that the conditional mean function of $Y_t(z)$ is estimated using both $S_t(z)$ and $X_t$ as covariates for $z\in\{0,1\}$. In other words, it regards the semiparametric partially linear models in \eqref{eq:model_cond_XS} as nonparametric models. OPT corresponds to the oracle estimator given in \eqref{eq:opt_est}. For a fair comparison, all the above designs estimate the ATE using the (adaptive) AIPW estimator. It is noteworthy that numerous prior studies (e.g., 
\textcolor{blue}{Dai et al.}, \textcolor{blue}{2023}; \textcolor{blue}{Kato et al.}, \textcolor{blue}{2025}; \textcolor{blue}{Neopane et al.}, \textcolor{blue}{2025a},\textcolor{blue}{b}) have shown that RAR outperforms RCT and Explore-then-Commit (\textcolor{blue}{Hahn et al.}, \textcolor{blue}{2011}), and that the AIPW estimator is more efficient than the IPW and Difference-in-Means estimators. Therefore, we omit comparisons with these baseline methods in this paper.
For sequential testing, we compare the sequence of CLT-based confidence intervals (denoted as CLT), the sequence of Bonferroni-corrected confidence intervals (denoted as BF), the asymptotic confidence sequence (denoted as ASY), and the empirical Bernstein confidence sequence (denoted as EB). 

Unless otherwise specified, when implementing the adaptive designs, we set $\eta=0.25,T_0=20, \kappa=2\beta+d,\beta=1$, and $c_h=2\hat{\sigma}(X_t)$, where $\hat{\sigma}(X_t)$ denotes the sample standard deviation of $\bX_t$, and use the Epanechnikov kernel $K(x)=(3/4)(1-x^2)I(|x|<1)$ in \eqref{eq:m_S_hat} and \eqref{eq:m_Y_hat}. For sequential testing, the nominal type-I error rate is fixed at $\alpha=0.05$; the initial peeking time is selected as $t_0=50$; for the Bonferroni correction in \eqref{eq:bf_ci}, the total time horizon is set to $2500$; for ASY in Theorem~\ref{thm:asy_cs}, the parameter $\varrho$ is chosen as in Section~\ref{subsec:asy_cs}; and for EB in Theorem~\ref{thm:finite_cs}, we set $c=0.5,\nu_0^2=1$, and $\hat{\xi}_0=0$. To accelerate computation, we update the estimators $\hat{\mu}_{z,t}(\cdot,\cdot)$ and $\hat{\sigma}_{z,t}$ for $z\in\{0,1\}$ using a batch-wise adaptive setting with batch size $b=50$; see Section~\ref{supsubsec:batch_adaptive} for detailed discussion. We repeat each simulation 2000 times.

\subsection{Results}

We present two sets of simulation results: the efficiency comparison in Section \ref{supsubsec:sim_efficiency}, and the sequential testing procedure comparison in Section \ref{supsubsec:sim_testing}.

\subsubsection{Efficiency comparison}\label{supsubsec:sim_efficiency}

Figure~\ref{fig:regret1} shows the curves of normalized variance $T\cdot \var(\hat{\tau}_T)$ for each design against the cumulative sample size $T$. The black dashed line represents the surrogate-leveraged semiparametric efficiency bound under each scenario. For each fixed $T$, the vertical distance between a curve and the black dashed line corresponds to the estimate of the expected regret of a given design, as defined in \eqref{eq:ney_regret}, scaled by $1/T$. From Figure \ref{fig:regret1}, we observe several apparent trends.

\textit{First}, SLOACI converges to the respective surrogate-leveraged semiparametric efficiency bound in all three scenarios, illustrating the semiparametric efficiency of our proposed estimator and validating Theorem~\ref{thm:asy_nor}. 
Meanwhile, although the semiparametric efficiency bound in Scenario 3 is larger than that in Scenario 1, the empirical regret is larger in Scenario 3. 
This is expected because a smaller $\underline{\pi}$ leads to a larger regret bound, as demonstrated in Theorem~\ref{thm:regret}. \textit{Second}, OPT attains the semiparametric efficiency bound even for small $T$, as it has has access to the true models. \textit{Third}, RAR converges only to the suboptimal semiparametric efficiency bound without surrogates, showing the advantage of leveraging surrogate information, as shown in \eqref{eq:ney_var_red}. Moreover, the efficiency gain in Scenario 1 is larger than in Scenarios 2 and 3, because $\rho_0,\rho_1$ and $\sigma_{Y,0}\sigma_{Y,1}$ are larger in Scenario 1 compared to Scenarios 2 and 3, respectively. 

\textit{Lastly}, although RARS leverages surrogate information and thus shares the same theoretical semiparametric efficiency bound as SLOACI, its normalized variance converges significantly more slowly than that of SLOACI. The convergence of RARS is also slower than that of RAR, while attaining a larger optimal normalized variance, which can be evidenced by the crossover of the RAR and RARS curves in Scenarios 1 and 3. This is because the covariate dimension is $d=1$ for SLOACI and RAR, whereas for RARS it is $d=2$. According to Theorem~\ref{thm:regret}, the order of the regret bound increases with $d$, indicating that regarding the surrogates as the covariates can reduce the finite-sample performance of designs. 

\begin{figure}[H]
\centering
\includegraphics[width=0.8\linewidth]{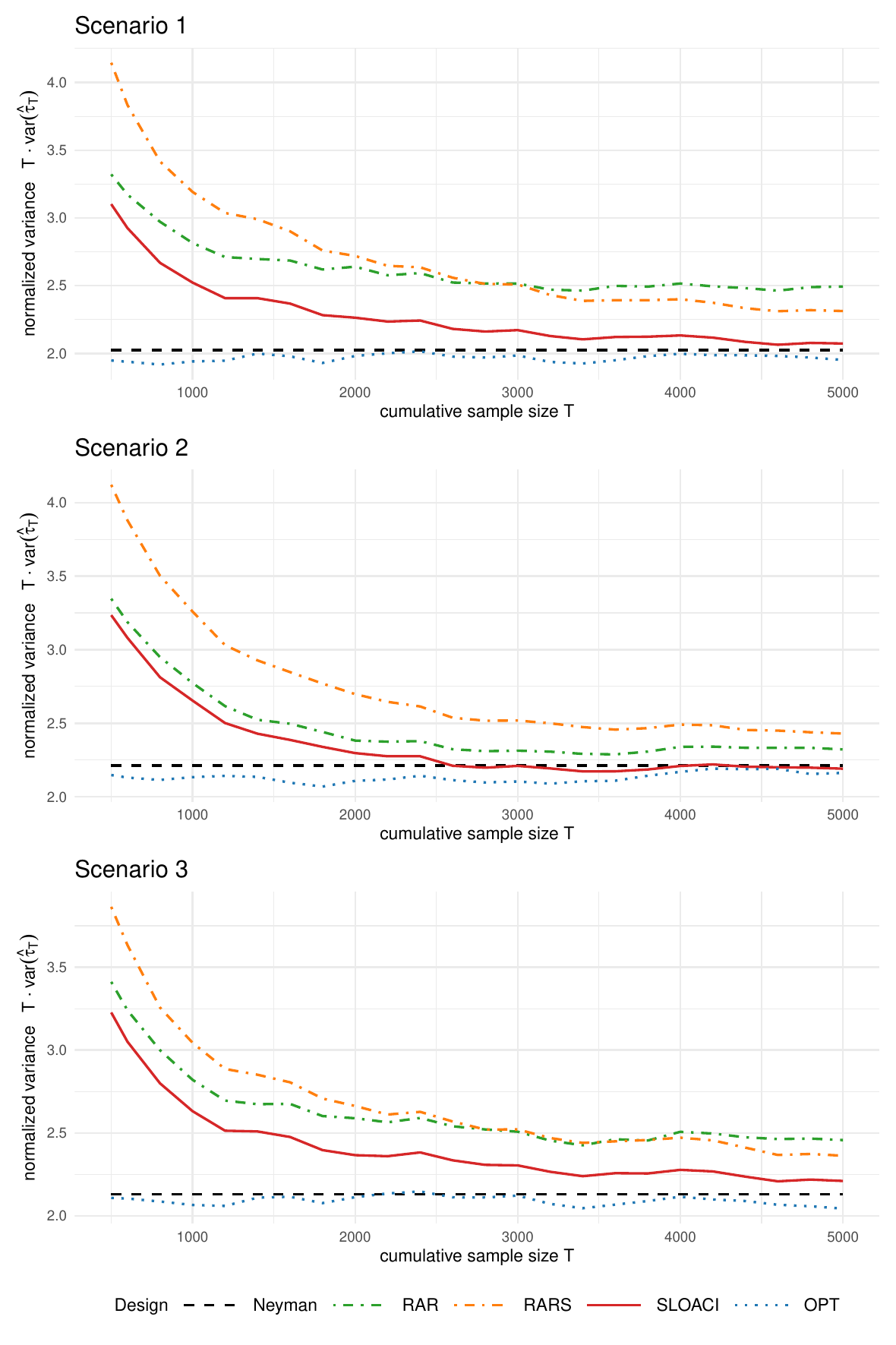}
\vspace{-1em}
\caption{\small Comparison of normalized variances of the (adaptive) AIPW estimator across different designs in Scenarios 1, 2 and 3.}
\label{fig:regret1}
\end{figure}

We further investigate an extreme Scenario 4 where the primary outcomes and the surrogates are conditionally uncorrelated, where the surrogate-leveraged semiparametric efficiency bound is $T\cdot V_T^*=2.333$, and the optimal allocation is $\pi^*=0.6.$ 
\begin{itemize}
    \item[(iv)] Scenario 4: $\rho_0=0,\rho_1=0,\sigma_{Y,0}=0.4,\sigma_{Y,1}=0.6$.
\end{itemize}
The normalized variance curves are displayed in Figure~\ref{fig:regret4}. It can be seen that SLOACI and RAR have almost overlapping curves, while the convergence of RARS is slower than both. This demonstrates the advantage of our proposed method: when it is uncertain whether the surrogates are conditionally correlated with the primary outcomes, using SLOACI incurs no adverse effect, whereas regarding them as covariates can slow the convergence rate due to the curse of dimensionality in nonparametric regression. The above observations offer solid validation of our theoretical results established in Section~\ref{sec:theory}. 

\begin{figure}[H]
\begin{center}
\includegraphics[width=1\linewidth]{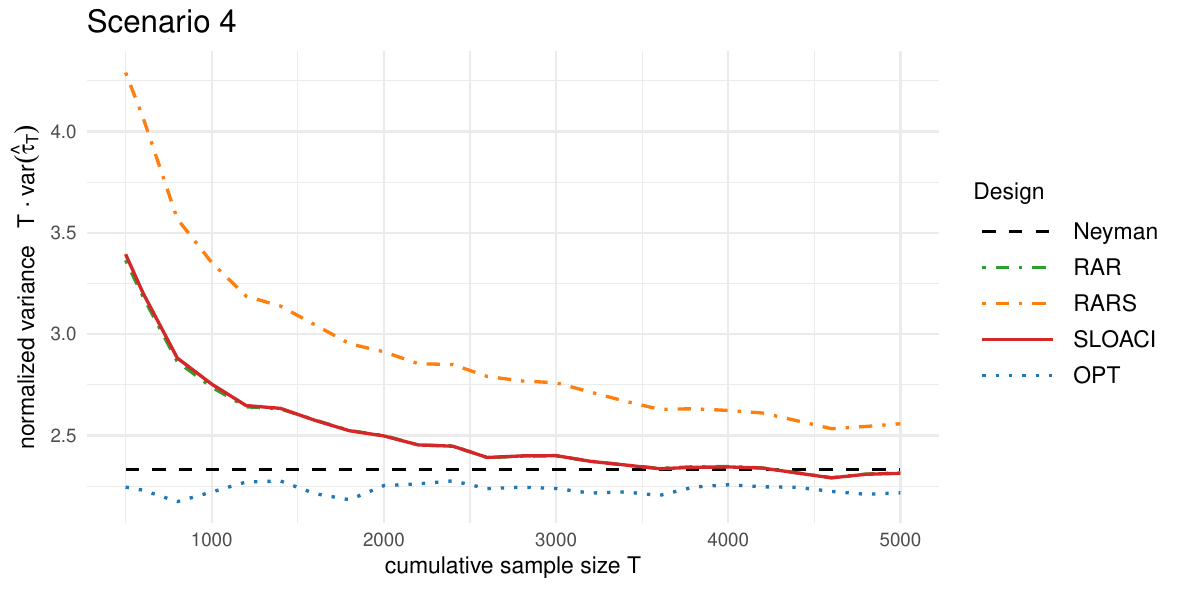}
\end{center}
\vspace{-1em}
\caption{\small Comparison of normalized variances of the (adaptive) AIPW estimator across different designs in Scenario 4.}
\label{fig:regret4}
\end{figure}

We also conduct additional simulations to investigate the effects of using different nonparametric regressions and other tuning parameters on the proposed algorithm. As discussed in Section~\ref{subsec:procedure}, the Nadaraya-Watson regressions in \eqref{eq:m_S_hat} and \eqref{eq:m_Y_hat} can be replaced by any nonparametric regression. Figure~\ref{fig:diff_est} uses Scenario 1 as an example to compare the performance of the Epanechnikov kernel and the Gaussian kernel regressions with different bandwidth parameters $c_h$, and the natural spline regression with different degrees of freedom. It can be observed that the convergence trends of the normalized variance of the ATE estimators obtained from these regressions are overall consistent, with the spline regression yielding a regret bound that is smaller in its constant terms. This suggests that the proposed method has broad applicability and can be combined with more advanced regression techniques. 

\begin{figure}[H]
\begin{center}
\includegraphics[width=1\linewidth]{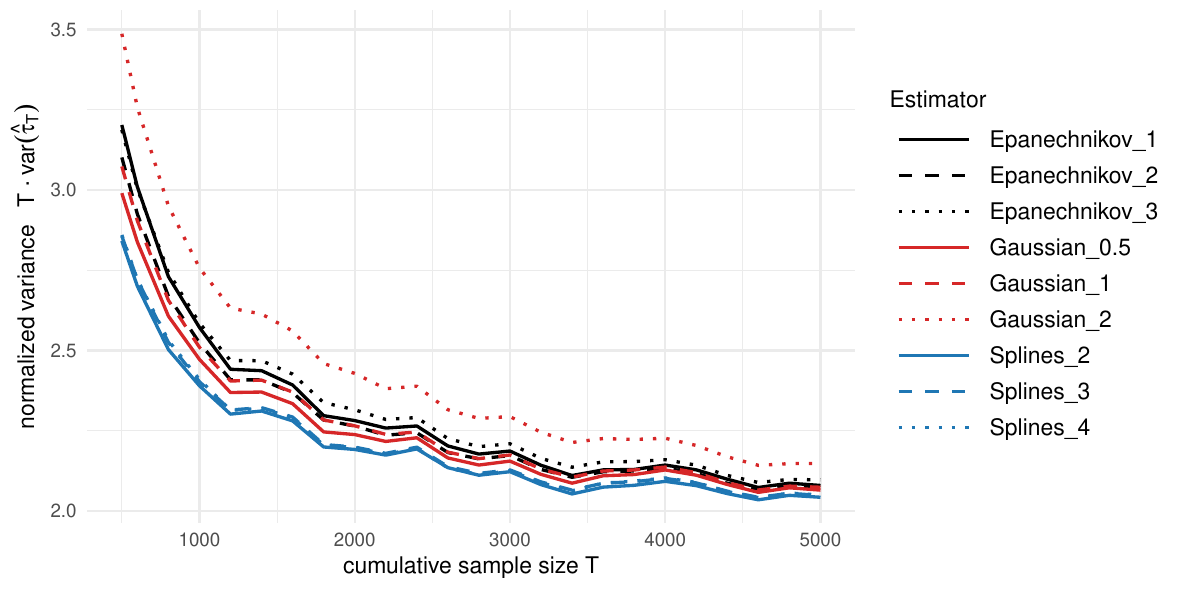}
\end{center}
\vspace{-1em}
\caption{\small Comparison of normalized variances of the proposed estimator across different nonparametric regression estimators and tuning parameters in Scenario 1. The number at the end of each legend label denotes the tuning parameter for each estimator.}
\label{fig:diff_est}
\end{figure}

To examine the effects of tuning parameters on the proposed design in Algorithm~\ref{alg:batch}, we compare different batch sizes, lengths of initialization phase, and clipping rates in Scenario 1 in Figures~\ref{fig:batch}, \ref{fig:T0}, and \ref{fig:clipping_rate}, respectively. The results indicate that the batch size $b$ and clipping rate $\eta$ have minimal impact on the convergence of the normalized variance, whereas the length of initialization phase $T_0$ has a more pronounced effect. Specifically, if $T_0$ is either too small or too large, the regret bound increases, leading to larger normalized variance, as shown in Theorem~\ref{thm:regret}. Choosing $T_0$ appropriately is therefore crucial in finite samples for reducing estimation variance and enhancing inference power.

\begin{figure}[H]
\begin{center}
\includegraphics[width=1\linewidth]{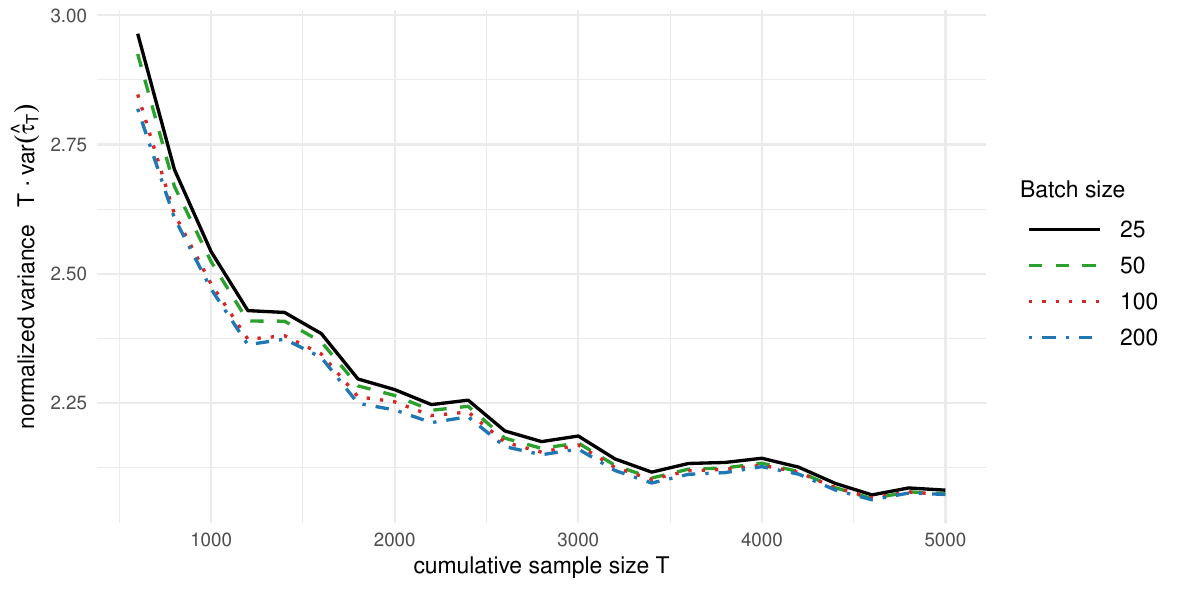}
\end{center}
\vspace{-1em}
\caption{\small Comparison of normalized variances of the proposed estimator across different batch sizes in Scenario 1.}
\label{fig:batch}
\end{figure}

\begin{figure}[H]
\begin{center}
\includegraphics[width=1\linewidth]{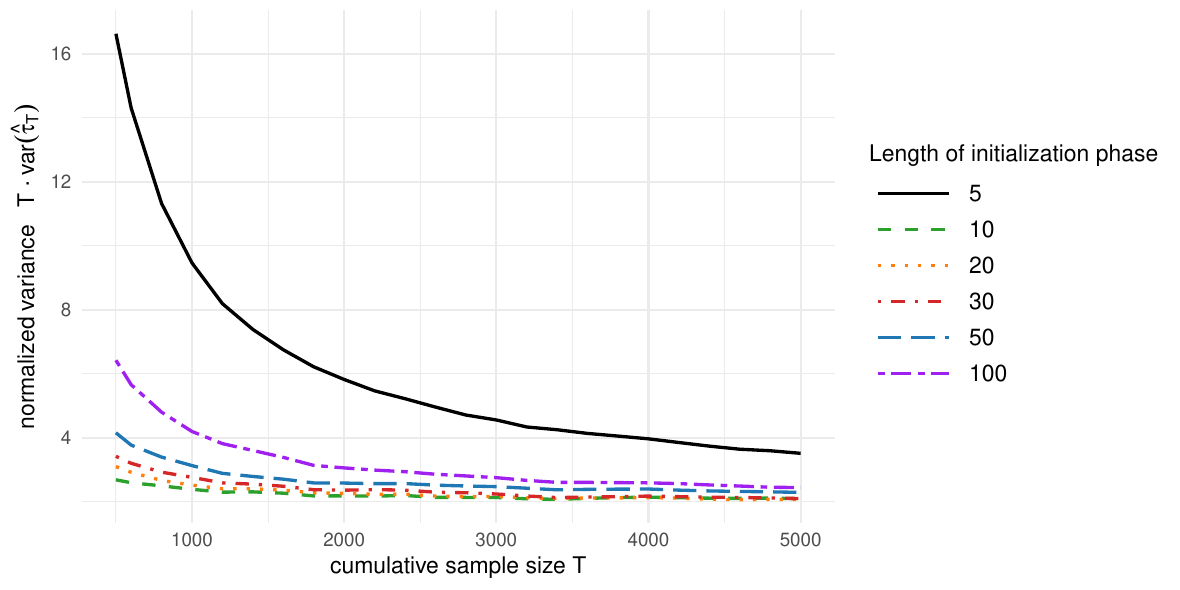}
\end{center}
\vspace{-1em}
\caption{\small Comparison of normalized variances of the proposed estimator across different lengths of initialization phase in Scenario 1.}
\label{fig:T0}
\end{figure}

\begin{figure}[H]
\begin{center}
\includegraphics[width=1\linewidth]{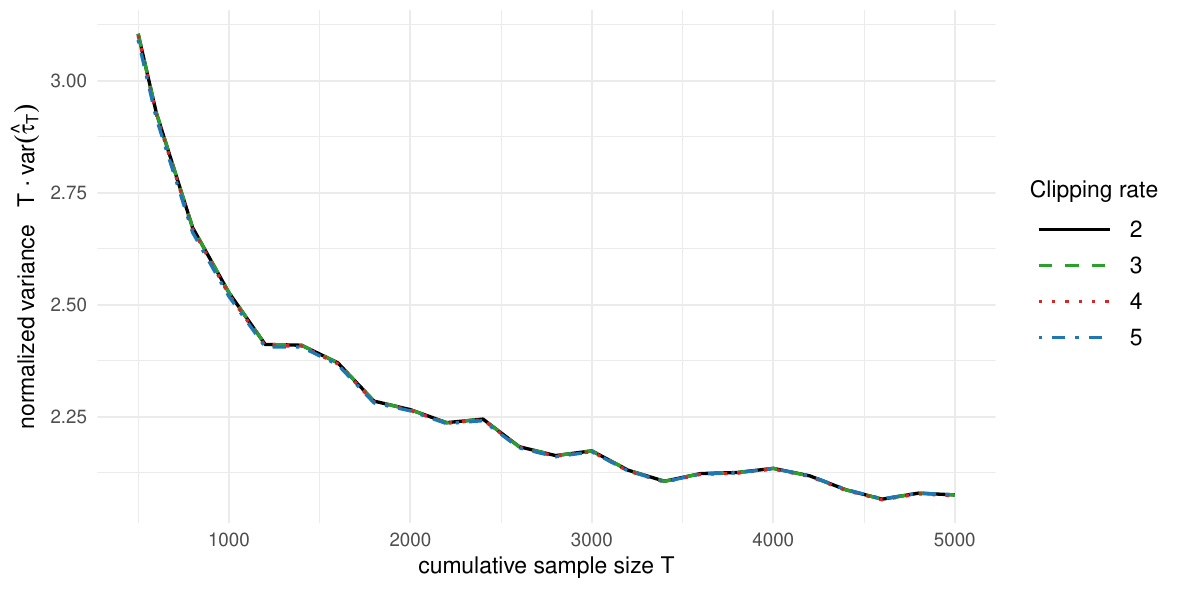}
\end{center}
\vspace{-1em}
\caption{\small Comparison of normalized variances of the proposed estimator across different clipping rates in Scenario 1.}
\label{fig:clipping_rate}
\end{figure}

Figure~\ref{fig:profile} compares the normalized variances of the proposed estimator obtained from Robinson’s (residual-based) method and the profile least squares method proposed in Section~\ref{supsubsec:profile}, across different nonparametric regression estimators in Scenario 1. The results indicate that the difference between the two methods is negligible.

\begin{figure}[H]
\begin{center}
\includegraphics[width=1\linewidth]{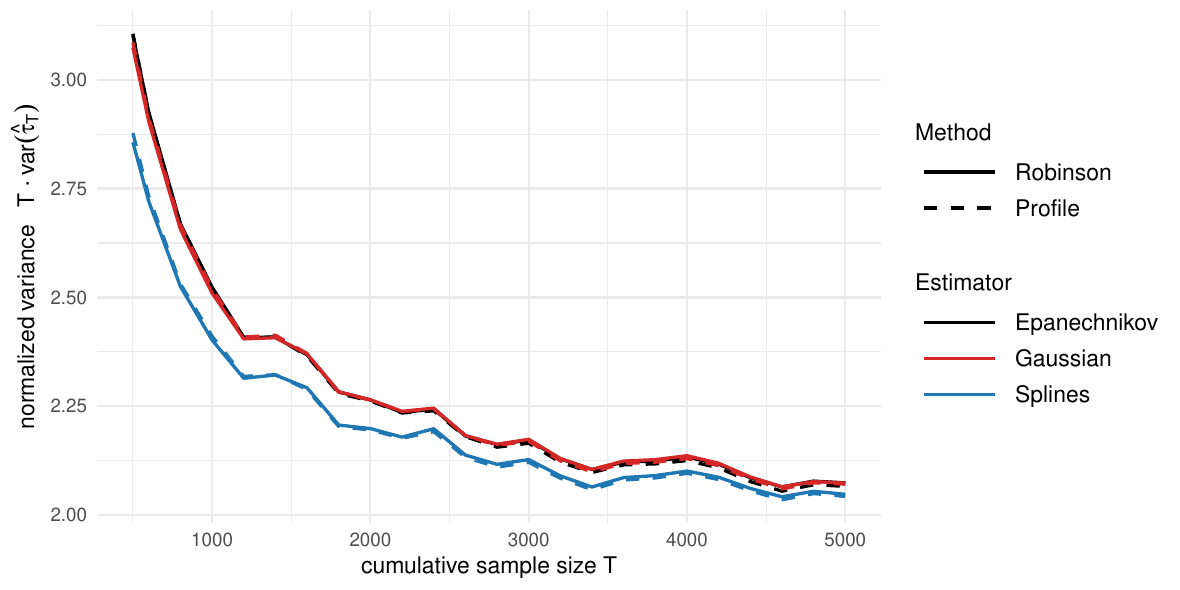}
\end{center}
\vspace{-1em}
\caption{\small Comparison of normalized variances of the proposed estimator between Robinson's method and profile least squares method in Scenario 1.}
\label{fig:profile}
\end{figure}

\subsubsection{Comparison of sequential testing procedures}\label{supsubsec:sim_testing}

We consider three null hypotheses: $\tau_{H_0}=\tau_0,\tau_{H_0}=0.8\tau_0,$ and $\tau_{H_0}=0.6\tau_0$, where the first is used to assess the cumulative coverage rate, and the latter two are used to evaluate the power. For a given confidence sequence (or a sequence of confidence intervals) $(L_t,U_t)_{t\ge1}$, following \textcolor{blue}{Waudby-Smith et al.} (\textcolor{blue}{2024a}); \textcolor{blue}{Cook et al.} (\textcolor{blue}{2024}), the cumulative miscoverage (error) rate up to Stage $T$ is defined as ${\rm Error}(T)=\eP\big(\exists t\in\{t_0,\dots,T\},\tau_0\notin[L_t,U_t]\big)$, and the power is defined as ${\rm Power}(T)=\eP\big(\exists t\in\{t_0,\dots,T\},\tau_{H_0}\notin[L_t,U_t]\big)$, which are estimated from 2000 simulation replications. In addition to these two metrics, we also compare the stopping times of different testing procedures, as defined in Section~\ref{subsec:tradeoff}. In particular, if the null hypothesis cannot eventually be rejected, the stopping time is set to the total time horizon of 2500. Notably, OPT uses the optimal allocation and true conditional mean functions to construct the confidence sequence, so its empirical power and average stopping time reflect the intrinsic difficulty of the sequential testing problem under the given scenario.

Figures~\ref{fig:test1}--\ref{fig:test3} present the empirical cumulative miscoverage rates and empirical powers for different designs and tests in Scenarios 1--3, with the corresponding means and standard errors of the stopping times reported in Table~\ref{tab:stopping_times}. Figure~\ref{fig:test4} presents the results in Scenario 1 when the initial peeking time $t_0$ is set to 30, in comparison with Figure~\ref{fig:test1} with $t_0=50$. Several apparent patterns are observable. 

\textit{First}, when $\tau_{H_0}=\tau_0$, CLT completely fails to control the type-I error rate, ASY attains the nominal level of 0.05, whereas BF and EB yield empirical cumulative miscoverage rates below 0.05 but remain relatively conservative. Benefiting from the rescaling technique, even though Theorem~\ref{thm:finite_cs} does not provide theoretical guarantees for unbounded primary outcomes (as in our simulation setup), EB still performs well in practice. Note that the inference begins at $t_0=50$, which avoids severe miscoverage at early stages. Additional simulations with $t_0=30$, presented in Figure~\ref{fig:test4}, suggest that when the initial peeking time is small, ASY may inflate the miscoverage error, whereas the more conservative BF and EB still maintain control.

\begin{figure}[H]
\begin{center}
\includegraphics[width=0.8\linewidth]{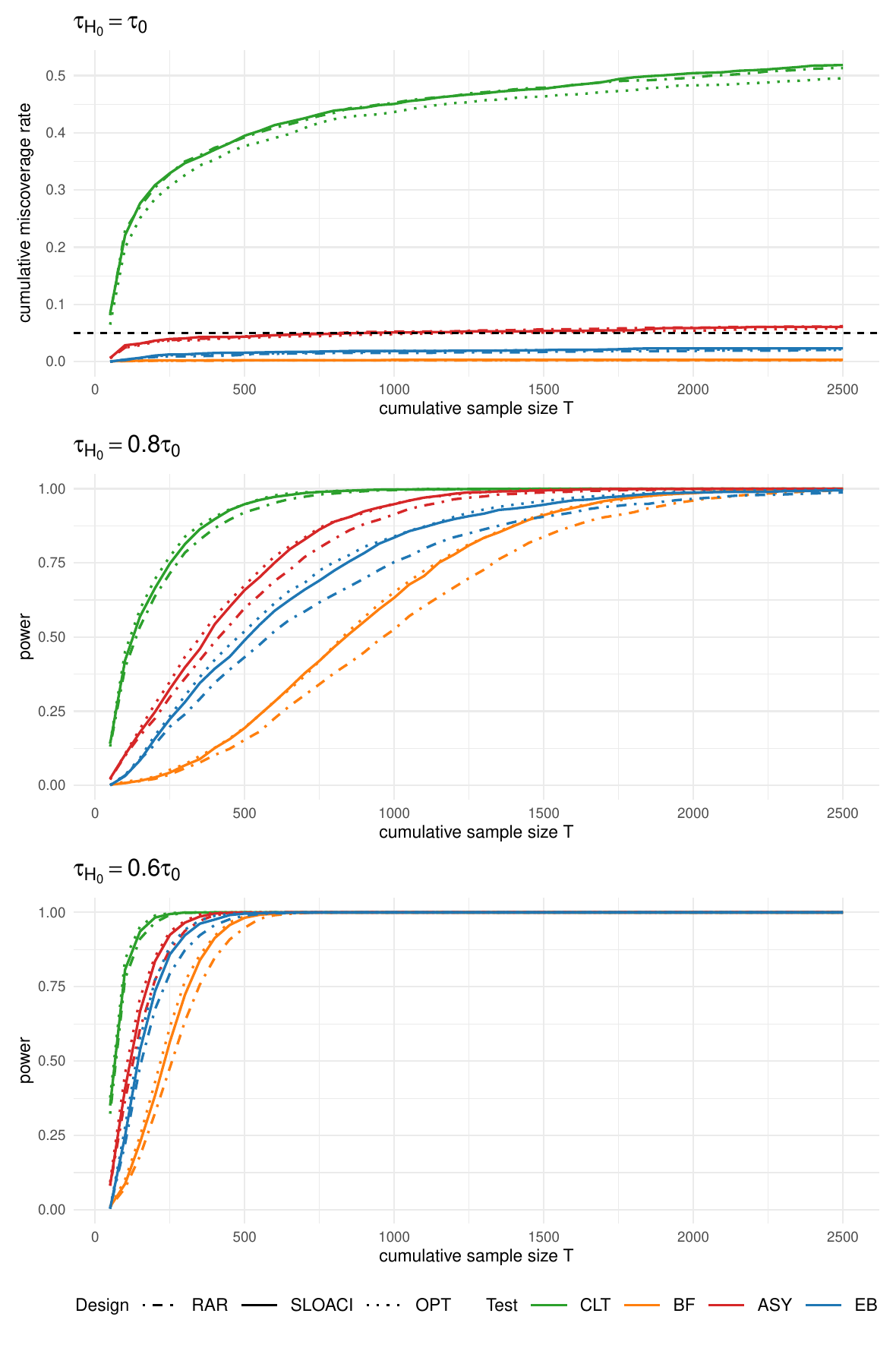}
\end{center}
\vspace{-1em}
\caption{\small Comparison of empirical cumulative miscoverage rates and empirical powers across different designs and tests in Scenario 1.}
\label{fig:test1}
\end{figure}

\begin{figure}[H]
\begin{center}
\includegraphics[width=0.8\linewidth]{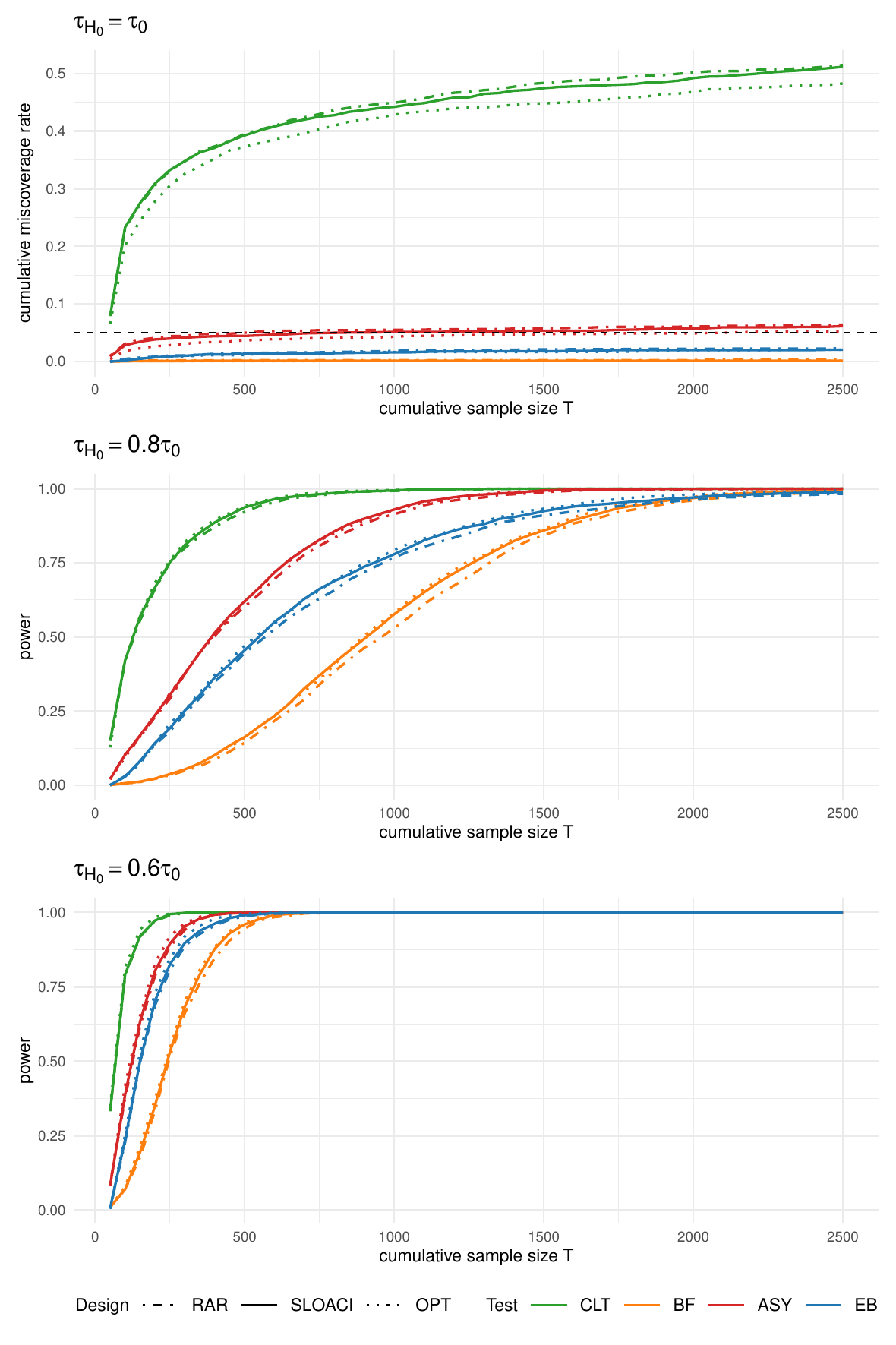}
\end{center}
\vspace{-1em}
\caption{\small Comparison of empirical cumulative miscoverage rates and empirical powers across different designs and tests in Scenario 2.}
\label{fig:test2}
\end{figure}

\begin{figure}[H]
\begin{center}
\includegraphics[width=0.8\linewidth]{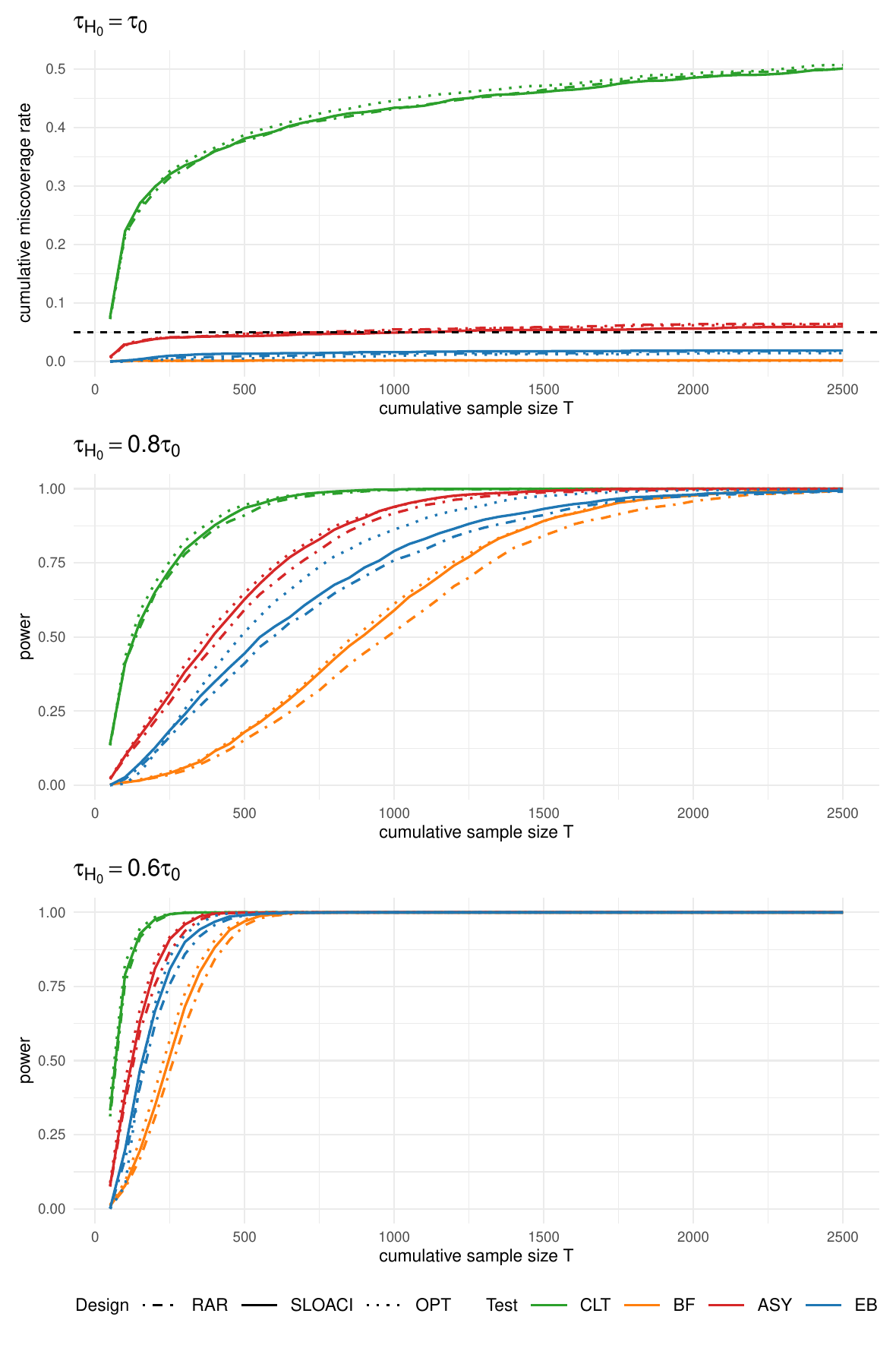}
\end{center}
\vspace{-1em}
\caption{\small Comparison of empirical cumulative miscoverage rates and empirical powers across different designs and tests in Scenario 3.}
\label{fig:test3}
\end{figure}

\begin{figure}[H]
\begin{center}
\includegraphics[width=0.8\linewidth]{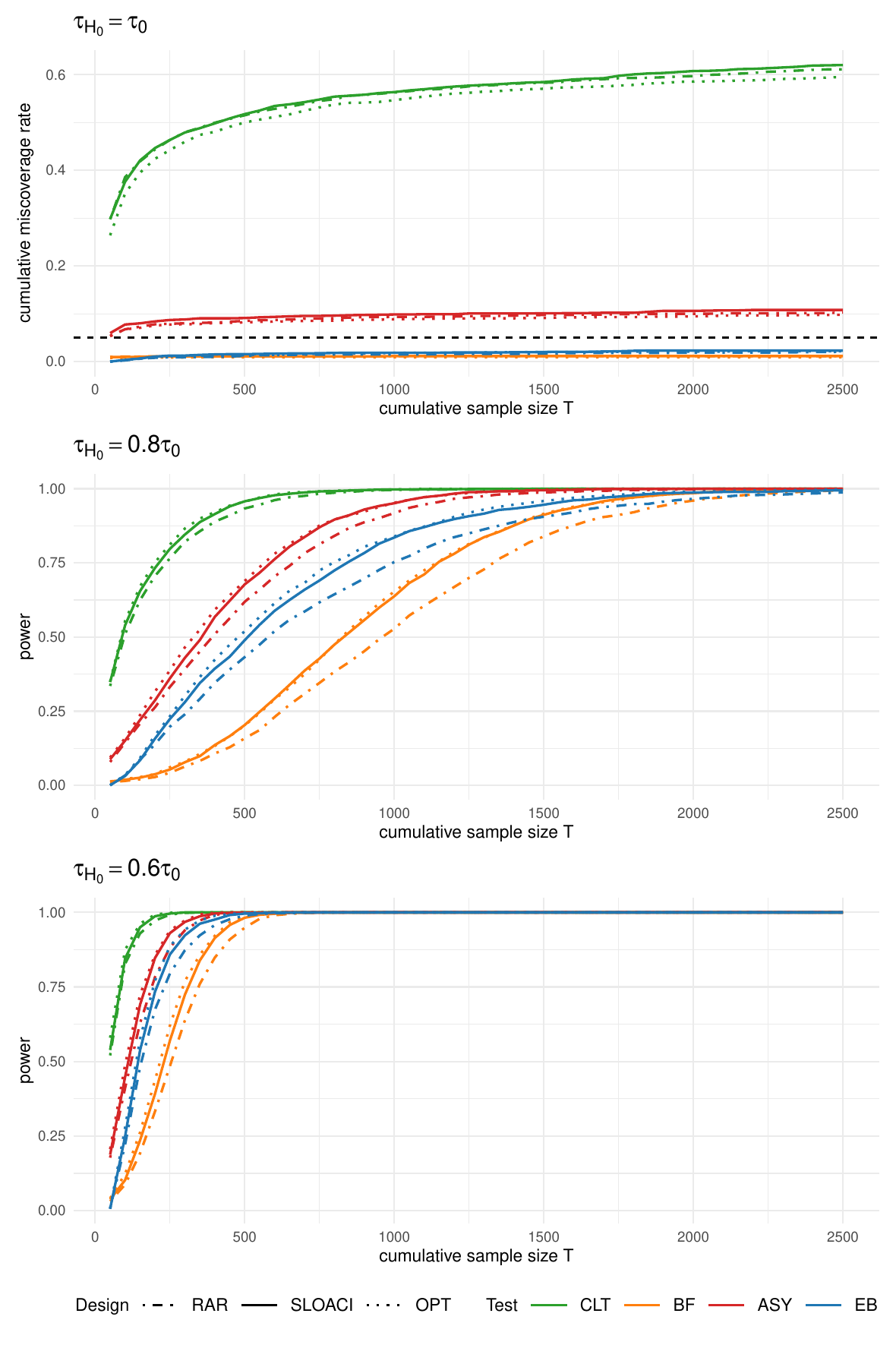}
\end{center}
\vspace{-1em}
\caption{\small Comparison of empirical cumulative miscoverage rates and empirical powers across different designs and tests in Scenario 1 with initial peeking time $t_0=30$.}
\label{fig:test4}
\end{figure}

\renewcommand{\arraystretch}{0.75}
\begin{table}[H]
\footnotesize
\centering
\caption{\small The means and standard errors (in parentheses) of stopping times.}
\label{tab:stopping_times}
\vspace{1em}
\begin{tabular}{llcccccc}
\hline
& & \multicolumn{2}{c}{Scenario 1} & \multicolumn{2}{c}{Scenario 2} & \multicolumn{2}{c}{Scenario 3}\\
Test & Design & $\tau_{H_0}=0.8\tau_0$ & $\tau_{H_0}=0.6\tau_0$ & $\tau_{H_0}=0.8\tau_0$ & $\tau_{H_0}=0.6\tau_0$ & $\tau_{H_0}=0.8\tau_0$ & $\tau_{H_0}=0.6\tau_0$ \\
\hline
\multicolumn{1}{l}{CLT} & \multicolumn{1}{l}{RAR} & 203.7  & 81.7  & 196.5  & 80.3  & 204.1  & 82.1  \\
          &       & (4.2) & (1.0)   & (4.2) & (1.0)   & (4.3) & (1.0) \\
          & \multicolumn{1}{l}{SLOACI} & 184.1  & 77.4  & 189.4  & 79.6  & 193.5  & 79.2  \\
          &       & (3.7) & (0.9) & (4.0)   & (1.0)   & (3.9) & (0.9) \\
          & \multicolumn{1}{l}{OPT} & 175.0  & 72.7  & 185.0  & 75.8  & 182.3  & 74.7  \\
          &       & (3.6) & (0.8) & (3.8) & (0.9) & (3.7) & (0.8) \\
\hline
    \multicolumn{1}{l}{BF} & \multicolumn{1}{l}{RAR} & 1011.3  & 268.5  & 1010.2  & 266.9  & 1018.9  & 273.1  \\
          &       & (11.2) & (2.8) & (11.0)  & (2.9) & (11.1) & (2.8) \\
          & \multicolumn{1}{l}{SLOACI} & 885.3  & 241.9  & 966.1  & 257.7  & 932.9  & 254.9  \\
          &       & (9.6) & (2.5) & (10.6) & (2.7) & (10.1) & (2.6) \\
          & \multicolumn{1}{l}{OPT} & 878.6  & 230.7  & 957.9  & 251.7  & 918.9  & 241.6  \\
          &       & (9.7) & (2.4) & (10.3) & (2.7) & (10.1) & (2.6) \\
\hline
    \multicolumn{1}{l}{ASY} & \multicolumn{1}{l}{RAR} & 483.6  & 146.4  & 476.8  & 142.5  & 490.1  & 147.8  \\
          &       & (7.7) & (1.9) & (7.7) & (1.9) & (7.7) & (1.9) \\
          & \multicolumn{1}{l}{SLOACI} & 430.2  & 133.0  & 458.0  & 139.4  & 453.7  & 138.0  \\
          &       & (6.6) & (1.7) & (7.1) & (1.8) & (7.0)   & (1.7) \\
          & \multicolumn{1}{l}{OPT} & 416.9  & 124.3  & 461.9  & 133.5  & 439.1  & 130.2  \\
          &       & (6.6) & (1.6) & (7.2) & (1.7) & (7.0)   & (1.7) \\
\hline
    \multicolumn{1}{l}{EB} & \multicolumn{1}{l}{RAR} & 718.8  & 180.9  & 710.6  & 178.1  & 720.1  & 193.0  \\
          &       & (12.0)  & (2.3) & (12.0)  & (2.3) & (11.2) & (2.3) \\
          & \multicolumn{1}{l}{SLOACI} & 615.6  & 164.0  & 679.4  & 174.0  & 676.1  & 179.8  \\
          &       & (10.0)  & (1.9) & (11.3) & (2.2) & (10.7) & (2.1) \\
          & \multicolumn{1}{l}{OPT} & 587.5  & 154.7  & 662.6  & 165.6  & 581.2  & 179.6  \\
          &       & (9.6) & (1.8) & (10.8) & (2.0)   & (8.4) & (1.6) \\
\hline
\end{tabular}
\end{table}

\textit{Second}, when $\tau_{H_0}\neq\tau_0$, all sequential testing procedures achieve 100\% empirical power when the cumulative sample size is sufficiently large, and for the same method, the cumulative sample size required decreases as the difference between $\tau_{H_0}$ and $\tau_0$ increases. Among the three valid tests, ASY attains the highest empirical power and thus yields the narrowest confidence sequences, followed by EB, while BF yields the lowest detection ability. For different designs under a given test, SLOACI is very close to OPT and achieves higher power than RAR, which arises from the variance-adaptivity of our constructed tests, together with the smaller semiparametric efficiency bound of SLOACI compared to RAR. The improvement diminishes as the conditional correlations between the primary outcomes and the surrogates decrease, as shown by the comparison between Figures~\ref{fig:test1} and \ref{fig:test2}. 

\textit{Third}, the average stopping time results confirm the differences in detection ability across designs and tests. Accounting for standard errors, the proposed SLOACI significantly reduces the average stopping time compared to RAR across all tests, especially in Scenarios 1 and 3. For example, in Scenario 1 when $\tau_{H_0}=0.6\tau_0$, SLOACI reduces the average stopping time of BF from 268.5 to 241.9, ASY from 146.4 to 133.0, and EB from 180.9 to 164.0, relative to RAR. It is interesting to see that for OPT-EB, when $\tau_{H_0}=0.6\tau_0$ (an easier case, reflecting detection ability for small $T$), Scenario 1 yields a shorter average stopping time than Scenario 3, which is characterized by a small $\underline{\pi}$, but the opposite holds when $\tau_{H_0}=0.8\tau_0$ (a more challenging case). This distinction indicates that EB may suffer from reduced power in small samples, as discussed in Section~\ref{subsec:tradeoff}. Figures~\ref{fig:boxplot1}--\ref{fig:boxplot3} below present boxplots of stopping times as a complement. The above observations validate Theorems~\ref{thm:asy_cs} and \ref{thm:finite_cs}, and effectively illustrate the practical trade-offs among three sequential testing procedures as discussed in Section~\ref{subsec:tradeoff}.

\begin{figure}[H]
\begin{center}
\includegraphics[width=0.9\linewidth]{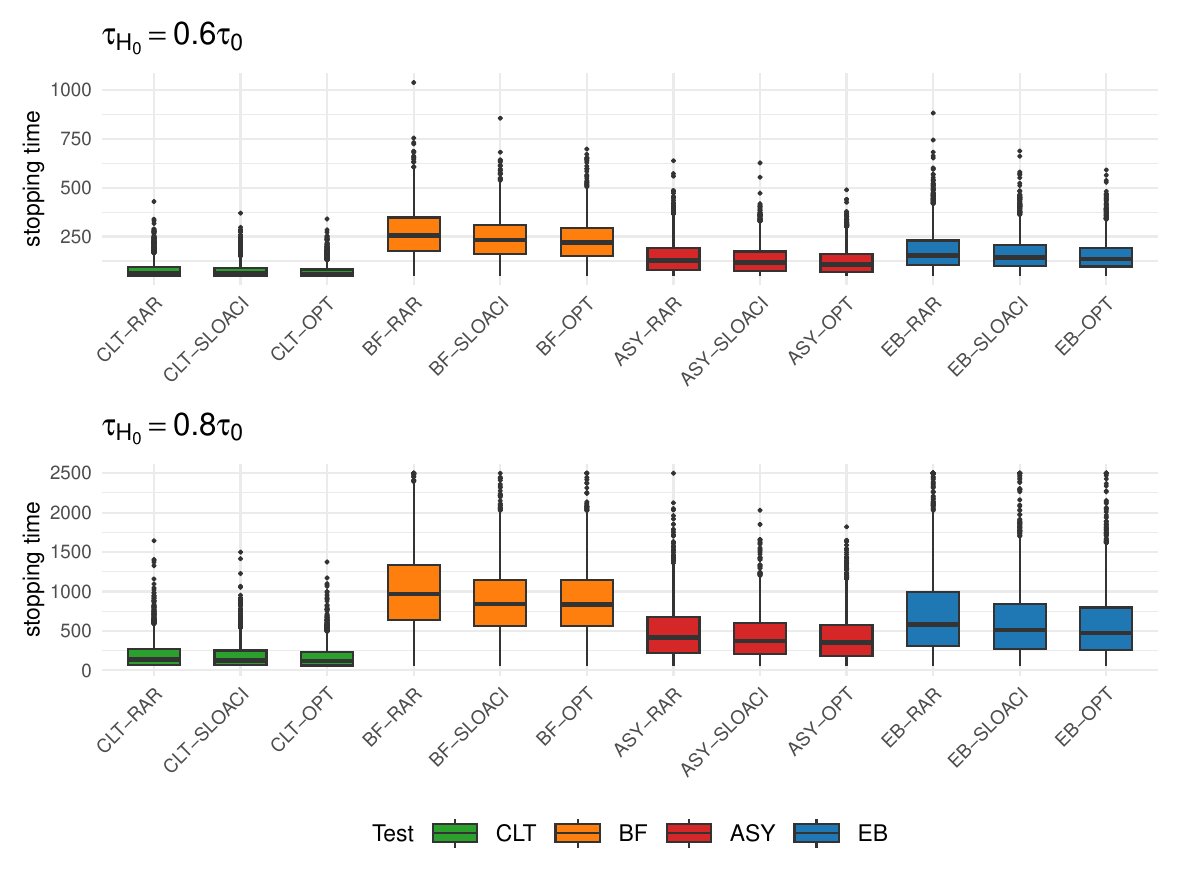}
\end{center}
\vspace{-3em}
\caption{\small Boxplots of stopping times for different designs and tests in Scenario 1.}
\label{fig:boxplot1}
\end{figure}

\begin{figure}[H]
\begin{center}
\includegraphics[width=0.9\linewidth]{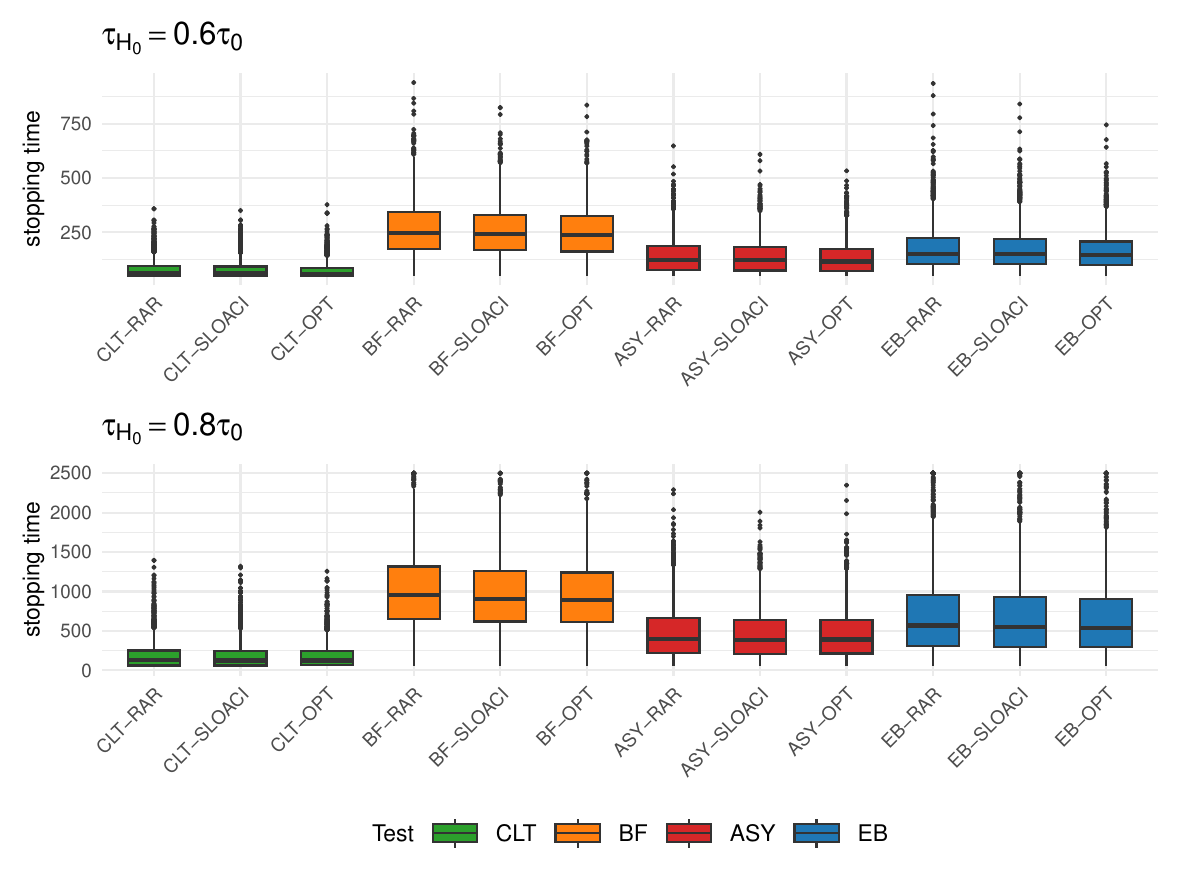}
\end{center}
\vspace{-3em}
\caption{\small Boxplots of stopping times for different designs and tests in Scenario 2.}
\label{fig:boxplot2}
\end{figure}

\begin{figure}[H]
\begin{center}
\includegraphics[width=0.9\linewidth]{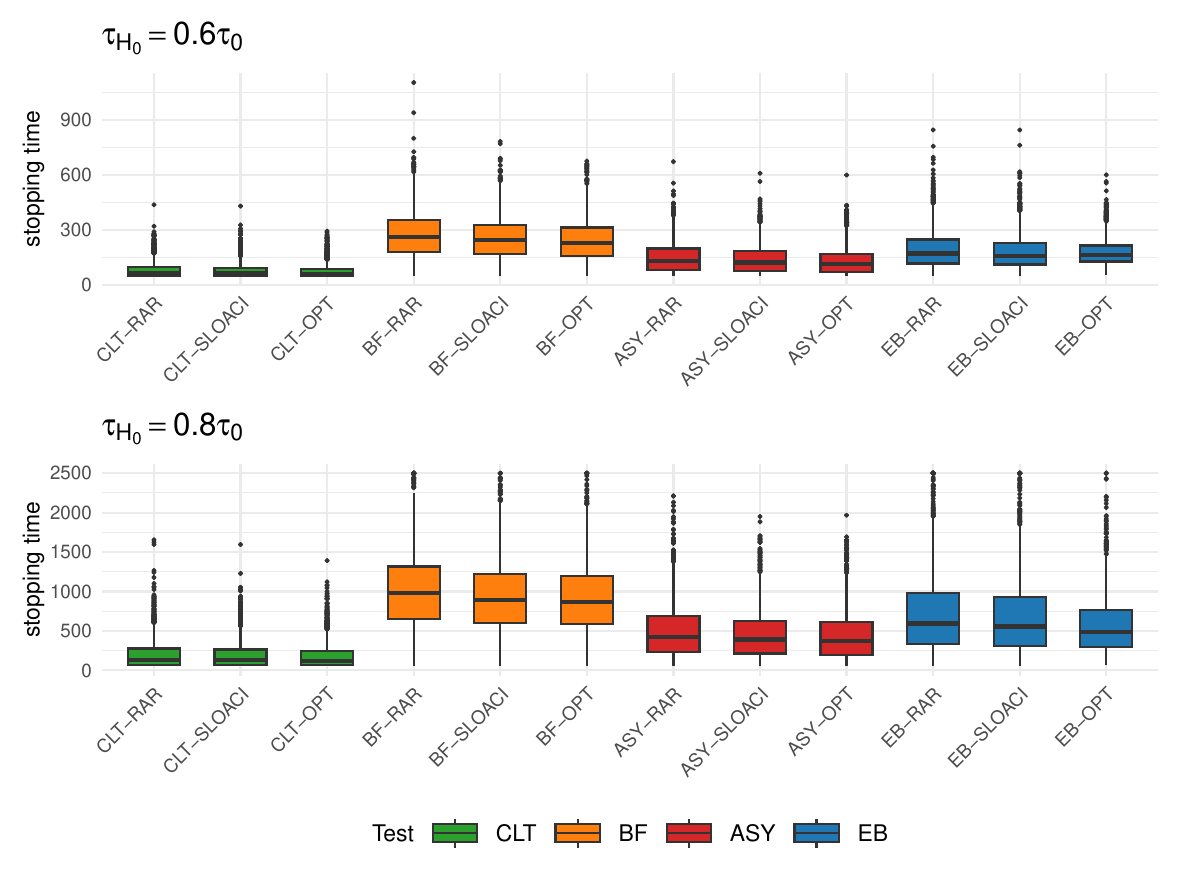}
\end{center}
\vspace{-3em}
\caption{\small Boxplots of stopping times for different designs and tests in Scenario 3.}
\label{fig:boxplot3}
\end{figure}


\section{Additional empirical results}
\label{supsec:empirical}

This section provides additional results supporting Section \ref{sec:case}. We present a specific replication from the simulations with $b=400$ to illustrate of how to conduct the sequential testing in a real-world setting. Figure~\ref{fig:case_width} displays the confidence sequences for different designs and tests along with their corresponding interval widths, revealing several notable trends.
\textit{First}, as the cumulative sample size $T$ increases, the interval widths of all designs and tests gradually decrease, with the proposed ASY under SLOACI achieving the smallest width. 
\textit{Second}, the confidence sequences for the three tests under SLOACI consistently contain $\tau_0=0.229$ at each stage and eventually lie above $\tau_{H_0}=0$ when $T$ is large. This result suggests that assignment to a small class significantly improves the third-grade scores, consistent with evidence from the existing literature (e.g., \textcolor{blue}{Athey et al.}, \textcolor{blue}{2025}). 
\textit{Third}, the stopping times of ASY under SLOACI, RAR, and RCT are 934, 1229, and 1708, respectively. Hence, compared with RCT which requires 5 years of randomized experiment to detect the treatment effect, RAR reduces the required duration to 4 years, while SLOACI further shortens it to 3 years. 

\begin{figure}[H]
\begin{center}
\includegraphics[width=1\linewidth]{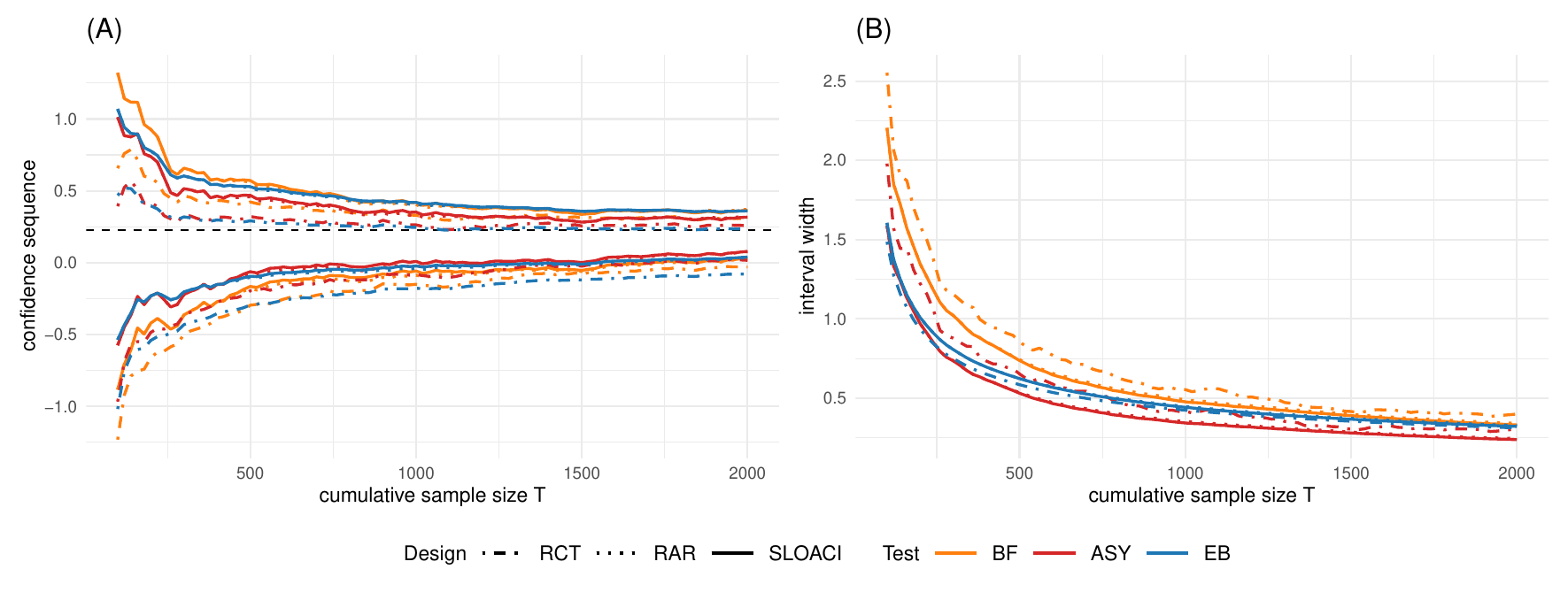}
\end{center}
\vspace{-1em}
\caption{\small Comparison of confidence sequences in panel (A) and their corresponding interval widths in panel (B) across different designs and tests for the STAR data. The black dashed line in panel (A) represents $\tau_0=0.229$.}
\label{fig:case_width}
\end{figure}

{\small 
\spacingset{1}
\section*{References}
\begin{description}

    \item 
    Athey, S., R. Chetty, and G. Imbens (2025). The experimental selection correction estimator: Using experiments to remove biases in observational estimates. \emph{arXiv preprint arXiv:2006.09676v2}.

    \item 
    Balsubramani, A. and A. Ramdas (2016). Sequential nonparametric testing with the law of the iterated logarithm. In \emph{Proceedings of the Thirty-Second Conference on Uncertainty in Artificial Intelligence}, pp. 42--51.

    \item 
    Cai, Y. and A. Rafi (2024). On the performance of the Neyman allocation with small pilots. \emph{Journal of Econometrics} \emph{242}(1), 105793.

    \item
    Chernozhukov, V., D. Chetverikov, M. Demirer, E. Duflo, C. Hansen, W. Newey, and J. Robins (2018). Double/debiased machine learning for treatment and structural parameters. \emph{The Econometrics Journal} \emph{21}(1), C1--C68.
    
    \item
    Cook, T., A. Mishler, and A. Ramdas (2024). Semiparametric efficient inference in adaptive experiments. In \emph{Causal Learning and Reasoning}, pp. 1033--1064. 
    
    \item 
    Dai, J., P. Gradu, and C. Harshaw (2023). Clip-OGD: An experimental design for adaptive Neyman allocation in sequential experiments. In \emph{Advances in Neural Information Processing Systems}, Volume 36, pp. 32235–32269. 
    
    \item 
    Durrett, R. (2019). \emph{Probability: Theory and Examples} (Fifth ed.). Cambridge University Press.

    \item 
    Fan, J. and I. Gijbels (1996). \emph{Local Polynomial Modelling and Its Applications: Monographs on Statistics and Applied Probability 66}. CRC Press.

    \item 
    Gy{\"o}rfi, L., M. Kohler, A. Krzyzak, and H. Walk (2006). \emph{A Distribution-Free Theory of Nonparametric Regression}. Springer Science \& Business Media.

    \item 
    Hahn, J., K. Hirano, and D. Karlan (2011). Adaptive experimental design using the propensity score. \emph{Journal of Business \& Economic Statistics} \emph{29}(1), 96–108.
    
    \item 
    Hall, P. and C. C. Heyde (1980). \emph{Martingale Limit Theory and Its Application}. Academic Press.

    \item 
    Imbens, G. W. and D. B. Rubin (2015). \emph{Causal Inference in Statistics, Social, and Biomedical Sciences}. Cambridge University Press.

    \item
    Kato, M., T. Ishihara, J. Honda, and Y. Narita (2025). Efficient adaptive experimental design for average treatment effect estimation. \emph{arXiv preprint arXiv:2002.05308v7}.

    \item 
    Linton, O. (1995). Second order approximation in the partially linear regression model. \emph{Econometrica} \emph{63}(5), 1079–1112.
    
    \item
    Neopane, O., A. Ramdas, and A. Singh (2025a). Logarithmic Neyman regret for adaptive estimation of the average treatment effect. In \emph{International Conference on Artificial Intelligence and Statistics}. 

    \item
    Neopane, O., A. Ramdas, and A. Singh (2025b). Optimistic algorithms for adaptive estimation of the average treatment effect. In \emph{International Conference on Machine Learning}. 

    \item 
    Perchet, V. and P. Rigollet (2013). The multi-armed bandit problem with covariates. \emph{The Annals of Statistics} \emph{41}(2), 693–721.
    
    \item
    Qian, W. and Y. Yang (2016). Kernel estimation and model combination in a bandit problem with covariates. \emph{Journal of Machine Learning Research} \emph{17}(149), 1--37.

    \item 
    Robinson, P. M. (1988). Root-N-consistent semiparametric regression. \emph{Econometrica} \emph{56}(4), 931–954.

    \item
    Tropp, J. A. (2015). An introduction to matrix concentration inequalities. \emph{Foundations and Trends{\textregistered} in Machine Learning} \emph{8}(1-2), 1--230.

    \item 
    Tsybakov, A. B. (2009). \emph{Introduction to Nonparametric Estimation}. Springer New York
    
    \item 
    Vershynin, R. (2025). \emph{High-Dimensional Probability: An Introduction with Applications in Data Science} (Second ed.). Cambridge University Press.


    \item 
    Waudby-Smith, I., D. Arbour, R. Sinha, E. H. Kennedy, and A. Ramdas (2024a). Time uniform central limit theory and asymptotic confidence sequences. \emph{The Annals of Statistics} \emph{52}(6), 2613–2640.
    
    \item
    Waudby-Smith, I., L. Wu, A. Ramdas, N. Karampatziakis, and P. Mineiro (2024b). Anytime-valid off-policy inference for contextual bandits. \emph{ACM/IMS Journal of Data Science} \emph{1}(3), 1--42.

    \item 
    Yang, Y. and D. Zhu (2002). Randomized allocation with nonparametric estimation for a multi-armed bandit problem with covariates. \emph{The Annals of Statistics} \emph{30}(1), 100–121.
    
    \item 
    Zhang, H. and S. X. Chen (2020). Concentration inequalities for statistical inference. \emph{arXiv preprint arXiv:2011.02258}.
    
\end{description}
}

\end{document}